%% file: lipics-v2021-sample-article.tex
\pdfoutput=1

\documentclass[a4paper,UKenglish,cleveref, autoref, thm-restate]{lipics-v2021}

\nolinenumbers

\title{From Local to Global Determinacy\\in Concurrent Graph Games}

\titlerunning{From Local to Global Determinacy}

\author{Benjamin Bordais, Patricia Bouyer and Stéphane Le Roux}{Université Paris-Saclay, ENS Paris-Saclay CNRS, LMF,
91190 Gif-sur-Yvette,France}{}{}{}

\authorrunning{B. Bordais, P. Bouyer and S. Le Roux}

\Copyright{Benjamin Bordais, Patricia Bouyer and Stéphane Le Roux}



\category{} 

\relatedversion{} 
\input{Tex/macros}

\begin{document}

\maketitle

\begin{abstract}
	In general, finite concurrent two-player reachability games are only determined in a weak sense: the supremum probability to win can be approached via stochastic strategies, but cannot be realized.
	
	We introduce a class of concurrent games that are determined in a much stronger sense, and in a way, it is the larger class with this property. To this end, we introduce the notion of \emph{local interaction} at a state of a graph game: it is a \emph{game form} whose outcomes (i.e. a table whose entries) are the next states, which depend on the concurrent actions of the players. By definition, a game form is \emph{determined} iff it always yields games that are determined via deterministic strategies when used as a local interaction in a Nature-free, one-shot reachability game. We show that if all the local interactions of a graph game with Borel objective are determined game forms, the game itself is determined: if Nature does not play, one player has a winning strategy; if Nature plays, both players have deterministic strategies that maximize the probability to win. This constitutes a clear-cut separation: either a game form behaves poorly already when used alone with basic objectives, or it behaves well even when used together with other well-behaved game forms and complex objectives. 
	
	Existing results for positional and finite-memory determinacy in turn-based games are extended this way to concurrent games with determined local interactions (CG-DLI). 
\end{abstract}

	\section{Introduction}
	\label{sec:introduction}
	\input{Tex/Introduction}
	
	\section{Preliminaries}
	\label{sec:preliminaries}
	\input{Tex/Preliminaries}

	\section{Game Forms and Win/Lose Games}
	\label{sec:gameForm}
	\input{Tex/GameForm}

	\section{Colored Arenas, Games, and Strategies}
	\label{sec:arenas}
	\input{Tex/Arenas}

	\section{Sequentialization of Games and Strategies, Parallelization of Strategies}
	\label{sec:SeqPar}
	\input{Tex/SeqPar}

	\section{Applications}
	\label{sec:applications}
	\input{Tex/Applications}

	
	\bibliography{lipics-v2021-sample-article}
	
	\appendix
	\input{Tex/Appendix}

	
\end{document}

%% file: Tex/macros.tex
\usepackage{stmaryrd}
\usepackage{graphicx,tikz,enumerate}
\newtheoremstyle{slant}
{\topsep}   
{\topsep}   
{\slshape}  
{}       
{\bfseries} 
{}         
{} 
{
}          

\newcommand{\outCome}[2]{\Gamma^{#1}_{#2}}

\newcommand{\subVal}{\ensuremath{\mathcal{V}}}

\newcommand{\tr}[1]{\mathsf{tr}_{#1}}

\newcommand{\formNF}{\mathcal{F}}
\newcommand{\allDet}{\mathsf{Det}_{\mathsf{GF}}}
\newcommand{\gameNF}{\mathcal{G}}
\newcommand{\outComeNF}{\mathsf{O}}
\newcommand{\outCNF}{\varrho}

\newcommand{\OutStrat}[1]{\ensuremath{\Pi_{#1}}}

\newcommand{\A}{\mathsf{A} }
\newcommand{\B}{\mathsf{B} }
\newcommand{\reachStrat}[1]{\mathcal{W}_{#1} }
\newcommand{\setA}{A}
\newcommand{\setB}{B}
\newcommand{\wa}{\mathsf{w_\A} }
\newcommand{\wb}{\mathsf{w_\B} }

\newcommand{\Sq}[2]{\distribSet_{#1}^{#2}}

\newcommand{\Seq}{\mathsf{Seq}}
\newcommand{\Par}{\mathsf{Par}}

\newcommand{\colSet}{\mathsf{K} }
\newcommand{\colSetSeq}{\mathsf{K}_\Aconc }
\newcommand{\rec}{\mathcal{R}}

\newcommand{\s}{\ensuremath{\mathsf{s}}}
\newcommand{\St}{\ensuremath{\mathsf{S}}}
\newcommand{\colFunc}{\mathsf{col} }
\newcommand{\Dist}{\mathsf{Dist} }
\newcommand{\Supp}{\mathsf{Supp} }
\newcommand{\distribSet}{\mathsf{D} }
\newcommand{\distribFunc}{\mathsf{dist} }
\newcommand{\prob}[3]{\mathbb{P}^{#1}_{#2,#3} }

\newcommand{\val}[2]{\chi^{#1}_{#2} }
\newcommand{\stratMod}[1]{\widetilde{#1} }
\newcommand{\deltaDistrib}{\Delta}
\newcommand{\deltaDirac}{\underline{\delta} }

\newcommand{\Aconc}{\mathcal{C} }

\newcommand{\ATurnConc}{\Seq(\Aconc) }
\newcommand{\Games}[2]{\langle #1,#2 \rangle }
\newcommand{\Acolorless}{\ensuremath{\langle 
		\setA,\setB,Q,q_0,\deltaDirac \rangle}}
\newcommand{\AdetConc}{\ensuremath{\langle 
\setA,\setB,Q,q_0,\deltaDirac,\colSet,\colFunc \rangle}}

\newcommand{\AdetTurn}{\ensuremath{\langle 
		\setA,\setB,Q_\A \cup Q_\B,q_0,\delta,\colSet,\colFunc \rangle}}
\newcommand{\AcoloredConc}{\ensuremath{\langle 
\setA,\setB,Q,q_0,\distribSet,\delta,\distribFunc,\colSet,\colFunc \rangle}}
\newcommand{\AcoloredSeq}{\ensuremath{\langle \setA,\setB,V,q_0,\distribSet_\A 
\uplus \distribSet_\B, \delta_\Aconc, \distribFunc_\Aconc, \colSet_\Aconc, \colFunc_\Aconc \rangle}}

\newcommand{\SetColStrat}[2]{\ensuremath{\mathsf{ColSt}_{#1}^{#2} }}
\newcommand{\SetStrat}[2]{\ensuremath{\mathsf{StaSt}_{#1}^{#2} }}
\newcommand{\SetSt}[2]{\ensuremath{\mathcal{S}^{#1}_{#2}}}

\newcommand{\minit}{m_{\text{\upshape init}}}
\newcommand{\Rech}[4]{\mathsf{Rch}^{#1,#2}_{#3,#4}}
\newcommand{\rech}[4]{\mathsf{rch}^{#1,#2}_{#3,#4}}

\newcommand{\extendFunc}[1]{\overline{#1}}

\newcommand{\projec}[2]{\ensuremath{\proj_{#1,#2}}}
\newcommand{\proj}{\ensuremath{\phi}}

\newcommand{\head}{\ensuremath{\mathsf{lt}}}
\newcommand{\tail}{\ensuremath{\mathsf{tl}}}
\newcommand{\cyl}{\ensuremath{\mathsf{Cyl}}}
\newcommand{\starom}[1]{\ensuremath{#1}^\uparrow}
\newcommand{\inv}[1]{\ensuremath{{#1}^{-1}}}
\newcommand{\Borel}{\mathsf{Borel}}

\newcommand{\Nat}{\ensuremath{\mathbb{N}}}

%% file: Tex/Introduction.tex
	

In this paper we consider games that involve at most two players and that are played on infinite (unless otherwise stated) graphs. Consider the turn-based game in Figure~\ref{fig:turn_based_surely}. It starts in state $q_0$. There, Player $\A$ chooses either the self-loop or the edge to $q_1$; a symbol $x$, called a color, is seen in either case; then the game proceeds to state $q_0$ or $q_1$. In $q_1$ Player $\B$ chooses either the $y$-labeled self-loop or the $x$-labeled edge to $q_0$. This generates an infinite sequence over $\{x,y\}$. In this game and the next three examples, the objective of Player $\A$ is that $y$ occurs at some point, while Player $\B$ wins if $y$ never occurs. Player $\B$ has a winning strategy, which consists in never using the self-loop in $q_1$: however Player $\A$ may play, the generated sequence is $x^\omega$. Thus, the game is said to be determined, in a very strong sense, and many sorts of objectives enjoy similar properties on such turn-based games. More generally, Martin~\cite{martin1975borel,martin1985purely} proved that turn-based games with Borel objective are deterministically determined.
\cite{DBLP:journals/geb/Kukushkin02}
\begin{figure*}[htb]
	\begin{minipage}[t]{0.45\linewidth}
		\centering
		\includegraphics[scale=.9]{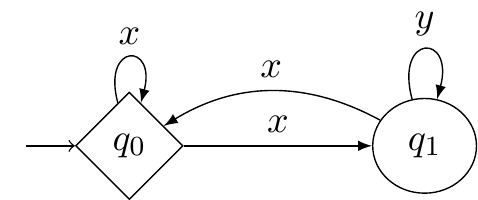}
		\caption{A turn-based game without Nature
			with diamond-shaped nodes for Player $\A$, circle-shaped nodes for Player $\B$ and color labels on edges.
		}
		\label{fig:turn_based_surely}
	\end{minipage} \hfill
	\begin{minipage}[t]{0.45\linewidth}
		\centering
		\includegraphics[scale=1]{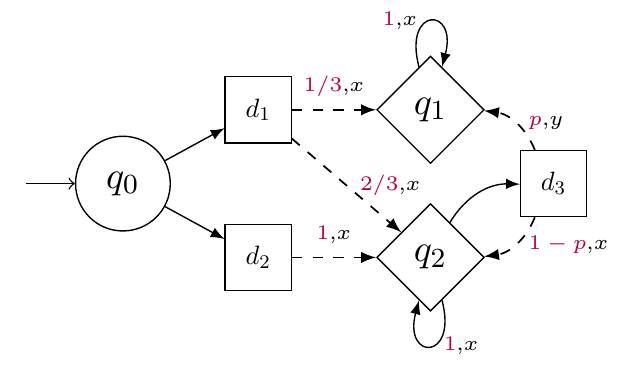}
		\caption{A turn-based game with Nature 
			with probabilities 
			displayed in purple on Nature-to-player
			edges, and colors in black on the same edges for convenience. 
		}
		\label{fig:turn_based_almost_surely}
	\end{minipage}
\end{figure*}

\begin{figure*}[htb]
	\begin{minipage}[t]{0.45\linewidth}
		\centering
		\includegraphics[scale=0.75]{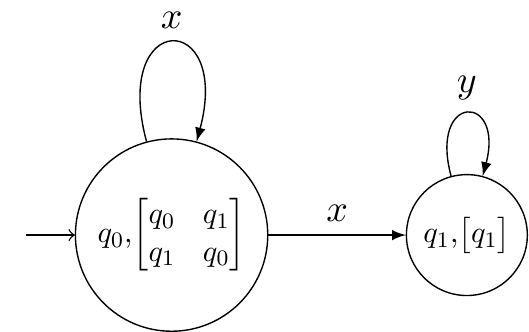}
		\caption{A concurrent game without Nature, 
			with two actions for each player:
			Player $\A$ chooses a row, 
			Player $\B$ chooses a column. 
		}
		\label{fig:repeated-matching-pennies}
	\end{minipage} \hfill
	\begin{minipage}[t]{0.45\linewidth}
		\centering
		\includegraphics[scale=0.75]{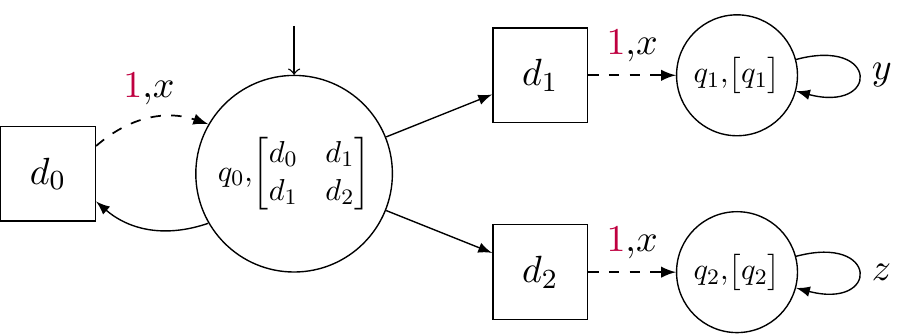}
		\caption{A concurrent game with Nature (albeit deterministic) with probabilities 
			displayed in purple on Nature-to-player edges
			. 
		}
		\label{fig:arbitrarilyClose}
	\end{minipage}
\end{figure*}

Now consider the turn-based game with (stochastic) Nature in
Figure~\ref{fig:turn_based_almost_surely}. In $q_0$ Player
$\B$ moves to Nature state $d_1$ or $d_2$. In $d_1$ Nature
goes to $q_1$ and $q_2$ with probability $\frac{1}{3}$ and
$\frac{2}{3}$, respectively; in $d_2$ with probability $1$ to
$q_2$. In $q_1$ there is only a self-loop, which is a
shorthand for an edge towards a Nature state that goes back to
$q_1$ with probability $1$. In $q_2$ Player $\A$ stays in
$q_2$ or moves to $d_3$. In $d_3$ Nature goes to $q_1$ and
$q_2$ with probabilities $p, 1-p \in {]}0,1{[}$. The edge
$(q_2,q_1)$ is labeled with $y$, and the other edges between
the $q_i$ are labeled with $x$. From $q_0$, Player $\B$ has a
strategy that minimizes the probability to see $y$ (to
$\frac{2}{3}$), namely to play towards $d_1$; and Player $\A$
maximizes this probability (to the same value $\frac{2}{3}$)
by playing $d_3$ when in $q_2$. Note that from $q_2$, this
Player $\A$'s strategy wins almost surely but not
surely. These optimal strategies of the Players are
deterministic.
The game is said to be determined, in a sense that is rather strong but weaker than above without Nature, and several objectives (though fewer than above) enjoy similar properties on turn-based games with Nature. More generally, it was proved~\cite{DBLP:conf/soda/ChatterjeeJH04,DBLP:conf/fossacs/Zielonka04} that turn-based parity games played on finite graphs with stochastic Nature have deterministic optimal strategies.


Consider the game from \cite{DBLP:conf/focs/AlfaroHK98} in
Figure~\ref{fig:repeated-matching-pennies}. The table depicted within state $q_0$ records the concurrent
interaction between the two players at $q_0$: Player $\A$ chooses a
row of the table while Player $\B$ independently chooses a column
of the table; depending on the two choices, the game proceeds either to state $q_0$ again (first row first column, or second row second column)
or to state $q_1$. 
In the two cases $x$ is seen.
In $q_1$ the interaction is trivial, i.e. each player has only one option,
and $y$ is seen. It is
easy to see that Player $\A$ has no deterministic winning strategy, but a stochastic strategy that wins almost surely: in $q_0$, she picks
each row with probability one half.

In the game in Figure~\ref{fig:arbitrarilyClose}, Player $\A$ has no stochastic strategy that wins almost surely, but for all $\epsilon \in {]}0,1{]}$, she has a
stochastic strategy that wins with a probability at least $1 -
\epsilon$: in $q_0$, she chooses the second row with probability
$\epsilon$.	More generally, Martin~\cite{martin1998determinacy} proved that such a weak determinacy holds in games with Borel objective if the local interactions involve finitely many rows and columns.

The above examples and existing results suggest that what prevents the existence of optimal strategies is more the structure of the local interaction rather than the presence of a stochastic Nature. This article substantiates this impression.

%
%

\begin{figure*}[htb]
	\begin{minipage}[t]{0.45\linewidth}
		\centering
		\includegraphics[scale=0.75]{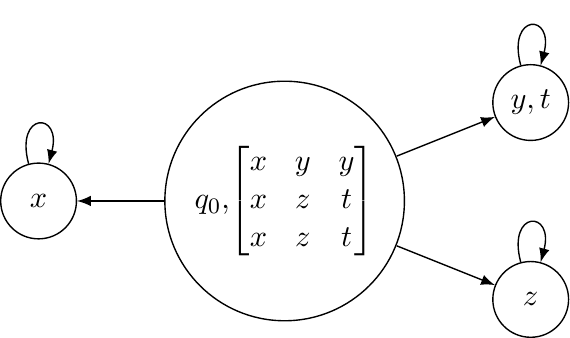}
		\caption{A concurrent game reachability game that is determined for every set of states as target set for either of the players.}
		\label{fig:one_shot_good}
	\end{minipage} \hfill
	\begin{minipage}[t]{0.45\linewidth}
		\centering
		\includegraphics[scale=0.75]{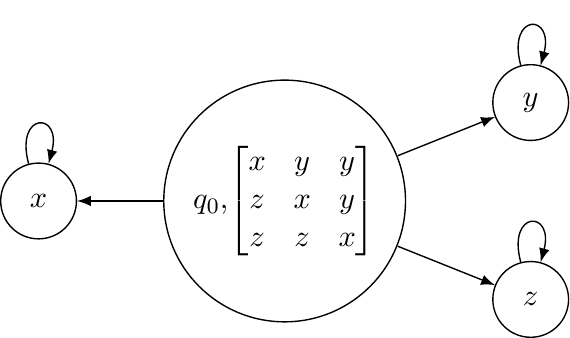}
		\caption{A concurrent reachability game that is not determined if 
			a player tries to reach the set of states $T = \{ x \}$.}
		\label{fig:one_shot_bad}
	\end{minipage}
\end{figure*}

\subparagraph{Our contribution} A game form is a table whose
entries are called outcomes, see
e.g. Figures~\ref{fig:one_shot_good},\ref{fig:one_shot_bad},\ref{fig:examples_game_forms}. By
definition, it is determined if replacing each outcome with
$1$ or $0$ yields a table with a row full of $1$ (Player $\A$
wins) or a column full of $0$ (Player $\B$ wins). It is easy to show that it is determined iff every ``one-shot'' reachability game using it as local interaction is deterministically determined. E.g. consider the one-shot reachability arena in Figure~\ref{fig:one_shot_good}, involving a determined game form. Setting any subset of $\{x,y,z,t\}$ as target for either of the players yields a deterministically determined game. However, the game form in Figure~\ref{fig:one_shot_bad} is non-determined, e.g. by setting $x:=1$ and $y,z:=0$. Equivalenty, setting the target of Player $\A$ to $\{x\}$ yields a game with no winning strategies. 
Thus, the determinacy of a game form amounts to its good behavior when used individually as local interaction in very simple games. We will show that individually well-behaved game forms are collectively well-behaved. More specifically, we extend various determinacy results from turn-based \cite{martin1975borel,DBLP:conf/concur/Bouyer0ORV20,DBLP:conf/soda/ChatterjeeJH04,DBLP:conf/fossacs/Zielonka04,DBLP:conf/soda/GimbertH10} to concurrent games with determined local interactions (CG-DLI). Fix a set $\colSet$ of colors. Each edge of our games is labeled with some color, and the winning objective is expressed as a subset of $\colSet^\omega$.
\begin{enumerate}
	\item In all CG-DLI with Borel (parity) objective, one player has a (positional) winning strategy. 
	
	\item In all CG-DLI with Borel (parity) objective and stochastic Nature, both players have (positional) optimal strategies.
	
	\item Let $\mathcal{M}$ be a memory skeleton (DFA on $\colSet$, explained later).  The following are equivalent. 
	\begin{itemize}
		\item $W$ and $\colSet^\omega \setminus W$ are
		$\mathcal{M}$-monotone and
		$\mathcal{M}$-selective
		(notions recalled
		later).
		
		\item All CG-DLI with finitely many states and actions, and objective $W$ has a finite-memory winning strategy implemented via $\mathcal{M}$.
	\end{itemize}
\end{enumerate}
Moreover in the three statements above, the winning/optimal
strategies can be chosen both deterministic and dependent only
on the history of observed colors, rather than visited
states.

Conversely, let $G$ be any non-determined game form. As hinted at above, one can show that for all Borel objectives $\emptyset \subsetneq W \subsetneq \colSet^\omega$, there is a Nature-free game with one single non-trivial state whose local interaction is $G$, and no deterministic optimal (or winning) strategy, even one that would depend on the history of visited states. A similar resut holds for finite-memory strategies. Hence, these results provide a clear-cut separation: determined game forms are well-behaved basic bricks that collectively build well-behaved-only concurrent games, while non-determined game forms are ill-behaved already when used alone.

A large part of the proofs of the above extensions is factored out by our following theorem: a CG-DLI is (finite-memory, positionnaly, ``plainly'') determined (via deterministic or optimal strategies) if and only if its sequential version is. The sequential version of the game is obtained by letting one player (whichever, but keep the convention) act first at each state of the game, and the opponent act second. Although most of the extensions are straightforward applications of this theorem, the finite-memory case is different: the result in \cite{DBLP:conf/concur/Bouyer0ORV20} requires the objective to satisfy specific properties, and it is rather long to prove that these properties satisfy the assumptions of the theorem. 


\subparagraph{Outline}
Section~\ref{sec:preliminaries} contains notations;
Section~\ref{sec:gameForm} recalls the notion of game form;
Section~\ref{sec:arenas} presents the game-theoretic formalism;
Section~\ref{sec:SeqPar} defines the sequentialization and parallelization, and proves related preservation results;
Section~\ref{sec:applications} presents determinacy extensions


%% file: Tex/Preliminaries.tex

Consider a non-empty set $D$. We denote by $\starom{D} := D^* \cup D^\omega$ the set of finite or infinite sequences in $D$. 
For a sequence $\pi = \pi_0 \pi_1 \ldots \pi_n \in D^*$, we denote by $\head(\pi)$ the last element of the sequence: $\head(\pi) = \pi_n$. 

For a function $\mathsf{f}: E \rightarrow F$ 
$F$
and $F' \subseteq F$, the notation $\mathsf{f}^{-1}[F']$ refers to the preimage $\{ e \in E \mid \mathsf{f}(e) \in F' \}$ of $F'$ by the 
function $\mathsf{f}$. 
Furthermore, a function 
$\mathsf{f}: E \rightarrow F$ can be lifted into a function $\extendFunc{\mathsf{f}}: \starom{E} \rightarrow \starom{F}$ defined by: 
$\extendFunc{\mathsf{f}}(\varepsilon) = \varepsilon$, $\extendFunc{\mathsf{f}}(e) = \mathsf{f}(e)$ for all $e \in E$, and 
$\extendFunc{\mathsf{f}}(\pi \cdot \pi') = \extendFunc{\mathsf{f}}(\pi) \cdot 
\extendFunc{\mathsf{f}}(\pi')$ for all $\pi 
\in E^*$ and $\pi' \in \starom{E}$. For a set $E' \subseteq E$, we define the 
projection function $\projec{E}{E'}: \starom{E} \rightarrow \starom{E'}$ such 
that $\projec{E}{E'}(e) = e$ if $e \in E'$, $\projec{E}{E'}(e) = \epsilon$ 
otherwise and $\projec{E}{E'}(\pi \cdot \pi') = \projec{E}{E'}(\pi) \cdot 
\projec{E}{E'}(\pi')$ for all $\pi \in E^*$ and $\pi' \in \starom{E}$.
For a set $Q$ and a function 
$\mathsf{f}: Q \times Q \rightarrow T$, we denote by $\tr{\mathsf{f}}: Q^+ 
\rightarrow T^* \times Q$ the function that associates to a sequence $\pi \in 
Q^+$, its trace $\tr{\mathsf{f}}(\pi) = 
(\extendFunc{\mathsf{f}}(\pi),\head(\pi))$. For instance, $\tr{\mathsf{f}}(a \cdot b \cdot c) = (f(a,b) \cdot f(b,c),c)$.


Let us now recall the definition of cylinder sets. For a non-empty set $Q$, for all $\pi \in Q^*$, the cylinder set $\cyl(\pi)$ generated by $\pi$ is the set $\cyl(\pi) = \{ \pi \cdot \rho \in Q^\omega \mid \rho \in Q^\omega \}$. We denote by $\cyl_Q$ the set of all cylinder sets on $Q^\omega$. The open sets of $Q^\omega$ are 
the sets 
equal to an arbitrary union of cylinder sets. The set of Borel sets on $Q^\omega$, denoted $\Borel(Q)$, is then equal to the smallest set containing all open sets that is closed under complementation and countable union. 
Recall that, considering two probability measures $\nu,\nu': \Borel(Q) \rightarrow [0,1]$ such that, for all $C \in \cyl_Q$, we have $\nu(C) = \nu'(C)$, we have $\nu = \nu$.


%% file: Tex/GameForm.tex
Informally, game forms (used in \cite{10.2307/1914083,GURVICH197574}) are games without objectives, see Definition~\ref{def:arena_game_nf} and examples in Figure~\ref{fig:examples_game_forms}. They are 
similar to what is sometimes called arena, but 
are presented in normal form, i.e. by ignoring their possible underlying 
graph or tree structure.

\begin{definition}[Game form and win/lose game]
  A \emph{game form} is a tuple
  $\formNF = \langle \St_\A,\St_\B,\outComeNF,\outCNF \rangle$ where
  $\St_\A$ (resp. $\St_\B$) is the non-empty set of strategies
  available to Player $\A$ (resp. $\B$), $\outComeNF$ is a non-empty
  set of possible outcomes, and
  $\outCNF: \St_\A \times \St_\B \rightarrow \outComeNF$ is a function
  that associates an outcome to each pair of strategies. A
  \emph{win/lose game} is a pair
  $\gameNF = \langle \formNF,\subVal \rangle$ where $\formNF$ is a
  game form and $\subVal \subseteq \outComeNF$ is the objective, i.e. a subset of outcomes informally corresponding to the set of winning outcomes for Player
  $\A$ (also called winning set of $\A$) whereas $\outComeNF \setminus \subVal$ is the set of winning outcomes for Player $\B$.
  \label{def:arena_game_nf}
\end{definition}

\begin{figure*}
	\centering
	\includegraphics[scale=1]{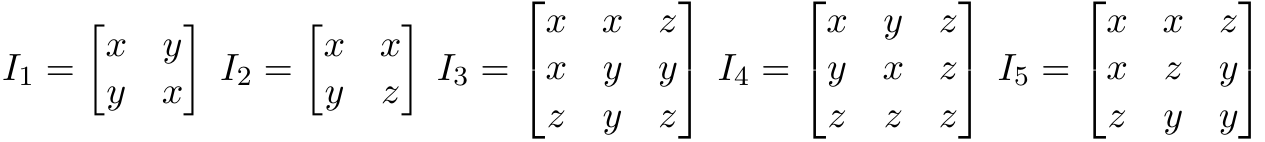}
	\caption{Five 
		game forms: $I_1$ and $I_5$ are not determined, whereas $I_2, I_3$, and 
		$I_4$ are.}
	\label{fig:examples_game_forms}
\end{figure*}

In such a game, a player wins if she obtains an outcome that makes her
win, hence winning for Player $\A$ means reaching an outcome in
$\subVal$, whereas winning for Player $\B$ means reaching an outcome
in $\outComeNF \setminus \subVal$. So, one player wins if and only if
the other player loses, hence the terminology. In the concurrent graph
games that we will consider in Section~\ref{sec:arenas}, the local
interactions that determine what the next visited state will be are
just game forms whose strategies are the available actions of the
players in the "global" game and whose possible outcomes are the
states of the graph. In the context of win/lose games, we can define the notion of winning strategy, 
that is, a strategy for a player that ensures winning regardless of his 
opponent's strategy. The definition of determinacy follows.

\begin{definition}[Winning Strategies and Determinacy]
	Consider a game form $\formNF = \langle \St_\A,\St_\B,\outComeNF,\outCNF 
	\rangle$ and a subset of outcomes $\subVal \subseteq \outComeNF$. In the win/lose game $\gameNF = \langle \formNF,\subVal \rangle$, a
	\emph{winning strategy} $\s_\A \in \St_\A$ (resp. $\s_\B \in \St_\B$) for 
	Player $\A$ (resp. $\B$) is a strategy such that, for all $\s_\B \in \St_\B$ (resp. $\s_\A \in \St_\A$), we have $\outCNF(\s_\A,\s_\B) \in \subVal$ (resp.	$\outComeNF \setminus \subVal$). We write $\reachStrat{\A}(\formNF,\subVal)$ (resp. $\reachStrat{\B}(\formNF,\outComeNF \setminus \subVal)$) the set of winning strategies for Player $\A$ (resp. Player $\B$) with objective $\subVal$.
	The win/lose game $\gameNF$ is \emph{determined} if either of the players has a winning strategy. That is, if $\reachStrat{\A}(\formNF,\subVal) \cup
	\reachStrat{\B}(\formNF,\outComeNF \setminus \subVal) \neq
	\emptyset$. Finally, the game form $\formNF$ is said to be \emph{determined} if, for all $\subVal \subseteq \outComeNF$, the win/lose game $\gameNF = \langle \formNF, \subVal \rangle$ is determined. We denote by $\allDet$ the set of determined game forms.
\end{definition}



For example, consider the game forms represented in Figure~\ref{fig:examples_game_forms}. We argue below that $I_2$, $I_3$, and 
	$I_4$ are determined, while $I_1$ and $I_5$ are not. Consider any subset 
	$\subVal$ of the outcomes and, in $I_2$, $I_3$, and $I_4$, replace each 
	occurence of outcome in $\subVal$ with $\wa$ (indicating
        winning outcomes for Player $\A$) and the others with $\wb$
        (indicating winning outcomes for Player $\B$). There 
	is always a row of $\wa$ or a column of $\wb$, so these game forms are 
	determined. However, rewriting $x$ with $\wa$ and $y$ with $\wb$ in $I_1$ 
	yields the well-known matching-pennies game, which clearly has no winning 
	strategies. Similarly, rewriting $z$ with $\wa$ and $x,y$ with $\wb$ in 
	$I_5$ leads to no row full of $\wa$ and no column full of $\wb$.

        As we shall see, determined game forms are exactly the game
        forms that share enough similarities with ``two-step tree game
        forms'' (roughly, tree game forms are finite
          turn-based games played on a tree, with outcomes at the
          leaves), so that our determinacy transfer may hold. Hence,
        we may ask whether the determined game forms are nothing but
        two-step tree game forms in disguise. Of course, the answer
        depends on what we mean by ``in disguise''. For a natural
        notion of being similar to a (two-step) tree game form, and
        even for a more generous notion, the answer is negative. Thus,
        determined game forms are more than tree game forms, see
        Appendix~\ref{sec:determinedVStree}.

In addition to the toy examples in
Figure~\ref{fig:examples_game_forms}, let us exemplify that determined
game forms arise naturally in computer science. A parity game
(\cite{DBLP:conf/focs/EmersonJ91,mostowski1991games,zielonka1998infinite})
is defined on a priority arena, i.e. a graph where each vertex is
controlled by one player and every edge is labeled with a natural
number less than a fixed bound. The outcome of an infinite run in such
an arena is the maximum of all the
numbers 
that occur infinitely often during the run. If the priorities are seen
not as concrete numbers but as abstract outcomes, the priority arena
can be seen as a game form. 
By a slight generalization of
\cite{DBLP:conf/focs/EmersonJ91,mostowski1991games,zielonka1998infinite}
described, e.g., in \cite[Corollary 3.8]{DBLP:conf/mfcs/Roux18}, it is
moreover a determined game form. So, as we shall see, choosing the
next state following a local interaction given by a parity game will
be a \emph{well-behaved} interaction.

%% file: Tex/Arenas.tex

\subsection{Colored Stochastic Win/Lose Concurrent Graph Games}\label{sec:game-formalism}
Informally, a stochastic concurrent game is played on a graph as
follows: from a given state, both players simultaneously choose an
action, and the next state is set according to a probability distribution that depends on the two actions.
%
We want to consider the ways the two players interact at each state
(which we call the local interactions of the game) as game forms. To
facilitate this, we decouple the concurrent interaction of the players from the stochastic choice of Nature; we therefore add intermediate
states belonging to Nature, and ensure that they do not impact
winning conditions by assigning colors to ordered pairs of player states, thus hiding the 
Nature states that are visited. To sum up, the outcome of an
interaction of the players is a Nature state from which the next
(relevant) state of the game is chosen via a probability distribution.


\begin{definition}[Stochastic concurrent games]
	A colored \emph{stochastic concurrent graph arena} $\Aconc$ is a tuple 
	$\AcoloredConc$ where $\setA$ (resp. $\setB$) is the non-empty set of 
	actions available to Player $\A$ (resp. $\B$), $Q$ is the (non-empty) set 
	of states, $q_0 \in Q$ 
	is the initial state, $\distribSet$ is the set of Nature states, 	
	$\delta: Q \times \setA \times \setB \rightarrow \distribSet$ is the 
	transition function, $\distribFunc: \distribSet \rightarrow \Dist(Q)$ is 
	the distribution function, $\colSet$ 
	is a non-empty set of colors, and $\colFunc: Q \times Q \rightarrow 
	\colSet$ is a coloring function
	. The composition of the transition and distribution functions $\distribFunc \circ \delta: Q \times A \times B \rightarrow \Dist(Q)$ will be denoted $\deltaDistrib$.	A \emph{win/lose concurrent graph game} is a pair $\Games{\Aconc}{W}$ where $W \in \Borel(\colSet)$ is the set of winning sequences of colors (for Player $\A$).
\end{definition}
In the following, the arena $\Aconc$ will always refer to the tuple $\AcoloredConc$ unless otherwise stated. In section~\ref{sec:applications}, we will be able to apply some of our results to a special kind of arenas: the finite ones, defined below.
\begin{definition}[Finite arenas]
	An arena $\Aconc = \AcoloredConc$ is \emph{finite} if the set of deterministic and Nature states $Q \cup \distribSet$ is finite.
	\label{def:finite_arenas}
\end{definition}



We consider two kinds of strategies: strategies that only depends on the sequence of colors seen (and the current state) and that outputs a specific action -- called color strategies -- and strategies that may depend on the whole sequence of states seen and that outputs a distribution over the available actions -- called 
state strategies. In the following, we will show that the concurrent games we consider are determined and that color strategies are sufficient to play optimally, however since the games considered are stochastic, for a strategy to be optimal, it has to achieve the optimal value against all strategy -- that is, state strategies -- of the antagonist player.
\begin{definition}[State and color strategies]
	Let $\Aconc
	$ be an arena. 
	\begin{itemize}
		\item A state strategy, for Player $\A$
		is a function $\s_\A: Q^+ \rightarrow \Dist(\setA)$
		and the set of all such strategies in arena $\Aconc$ for that player is denoted $\SetStrat{\Aconc}{\A}$
		.
		\item A color strategy for Player $\A$ 
		is a function $\s_\A: \colSet^* \times Q \rightarrow \setA$ 
		and the set of all such strategies in arena $\Aconc$ for that player is denoted $\SetColStrat{\Aconc}{\A}$
		. From a color strategy $\s_\A \in \SetColStrat{\Aconc}{\A}$
		, we can extract the color strategy $\stratMod{\s_\A}: Q^+ \rightarrow \Dist(\setA)$ 
		defined by $\stratMod{\s_\A} = \s_\A \circ \tr{\colFunc}$
		.
	\end{itemize}
	The definitions are likewise for Player $\B$.
\end{definition}


{\label{ref:pair_strat_proba}}
Two state strategies then induce a probability of occurrence of finite paths and, following of cylinder sets. This, in turn, induce a probability distribution over all Borel sets. This is formally defined in Appendix~\ref{subsec:strategies_probability_distribution}.



In these games, informally, Player $\A$ tries to maximize the probability to be in the set $W$ whereas Player $\B$ tries to minimize this probability. For both players, this induces the definitions of the value of a strategy and of the game below.
\begin{definition}[Value of strategies and color value of the game]
	Let $\Aconc
	$ be an arena. The corresponding winning set (for Player $\A$) to a Borel set $W \subseteq \colSet^\omega$ is equal to $U_W = \inv{\extendFunc{\colFunc}}[W] \subseteq Q^\omega$. Note that $U_W$ is also a Borel set\footnote{As the preimage of a Borel set by the continuous function $\extendFunc{\colFunc}$.}. Consider now a color strategy $\s_\A \in \SetColStrat{\Aconc}{\A}$
	for Player $\A$
	. Then, the value $\val{\Aconc}{\s_\A}[W]$
	of the strategy $\s_\A$
	is equal to $\val{\Aconc}{\s_\A}[W] = \inf_{\s_\B \in \SetStrat{\Aconc}{\B}} \prob{\Aconc}{\stratMod{\s_\A}}{\s_\B}[U_W]$ 
	. The \emph{color value} $\val{\Aconc}{\A}$ of the game for Player $\A$: 
	$\val{\Aconc}{\A}[W] := \sup_{\s_\A \in \SetColStrat{\Aconc}{\A}} \val{\Aconc}{\s_\A}[W]$. 
	The definitions are likewise for Player $\B$, by reversing the supremum and infimum.
	
	
	A win/lose stochastic concurrent graph game $\langle \Aconc,W \rangle$ is \emph{limit-determined} if we have $\val{\Aconc}{\A}[W] = \val{\Aconc}{\B}[W]$. If in addition there are strategies $\s_\A \in \SetColStrat{\Aconc}{\A}$ and $\s_\B \in \SetColStrat{\Aconc}{\B}$ such that $\val{\Aconc}{\s_\A}[W] = \val{\Aconc}{\A}[W]$ and $\val{\Aconc}{\s_\B}[W] = \val{\Aconc}{\B}[W]$, we say that the game is \emph{determined}. In this case, such strategies are called \emph{optimal} strategies.
	\label{def:determinacy}
\end{definition}


Let us look at what the \emph{local determinacy} of a concurrent game refers 
to, which will yield the definition of locally determined stochastic concurrent 
games. 

\begin{definition}[Local interactions]
	The \emph{local interaction} in a stochastic concurrent graph arena $\Aconc
	$ at state $q \in Q$ is the game form $\formNF_q = \langle 
	\setA,\setB,\delta(q,\cdot,\cdot),\distribSet \rangle$ where the strategies 
	available for Player $\A$ (resp. $\B$) are the actions in $A$ (resp. $B$) 
	and the outcomes are the Nature states reachable from $q$ in the arena 
	$\Aconc$. For a set of game forms $\mathcal{I}$, we say that a concurrent 
	arena $\Aconc = \AcoloredConc$ is \emph{built on} $\mathcal{I}$ if, for all 
	$q \in Q$, we have $\formNF_q \in \mathcal{I}$ (up to a renaming of the 
	outcomes). A stochastic concurrent graph arena/game is \emph{locally 
	determined} if it is built on $\allDet$.
\end{definition}

{\label{ref:turn_based}}
\textbf{Turn-based games} Usually, turn-based games and concurrent games are described in two different formalisms. Indeed, in a turn-based game, a player plays only in the states that she controls, whereas in a concurrent game, in each state both players play an action and subsequently the next (Nature) state is reached. However, turn-based games can be seen as a special case of concurrent games, where at each state, the next (Nature) state is chosen regardless of one of the player's action. We choose the second option (see~\ref{subsec:turn_based}).

Section~\ref{sec:SeqPar} will translate locally determined concurrent
games into turn-based games, then transfer existing determinacy
results on turn-based games back into extension results for the more
general locally determined concurrent games.

\subsection{Colored strategy Implementations}\label{sec:strat-impl}
We recall the notion of memory skeleton that was introduced in \cite{DBLP:conf/concur/Bouyer0ORV20} and we see how it can implement the color strategies that appear in the stochastic concurrent games we consider. For a set of colors $\colSet$ and a set of states $Q$, a memory skeleton on $\colSet$ is a triple $\mathcal{M}  = \langle M,\minit,\mu 
\rangle$, where $M$ is a non-empty set called the memory, $\minit \in 
M$ is the initial state of the memory and $\mu: M \times \colSet 
\rightarrow M$ is the update function. An action map with memory $M$ is a function $\lambda: M \times Q \rightarrow T$ for a non-empty set $T$. Note that $T$ is a set of possible decisions that can be made. Here, $T$ will be instantiated with the set of actions of either of the player. Also, we are only interested in color strategies since we will consider in which case optimal strategies -- that we search among color strategies -- can be chosen finite-memory. In fact, a memory skeleton and an action map implement a color strategy.

\begin{definition}[Implementation of strategies]
	Consider a concurrent colored arena $\Aconc
	$, 
	a player $p \in \{ \A,\B \}$ and the corresponding set of actions $T 
	\in \{ A,B \}$
	. A memory skeleton $\mathcal{M}  = \langle M,\minit,\mu \rangle$ on 
	$\colSet$ and an action map $\lambda: M 
	\times Q \rightarrow T$ implement the color strategy $\s: \colSet^* \times 
	Q \rightarrow T$ that is defined by $\s(\rho,q) = 
	\lambda(\extendFunc{\mu}(\minit,\rho),q) \in T$ for all $(\rho,q) \in 
	\colSet^* \times Q$.
	
	A strategy $\s$ is \emph{finite memory} if there exists a memory 
	skeleton 
	$\mathcal{M}  = \langle M,\minit,\mu \rangle$, with $M$ finite, and an action map $\lambda$ implementing $\s$. If $M$ is reduced to a singleton, $\s$ is 
	\emph{positional}, aka memoryless. The amount of memory used to implement the strategy $\s$ is $|M|$.
\end{definition}

Note that any color strategy $\s: \colSet^* \times Q \rightarrow T$ can be implemented with a memory skeleton and an action map: consider the memory skeleton $\mathcal{M} = \langle \colSet^*,\epsilon,\mu \rangle$ where $\mu: \colSet^* \times \colSet \rightarrow \colSet^*$ is defined by $\mu(\rho,k) = \rho \cdot k$ for all $\rho \in \colSet^*$ and $k \in \colSet$. If the color strategy $\s$ is seen as an action map $\lambda = \s: \colSet^* \times Q \rightarrow T$, then $\mathcal{M}$ and $\lambda$ implement the strategy $\s$.

\begin{definition}[Finite-memory determinacy]
	A game 
	is said to be \emph{finite-memory} (resp. \emph{positionally}) determined if it is determined and optimal strategies can be found among finite-memory (resp. positional) strategies. 
\end{definition}

%% file: Tex/SeqPar.tex
In this section, we describe operators to sequentialize a concurrent 
graph game and its strategies. We also describe an operator to parallelize 
the strategies of the \emph{first} player in the sequential version of a 
concurrent game. These operators on strategies are rather simple and do not 
worsen the value of the game. Then we introduce an operator to parallelize the 
strategies of the \emph{second} player in the sequential version of a 
concurrent game with determined local interactions. This second parallelization does not worsen the value of the game either, but definition and proof 
are not as simple as before: both highly rely on the determinacy of the local 
interactions in the original concurrent game. First, we define the sequentialization of a concurrent game, and we state the theorem we want to prove in this section.


\subsection{Sequential Version of a Concurrent Graph Game}
The sequential version of an arbitrary colored stochastic concurrent graph 
arena consists of a turn-based graph arena where Player $\A$ plays first and 
then Player $\B$ responds.

\begin{definition}[Sequentialization of a concurrent arena and game]
  Consider a concurrent arena $\Aconc = \AcoloredConc$ and an objective $W \in \Borel(\colSet)$. 
  \begin{itemize}
	\item The \emph{sequential version} of $\Aconc$ 
	is the turn-based arena $\ATurnConc = \AcoloredSeq$ where 
	$V = V_\A \uplus V_\B$ with $V_\A = Q$ and $V_\B = Q \times A$, 
	$\distribSet_\A = V_\B$ and $\distribSet_\B =
	\distribSet$. Furthermore, for all $q \in V_\A$, $a \in A$ and
	$b \in B$, we have 
	$\delta_\Aconc(q,a,b) = (q,a) \in V_\B = \distribSet_\A$ and
	$\distribFunc_\Aconc((q,a))[(q,a)] = 1$. In addition,
	for all $d \in \distribSet$, we have 
	$\distribFunc_\Aconc(d) = \distribFunc(d)$ and for all $a' \in A$,
	$b \in B$, and $(q,a) \in V_\B$ we have 
	$\delta_\Aconc((q,a),a',b) = \delta(q,a,b) \in \distribSet =
	\distribSet_\B$. 
	Finally, we have $\colSetSeq = \colSet \cup \{ k_\Aconc \}$ for some fresh color $k_\Aconc \not \in \colSet$ and $\colFunc_\Aconc(q,(q,a)) = k_\Aconc$ if $q \in V_\A$ and $(q,a) \in V_\B$ and $\colFunc_\Aconc((q,a),q') = \colFunc(q,q')$ if $(q,a) \in V_\B$ and $q' \in V_\A$. The function $\colFunc_\Aconc$ is defined arbitrarily on other pairs of states.
	\item The \emph{sequential version} of the concurrent game $\langle 
	\Aconc, W \rangle$ is the turn-based game 
	$\Games{\ATurnConc}{\Seq(W)}$, where $\ATurnConc
	$ is the sequential version of the concurrent arena $\Aconc$ and $\Seq(W) = 
	(\projec{\colSetSeq}{\colSet})^{-1}[W]$ is the preimage of the winning set 
	$W$ by the projection function $\projec{\colSetSeq}{\colSet}: 
	\starom{\colSetSeq} \rightarrow \starom{\colSet}$.
  \end{itemize}
\end{definition}


In the above definition, one can notice that the states in $V_\A$ belong to 
Player $\A$ whereas states in $V_\B$ belong to Player $\B$.

\begin{example}
	Sequentialization of an arena is a rather simple operation that we
	illustrate in Figure~\ref{fig:seq}. Note that the initial concurrent arena has deterministic Nature (all probabilities that appear equal 1), and the sequential version also does. From $q_0$, Player $\A$ selects
	either the first row (top choice in the figure) or the second row
	(bottom choice in the figure), and then Player $\B$ selects one of
	the options, that is, one of the next states offered in the subset
	-- this corresponds to choosing a column in the game form. The fresh
	color $k_\Aconc$ appears after the choice of Player $\A$, while the original
	colors appear after the choice of Player $\B$. 
	
	One can notice here that in the original concurrent game and its
	sequential version, the value of the game for the players are different: in the turn-based game, from $q_0$, Player $\B$ has a strategy to ensure never seeing the color $y$ (which induces a value of 0 for Player $\B$)
	whereas it is not the case in the original game. As we will see along that 
	paper, this is due to the fact that the local interaction at $q_0$ is not
	determined. 
	\begin{figure*}[htb]
		\centering
		\begin{tikzpicture}
		\path (0,0) node
		{\includegraphics[scale=1]{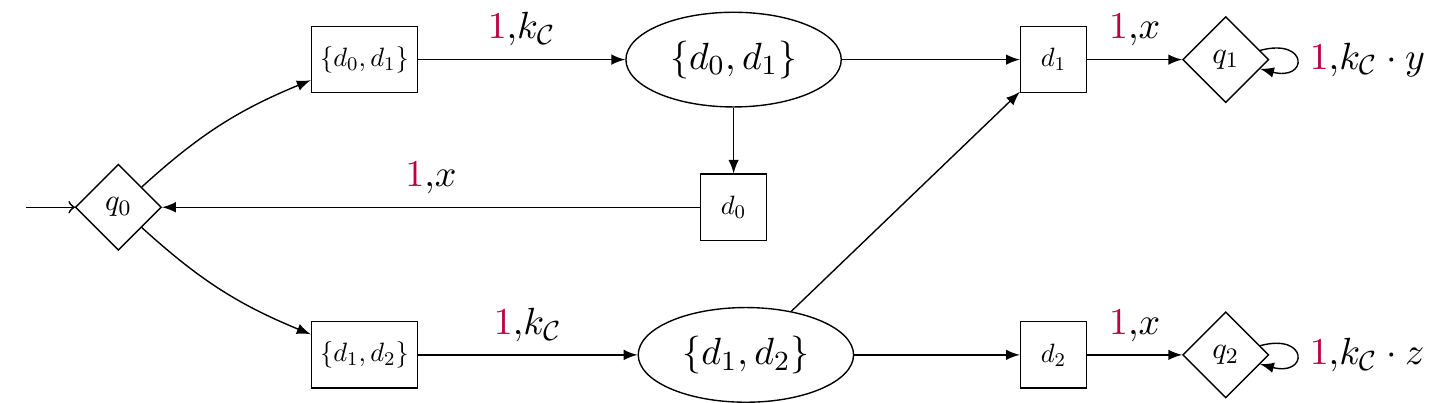}};
		\end{tikzpicture}
		\caption{\normalsize Sequentialization of the concurrent arena from Figure~\ref{fig:arbitrarilyClose}. Diamond-shaped nodes belong to Player $\A$, ellipse-shaped ones belong to Player $\B$ and the rectangle-shaped are Nature states. On the edges, probabilities appear in purple and colors in black. The pairs in $Q \times A$ are represented as the corresponding set of states $\delta(q,a,B) \subseteq \mathcal{P}(\distribSet)$.}
		\label{fig:seq}
	\end{figure*}
\end{example}
{\label{ref:prop_projection_colors_colors_projection}}
We make several remarks. First, paths in a concurrent arena and in its
sequential version relate via a projection (see Proposition~\ref{prop:projection_colors_colors_projection}). Second, if $W$ is Borel, so is $\Seq(W)$ as the continuous preimage 
of a Borel set.
{\label{ref:path_extension}}Also, note that if the probabilities of
finite paths in the concurrent arena are equal to probability of their
preimage in the sequential version, then it follows that the
probability of all Borel sets in the concurrent arena is the
probability of their preimage in the sequentialized version (see Subsection~\ref{subsec:path_extension}).


We can now state the theorem we want to prove in this Section 5.
\begin{theorem}
	\label{thm:main}
	Consider a concurrent game $\langle \Aconc,W \rangle$ and assume that it is locally determined.	
	Then, it is (resp. finite-memory, resp. positionnaly) determined if and only if its sequential version $\langle \Seq(\Aconc),\Seq(W) \rangle$ is (resp. finite-memory, resp. positionnaly) determined.
\end{theorem}

In the following, we will be working on a concurrent graph 
game $\langle \Aconc,W \rangle$ with $\Aconc = \AcoloredConc$ and its 
sequential version $\langle \ATurnConc,\Seq(W) \rangle$ with $\ATurnConc = 
\AcoloredSeq$.

We consider the translation of strategies from the concurrent game to its sequential version. {\label{ref:sequentialization_strat}}
Let us first translate memory skeletons. From a memory skeleton
$\mathcal{M}$ on a set of colors $\colSet$, we obtain its sequential
version $\Seq(\mathcal{M})$ on the set of colors $\colSet_\Aconc$ by
mimicing $\mathcal{M}$ on colors in $\colSet$ and ignoring the color
$k_\Aconc$ (see
Definition~\ref{def:sms}). {\label{ref:projection_memory}} Note that
the state of the memory w.r.t. $\mathcal{M}$ and $\Seq(\mathcal{M})$
relate through the projection of the colors (see Proposition~\ref{prop:projection_memory}). 

Consider now action maps. The sequential version $\Seq(\lambda): M
\times Q \rightarrow A/B$ of an action map $\lambda: M \times Q
\rightarrow A/B$ for either of the players essentially mimics the map
$\lambda$ on relevant states -- for $m \in M$, for Player $\A$:
$\Seq(\lambda)(m,v) = \lambda(m,v)$ if $v \in V_\A$ and for $(v,a) \in
V_\B$, for Player $\B$: $\Seq(\lambda)(m,(v,a)) = \lambda(m,v)$ -- and
plays a dummy action on other states (see Definition~\ref{def:sequen_action_function}). 


Then, if a strategy $\s$ is implemented by memory skeleton
$\mathcal{M}$ and an action map $\lambda$, its sequential version will
be implemented with sequential versions of the memory skeleton
$\Seq(\mathcal{M})$ and of the action map $\Seq(\lambda)$ (see Definition~\ref{def:sequentialization_strat}).



The interest of this definition of sequentialization is stated below:

\begin{observation}
	If a strategy 
	can be implemented with some amount of memory, then so 
	can its sequential version. 
	\label{obs:memory_used_sequen}
\end{observation}

\begin{lemma}[Proof page~\pageref{proof:lemma1}]
	Consider a color strategy $\s_\A \in \SetColStrat{\Aconc}{\A}$ 
	in the concurrent arena $\Aconc$ for Player $\A$ 
	and its sequential version $\Seq(\s_\A) \in 
	\SetColStrat{\ATurnConc}{\A}$ in the turn-based game $\Seq(\Aconc)$
	. Then, for all Borel 
	set $W$: $\val{\Aconc}{\s_\A}[W] \leq \val{\Seq(\Aconc)}{\Seq(\s_\A)}[\Seq(W)]$.
	It follows that $\val{\Aconc}{\A}[W] \leq \val{\Seq(\Aconc)}{\A}[\Seq(W)]$
	.
	\label{lem:sequential_win}
\end{lemma}
It holds similarly for Player $\B$ by reversing the inequalities or by replacing $W$ by $\colSet^\omega \setminus W$.

\subsection{Parallelization of strategies}
This subsection defines the 
parallelization of the strategies: it is the translation of a strategy in the
turn-based game into a strategy in the concurrent game. We first consider the 
parallelization of the memory skeleton, of the action map for Player $\A$ and of strategies. The parallelization of the action map of Player $\B$ will come later, as it is more involved since 
as Player $\B$ plays second, she knows what Player $\A$ has played when taking an action.

{\label{ref:parallelization}}
{\label{ref:parallelization_update}}
\begin{definition}[Parallelization of 
strategies][Justification page~\pageref{proof:remark6}]
	\begin{itemize}
		\item Consider a memory skeleton $\mathcal{M} = \langle M,\minit,\mu \rangle$ on 
		a set of colors $\colSetSeq$. The \emph{parallel version} of that 
		memory skeleton is the memory skeleton $\Par(\mathcal{M}) = \langle 
		M,\minit,\Par(\mu) \rangle$ on $\colSet$ where, for all $m \in M$ and $k \in \colSet$: $\Par(\mu): M \times \colSet 
		\rightarrow M$ is such that $\Par(\mu)(m,k) = \mu(\mu(m,k_\Aconc),k)$
		\item Consider an action map $\lambda: M \times V \rightarrow A$  for a non-empty set $M$ for Player $\A$ in the turn-based arena $\ATurnConc$. Its \emph{parallel version} is the action map $\Par(\lambda): M \times Q \rightarrow A$ where $\Par(\lambda)(m,q) = \lambda(m,q)$.
		\item Consider a strategy $\sigma$ in $\ATurnConc$ implemented by a memory 
		skeleton $\mathcal{M}$ and an action function $\lambda$. Then, the parallel 
		version of $\sigma$ is the strategy $\Par(\sigma)$ implemented by the 
		parallel version of the memory skeleton $\Par(\mathcal{M})$ and of the 
		action map $\Par(\lambda)$. (This definition holds for both players.)
	\end{itemize}
	\label{def:parallelization of memory skeletons}
\end{definition}

The two update functions $\Par(\mu)$ and $\mu$ relate through the extension of sequence of colors, as it is stated in Proposition~\ref{prop:projection_memory_parallel}. As stated below and like sequentialization, parallelization preserves the amount of memory.
\begin{observation}
	The amount of memory used by the parallel version of a memory skeleton $\Par(\mathcal{M})$ is the same as the amount of memory used by the memory skeleton $\mathcal{M}$.
	\label{obs:memory_used_para}
\end{observation}

With this definition, we obtain a lemma that is analogous to 
Lemma~\ref{lem:sequential_win}: the value of the game for Player $\A$ does not worsen with parallelization.
\begin{lemma}[Proof given page~\pageref{proof:lemma2}]
	Consider a strategy $\sigma_\A \in \SetColStrat{\ATurnConc}{\A}$ for Player $\A$ in the sequential version $\ATurnConc$ and its parallel version $\Par(\sigma_\A)$ in the concurrent game $\Aconc$. 
	For all Borel set $W$: $\val{\Aconc}{\Par(\sigma_\A)}[W] \geq \val{\Seq(\Aconc)}{\sigma_\A}[\Seq(W)]$. It follows that $\val{\Aconc}{\A}[W] \geq \val{\Seq(\Aconc)}{\A}[\Seq(W)]$
	\label{lem:parallel_win_A}
\end{lemma}

Now let us proceed to the more involved parallelization of action maps for Player $\B$. As mentioned earlier, this case is trickier
than the previous one since a strategy for Player $\B$ in the
turn-based arena has the information of the action previously taken by
Player $\A$ when choosing the next action. 
Since our goal is to ensure that the value of the game does not worsen,
  we want the parallelization of action maps to ensure that the Nature states reachable in $\Aconc$ with the parallel version of the action maps are also reachable in $\ATurnConc$ with the original action maps: 
that way, every path that can be generated with some probability in the concurrent game could also be generated (up to projection) with the same probability in the turn-based game. 
Let us first define the set of states reachable in two steps from a specific state in $\ATurnConc$ given a strategy for Player $\B$.

\begin{definition}[Reachable states w.r.t. a strategy for Pl. $\B$]
	Let $\mathcal{M} = \langle M,\minit,\mu \rangle$ be 
	a memory skeleton, $\lambda: M \times V \rightarrow B$ be an action map for Player $\B$, $m \in M$ be a state of the memory, and $q \in Q$ 
	be a state of the game. Let $\rech{\mu}{\lambda}{m}{q}: A \rightarrow \distribSet$ be such that $\rech{\mu}{\lambda}{m}{q}(a) = \delta(q,a,\lambda(m',(q,a)))$ for $m' = \mu(m,k_\Aconc)$ for all $a \in A$. Then, let $\Rech{\mu}{\lambda}{m}{q} = \rech{\mu}{\lambda}{m}{q}[A] \subseteq \distribSet$ be the set of Nature states reachable in $\ATurnConc$ from $m$ and $q$ if $\mu$ and 
	$\lambda$ implement a strategy for Player $\B$
	.
\end{definition}

Our goal is to define the parallel action map $\Par(\lambda)$
such that the set of Nature states reachable from a state of the memory $m
\in M$ and of the game $q \in Q$ is included in
$\Rech{\mu}{\lambda}{m}{q}$. To do so, we use the local determinacy
assumption on $\Aconc$. In fact, this assumption gives the following
proposition, which is 
central to our approach:
\begin{proposition}
	Let $\mathcal{M} = \langle M,\minit,\mu \rangle$ be 
	a memory skeleton, $\lambda: M \times V \rightarrow B$ be an action
	map for Player $\B$, $m \in M$ be a state of the memory, and $q \in Q$ 
	be a state of the game. Then, Player $\B$ has a winning strategy in the 
	win/lose game $\langle \formNF_q,Q \setminus \Rech{\mu}{\lambda}{m}{q} 
	\rangle$. That 
	is, $\reachStrat{\B}(\formNF_q,\Rech{\mu}{\lambda}{m}{q}) \neq \emptyset$.
	\label{prop:winning_strat_B_local_det}
\end{proposition}
\begin{proof}
	Consider an action $a \in A$. There exists $b \in B$ such that $\delta(q,a,b) = \rech{\mu}{\lambda}{m}{q}(a) \in \Rech{\mu}{\lambda}{m}{q}$. Since this is true for all 
	$a \in A$, it implies that Player $\A$ has no strategy to
        avoid the set $\Rech{\mu}{\lambda}{m}{q}$ in the game form
        $\formNF_q$, i.e. she has no winning strategies in the
        win/lose game $\langle \formNF_q,Q \setminus
        \Rech{\mu}{\lambda}{m}{q}\rangle$. In other words, we have
        $\reachStrat{\A}(\formNF_q,Q \setminus
        \Rech{\mu}{\lambda}{m}{q}) = \emptyset$. Since the game form
        $\formNF_q$ is determined (by local determinacy of $\Aconc$),
	we have that $\reachStrat{\B}(\formNF_q,\Rech{\mu}{\lambda}{m}{q}) \neq \emptyset$: Player $\B$ has a winning strategy in this game.
\end{proof}

\begin{definition}[Parallelization of action maps for Player $\B$]
Assume that the concurrent arena $\Aconc$ is locally determined. Consider an action map for Player $\B$ in the turn-based arena 
$\ATurnConc$: $\lambda: M \times V \rightarrow B$ on a non-empty set $M$ 
along with an update function $\mu: M \times \colSetSeq \rightarrow M$. Its \emph{parallel version} is the action map $\Par(\lambda): M \times Q 
\rightarrow B$ where, for all $m 
\in M$ and $q \in Q$: $\Par(\lambda)(m,q) = \min_{<_\B} 
\reachStrat{\B}(\formNF_q,\Rech{\mu}{\lambda}{m}{q})$ with 
$\reachStrat{\B}(\formNF_q,\Rech{\mu}{\lambda}{m}{q}) \neq \emptyset$ by 
Proposition~\ref{prop:winning_strat_B_local_det}  
(recall that $<_\B$ is only used to implement an arbitrary choice).	
%
%
\end{definition}
  Note that, contrary to the three other cases of the parallelization for Player $\A$ and the sequentialization of strategies for both players, the parallelization of action map for Player $\B$ depends on an update function.

\begin{figure*}[htb]
	\centering
    \begin{tikzpicture}
      \path (0,0.5) node
      {\includegraphics[scale=1]{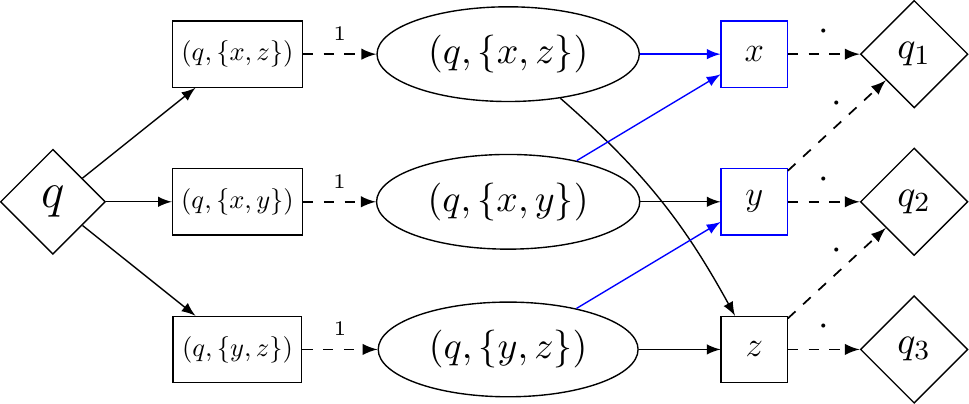}};
      
      \path (6.75,2.7) node {\includegraphics[scale=1]{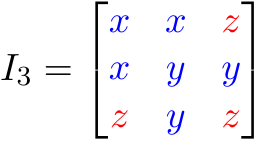}};
      
      \path (8.5,0) node {\includegraphics[scale=1]{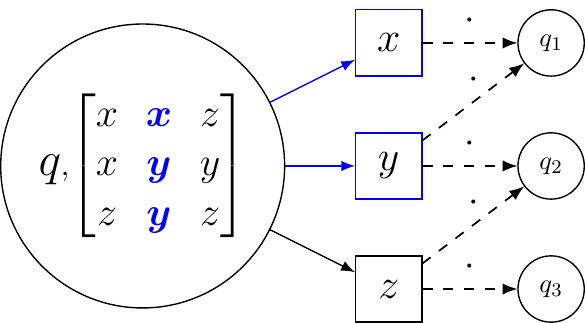}};
    \end{tikzpicture}
	\caption{On the left-hand side, we have a portion of a turn-based graph arena, 
	with all the states reachable in at most two steps from the state $q$. This turn-based arena corresponds to the sequentialization of the portion of the concurrent arena on the right-hand side
    with the local interaction in the state $q$ being $I_3$ from Figure~\ref{fig:examples_game_forms}. 
    Out of $q$, Player $\A$ has three choices (corresponding to the three rows), hence the three outgoing edges; leading to three Nature states from which a specific state belonging to Player $\B$ is reached with probability 1. From each of these three states, Player $\B$ has two choices, leading to two out of the three $x$, $y$ and $z$ Nature states.
    A 
    strategy for Player $\B$ is represented in blue arrows in the turn-based arena, with the Nature states reachable with that strategy represented in blue. It is done similarly in the local interaction $I_3$, with the state that is not reachable, i.e. $z$, in red. Finally, in the concurrent arena, the blue states are the Nature states reachable if Player $\B$ opts for the second column, which is the winning strategy for Player $\B$ in the win/lose game obtained from the game form $I_3$ if she has $\{ x,y \}$ as winning set.}
	\label{fig:transfo_strat_det}
\end{figure*}
\begin{example}
  Let us illustrate that definition on an example on Figure~\ref{fig:transfo_strat_det}. We want to translate an action map $\Par(\lambda)$ into a concurrent arena from its
  sequential version while ensuring that every Nature state reachable with
  that new action map is also reachable in the sequential version
  with the original action map $\lambda$. Consider the example of
  the strategy depicted in the left-hand side of
  Figure~\ref{fig:transfo_strat_det} from the state $q$ and an
  arbitrary state of the memory $m$ omitted on the figure. For each
  possible choice of Player $\A$ (which corresponds to the rows of the local 
  interaction $I_3$), Player $\B$ reacts with his strategy and either $x$ or 
  $y$ is reached. Specifically, we have $\rech{\lambda}{\mu}{m}{q}(a_1) = x$, 
  $\rech{\lambda}{\mu}{m}{q}(a_2) = x$, and $\rech{\lambda}{\mu}{m}{q}(a_3) = 
  y$, where $a_i$ represents the action for player $\A$ for the $i$-th row 
  (similarly, $b_i$ represents the action for Player $\B$ for the $i$-th 
  column). Then, we must 
  define the action for Player $\B$ to play in the concurrent game in state 
  $q$, that is $\Par(\lambda)(m,q) \in B$, so that only the states $x$ and $y$ 
  can be reached.  To choose $\Par(\lambda)(m,q)$ we consider the local 
  interaction $\formNF_q = I_3$. We know that for each action of Player $\A$, 
  there is one for Player $\B$ to reach the set $\{ x,y \}$ (it is given by the 
  strategy depicted in the turn-based
  arena). It follows that Player $\A$ has no winning strategy in the
  win/lose game $I_3$ with $\{ x,y \}$ as winning set for Player
  $\B$. Since the local interaction $I_3$ is determined, Player $\B$
  has a winning strategy in that win/lose game which is a strategy
  that ensures reaching a state in
  $\{ x,y \}$. By opting for this strategy, which corresponds to
  choosing the second column in the local interaction, it follows that
  the states reachable in the concurrent arena from $q$ are
  $\{ x,y \}$ (depicted in blue). Hence, we set
  $\Par(\lambda)(m,q) = b_2$.
\end{example}

Just like for Player $\A$, the value of the game for Player $\B$ does not worsen with parallelization, as long as the arena is locally determined.
\begin{lemma}[Proof given page~\pageref{proof:lemma3}]
	\label{lem:parallel_win_B}
	Assume that the concurrent arena $\Aconc$ is locally determined. Consider a strategy $\sigma_\B \in \SetColStrat{\ATurnConc}{\B}$ in $\ATurnConc$ for Player $\B$ and its parallel version $\Par(\sigma_\B) \in \SetColStrat{\Aconc}{\B}$ in $\Aconc$. Then, $\val{\Aconc}{\Par(\sigma_\B)}[W] \leq \val{\Seq(\Aconc)}{\sigma_\B}[\Seq(W)]$. It follows that, $\val{\Aconc}{\B}[W] \leq \val{\Seq(\Aconc)}{\B}[\Seq(W)]$.
\end{lemma}

Overall, we obtain that the values of the concurrent game $\langle \Aconc,W \rangle$ and its sequential version $\langle \Seq(\Aconc),\Seq(W) \rangle$ are equal for both players. The following theorem is a direct consequences of Lemmas~\ref{lem:sequential_win} (stated for both players),~\ref{lem:parallel_win_A} and~\ref{lem:parallel_win_B}. Theorem~\ref{thm:main} is then a consequence of this theorem along with Observations~\ref{obs:memory_used_sequen} and~\ref{obs:memory_used_para}.
\begin{theorem}
	\label{thm:all_prop}
	If the game 
	$\langle \Aconc,W \rangle$ is locally determined, we have the following:
	\begin{itemize}
		\item $\val{\Aconc}{\A}[W] = \val{\Seq(\Aconc)}{\A}[\Seq(W)]$ and $\val{\Aconc}{\B}[W] = \val{\Seq(\Aconc)}{\B}[\Seq(W)]$;
		\item the game $\langle \Aconc,W \rangle$ is limit-determined iff its sequential version $\langle \Seq(\Aconc),\Seq(W) \rangle$ is;
		\item if a color strategy $\s$ is optimal in $\langle \Aconc,W \rangle$, so is 
		$\Seq(s)$ in $\langle \Seq(\Aconc),\Seq(W) \rangle$;
		\item if a color strategy $\sigma$ is optimal in $\langle \Seq(\Aconc),\Seq(W) \rangle$, so is 
		$\Par(\sigma)$ in $\langle \Aconc,W \rangle$.
	\end{itemize}
\end{theorem}

%% file: Tex/Applications.tex



\subsection{Games with deterministic Nature (i.e. without Nature)}
{\label{ref:deterministic}}
We first consider the special case of games with a deterministic
Nature, since, on turn-based games, this setting enjoys more
determinacy results than the stochastic one. A concurrent arena
$\Aconc$ is \emph{deterministic} if, for all Nature state $d \in
\distribSet$, there exists state $q \in Q$ such that
$\distribFunc(d)[q] = 1$ (see Definition~\ref{def:deterministic_game}).	
{\label{ref:compatible}} In such a setting, it is relevant to
consider all infinite paths that are compatible with a chromatic strategy,
and not only their probability.
A winning
strategy is then a strategy ensuring that its set of compatible paths
is included in the 
winning set -- $U_W$ for Player $\A$ and $Q^\omega \setminus U_W$ for
Player $\B$ (see
Definition~\ref{def:compatible_paths}). 

In a deterministic setting, we consider the notion of exact-determinacy: a deterministic concurrent game $\langle \Aconc,W \rangle$ is \emph{exactly-determined} (resp. positionally, resp. finite-memory) if either of the player has a (resp. positional, resp. finite-memory) winning strategy. In the
  literature this notion is sometimes called ``sure winning'', while
  winning with probability $1$ is called ``almost-sure winning''. 
{\label{ref:prop_det_value}}
However, in deterministic concurrent games 
with chromatic strategies (recall that they are deterministic strategies),
we have an equivalence between the two notions (see Lemma~\ref{lem:winning_proba_one}).
This immediately gives us the following corollary:
\begin{corollary}
	A deterministic concurrent game 
	is (resp. positionally, resp. finite-memory) exactly-determined if and only if 
	it is (resp. positionally, resp. finite-memory)
        determined.
	\label{thm:deterministic_determined}
\end{corollary}


In the following, the determinacy of a deterministic game will refer to exact-determinacy. Let us now consider how to translate determinacy results from turn-based games to locally-determined concurrent game in two cases. 

\subparagraph{Borel determinacy} Let us apply Theorem~\ref{thm:main}, and Corollary~\ref{thm:deterministic_determined} to prove the Borel determinacy of locally determined concurrent games. {\label{ref:martin_borel}} By rephrasing the famous result of Borel determinacy in our formalism
, we have that a deterministic turn-based graph arena $\Aconc
$ is determined for all Borel winning set $W \subseteq \Borel(\colSet)$ 
(see Theorem~\ref{thm:turn_based_borel_determined}).
Note that this theorem is not directly given by the results proved by Martin in~\cite{martin1975borel,martin1985purely,martin1998determinacy}. To obtain this theorem, we additionally need to apply a result from~\cite{DBLP:conf/cie/Roux20} since a strategy depends on colors history instead of state history
.
Let us now use 
this result to prove the determinacy of locally determined concurrent games. (Which can be written as an equivalence, see Appendix~\ref{details:rmq64}).
\begin{theorem}[Proof in Appendix~\ref{proof:coro62}]
  For all Borel winning set $W$, for all locally determined deterministic concurrent graph arena $\Aconc$, the concurrent game $\Games{\Aconc}{W}$ is determined. Conversely, for all non-trivial Borel winning set $\emptyset \subsetneq W \subsetneq \colSet^\omega$, for all non-determined game form $\formNF$, there exists a deterministic concurrent arena $\Aconc$ with only $\formNF$ as non-determined local interaction such that the game $\Games{\Aconc}{W}$ is not determined.
  \label{coro:conc_borel_determined}
\end{theorem}
{\label{ref:equivalence_borel}}

\subparagraph{Finite-memory determinacy} The next application only applies to finite arenas. 
In~\cite{DBLP:conf/concur/Bouyer0ORV20}, the authors
proved an equivalence between the shape of a winning set and the
existence of winning strategies that can be implemented with a given
memory skeleton $\mathcal{M}$\footnote{In fact, they looked at the
  existence of Nash equilibria with antagonistic preference relations
  instead of winning sets. However, a winning set $W \subseteq \colSet^\omega$ can be directly translated into an equivalent preference relation ${\prec_W} \subseteq \colSet^\omega \times \colSet^\omega$ by $\rho \prec_W \rho' \Leftrightarrow \rho \not \in W \wedge \rho' \in W$. In the following we will refer to the preference relation $\prec_W$ when mentionning the winning set
  $W$.}. They defined the properties of $\mathcal{M}$-selectivity and
$\mathcal{M}$-monotony (which we recall in the
appendix, see Definition~\ref{definition:monotony_selectivity}) and proved that for $\mathcal{M}$ a memory skeleton and $W \subseteq \colSet^\omega$, we have that $W$ and $\colSet^\omega \setminus W$ are
$\mathcal{M}$-monotone and $\mathcal{M}$-selective is equivalent to every finite deterministic turn-based game with $W$ as winning set is
determined with winning strategies for both players that can be
found among strategies implemented with memory skeleton
$\mathcal{M}$ (see Theorem~\ref{thm:turn_based_memory_winning}).{\label{ref:finite_mem}}

{\label{ref:thm_mono_sequen}}
For $\langle \Aconc,W \rangle$ be a deterministic concurrent game on
the concurrent arena $\Aconc$, $\langle \ATurnConc,\Seq(W) \rangle$
its sequential version, and $\mathcal{M}$ a memory skeleton on
$\colSet$, we have that $W$ is $\mathcal{M}$-monotone and
$\mathcal{M}$-selective if and only if $\Seq(W)$ is
$\Seq(\mathcal{M})$-monotone and $\Seq(\mathcal{M})$-selective (see Theorem~\ref{thm:winning_condition_monotony_selectivity}).
Note that the proof of this theorem, longer than the other applications 
requires to establish some algebraic properties of the projection function $\projec{\colSetSeq}{\colSet}: \starom{\colSetSeq} \rightarrow \starom{\colSet}$.
We can now extend Theorem~\ref{thm:turn_based_memory_winning} to some concurrent games. 
\begin{theorem}[Proof in Appendix~\ref{proof:thm66}]
  Let $\mathcal{M}$ be a memory skeleton and $W \subseteq \colSet^\omega$. T
  he two following assertions are equivalent:
  \begin{enumerate}
  \item every finite deterministic locally determined concurrent game $\langle \Aconc,W \rangle$ 
  with finite action sets is determined with winning strategies for both players that can be found among strategies implemented with memory skeleton
    $\mathcal{M}$;
  \item $W$ and $\colSet^\omega \setminus W$ are
    $\mathcal{M}$-monotone and $\mathcal{M}$-selective. 
  \end{enumerate}
	\label{thm:conc_equiv_memory_winning}
\end{theorem}
Note that, as in the case of Borel determinacy, we have that local
determinacy is somehow also a necessary condition since a one-shot
reachability game may not be determined as soon as the local
interaction at the initial state is not determined (see
Figures~\ref{fig:one_shot_good} and~\ref{fig:one_shot_bad} in the
introduction). {\label{ref:equiv_memory}} We can also rewrite this
theorem as a more involved equivalence (see Theorem~\ref{thm:memory_conc}).

\subsection{Stochastic Games (i.e. with Nature)}
There are fewer determinacy results on stochastic games, especially with deterministic strategies. Let us translate some of then into locally determined concurrent games. We consider parity objectives and the more general case of tail-objectives (a.k.a. prefix-independent).

\subparagraph{Parity Objectives} As already mentioned in Section~\ref{sec:gameForm}, parity objectives are defined as follows.
For a set of colors $\colSet = \llbracket m,n \rrbracket$ for some $m,n \in \mathbb{N}$, a parity objective on $\colSet$ is the winning set $W = \{ \rho \in \colSet^\omega \mid \max (\mathsf{n}_\infty(\rho)) \text{ is even } \}$ where $\mathsf{n}_\infty(\rho)$ is the set of colors seen infinitely often in $\rho$.
{\label{ref:thm_parity}} A result from
\cite{DBLP:conf/soda/ChatterjeeJH04,DBLP:conf/fossacs/Zielonka04}
gives us that any finite turn-based parity game is positionally
determined (see Theorem~\ref{thm:positional_determinacy_parity_turn}). This result can be directly transferred to locally determined concurrent games thanks to 
Theorem~\ref{thm:main}. Note that, as in the two previous cases, the local determinacy assumption is somewhat necessary. (See the equivalence in Theorem~\ref{thm:parity_positionaly_determined_conc}.
{\label{ref:parity_equiv}}

\begin{theorem}[Proof in Appendix~\ref{proof:coro71}]
	Consider a (stochastic) locally determined  finite concurrent
        graph arena $\Aconc = \AcoloredConc$ with $\colSet =
        \llbracket m,n \rrbracket$ for some $m,n \in \mathbb{N}$ and finite action sets. For all 
	winning set $W \in \Borel(\colSet)$ that is a parity objective on $\colSet$, the concurrent game $\Games{\Aconc}{W}$ is positionally determined.
	\label{coro:positional_determinacy_parity_conc}
\end{theorem}

\subparagraph{Tail Objectives}
Let us now consider more general objectives than the parity objectives. In particular, positional determinacy does not hold in the general case for these objectives (consider, for instance, the Muller objectives). A tail objective is a winning set that is closed by adding and removing finite prefixes, that is, for a set of colors $\colSet$, a winning set $W \in \Borel(\colSet)$ is a tail-objective if, for all $\rho \in \colSet^\omega$ and $\pi \in \colSet^*$, we have $\rho \in W \Leftrightarrow \pi \cdot \rho \in W$.
In particular, a parity objective is a tail objective. In fact, we
have that every  finite turn-based game that is limit-determined with
value $0$ or $1$ is determined (see Theorem~\ref{thm:prefix_independent_determinacy_parity_turn}). {\label{ref:thm_tail}}



This result can be directly transferred to locally determined
concurrent games. 
As usual, the local determinacy is a somewhat necessary condition.

\begin{theorem}[Proof in Appendix~\ref{proof:coro77}]
	Consider a locally determined (stochastic) finite concurrent graph arena $\Aconc
	$ with finite action sets. Then, for all Borel winning set $W \subseteq \Borel(\colSet)$ that is a tail objective, on $\colSet$, if $\val{\Aconc}{\A}[W] = 1$ or $\val{\Aconc}{\B}[W] = 0$, then the game is determined.
	\label{coro:prefix_independent_determinacy_parity_conc}
\end{theorem}

%% file: Tex/Appendix.tex
\section{Extended preliminaries}
\label{sec:appendix_preliminaries}
We assume that the sets of actions $A$ and $B$ for both players are provided with well-founded order $<_A$ and $<_B$ that can implement a choice function on these sets.

Consider a non-empty set $D$. The notations $D^*$, $D^+$, and $D^\omega$ 
respectively refer to the set of finite sequences, of non-empty finite 
sequences and of infinite sequences of elements of $D$. Recall that we denote $\starom{D} := D^* \cup D^\omega$. For a sequence $\pi = \pi_0 \pi_1 \ldots \in 
\starom{D}$, the length of $\pi$ is denoted $|\pi|$, and is equal to $\infty$ if $\pi \in D^\omega$. For $i < |\pi|$, $\pi_{\leq i}$ refers to the finite sequence $\pi_0 \ldots \pi_i$ and if $\pi \in D^+$, we have $\head(\pi) = \pi_{|\pi| - 1}$ and $\tail(\pi) = \pi_{\leq |\pi|-2}$. For instance, $\head(a \cdot b \cdot c) = c$ and $\tail(a \cdot b \cdot c) = a \cdot b$.

For a function $\mathsf{f}: E \rightarrow F$ between two arbitrary sets $E$ and 
$F$, and a subset $E' \subseteq E$, the notation $\mathsf{f}[E']$ refers to the 
set $\{ \mathsf{f}(e) \mid e \in E' \}$
. 
For a function $\mathsf{f}: E_1 \times E_2 \rightarrow 
F$ for two arbitrary sets $E_1$ and $E_2$, for $e \in E_1$, $\mathsf{f}(e,E_2)$ 
refers to the set $\{ \mathsf{f}(e,e') \mid e' \in E_2 \}$. 


A \emph{discrete probabilistic distribution} over a non-empty set $D$ is a function $\mu: D \rightarrow [0,1]$ such that $\sum_{x \in D} \mu(x) = 1$. The set of all distribution over the set $D$ is denoted $\Dist(D)$. The \emph{support} $\Supp(\mu)$ of a probability distribution $\mu \in \Dist(D)$ is the set of value whose image is nonzero: $\Supp(\mu) = \mu^{-1}[\; (0,1] \;]$. By definition, the support $\Supp(\mu)$ is finite or countable. A distribution $\mu$ is \emph{Dirac} if $|\Supp(\mu)| = 1$. In the following, a element $d \in D$ will be seen a the Dirac distribution $\mu: D \rightarrow [0,1]$ such that $\mu(d) = 1$. 

A $\sigma$-algebra $\mathcal{D}$ on a set $D^\omega$ is such that $\mathcal{D} \subseteq \mathcal{P}(D^\omega)$, $D \in \mathcal{D}$, and is closed under complementation and countable union. A probability measure on a $\sigma$-algebra $\mathcal{D}$ is a function $\upsilon: \mathcal{D} \rightarrow [0,1]$ such that $\upsilon(\emptyset) = 0$, $\upsilon(D) = 1$, and $\upsilon(\cup_{i \in \Nat} A_i) = \sum_{i \in \Nat} A_i$ for all $(A_i)_{i \in \Nat} \in \mathcal{D}^{\Nat}$ pairwise disjoint. 
As noted in the preliminaries, two probability measures that coincide on cylinder sets are equal, as stated below.
\begin{theorem}
	Consider two probability measures $\nu,\nu': \Borel(Q) \rightarrow [0,1]$ such that, for all $C \in \cyl_Q$, we have $\nu(C) = \nu'(C)$. Then, $\nu = \nu$.
	\label{thm:probability_unique_Borel}
\end{theorem}


For a non-empty set $D$, we say that $\pi \in D^+$ is a prefix of $\pi' \in 
\starom{D}$, denoted $\pi \sqsubseteq \pi'$, if $\pi = \pi'_{\leq |\pi|-1}$. An 
infinite set of prefixes 
$(\pi^n)_{n \geq 0} \in (D^+)^\mathbb{N}$ of increasing length ensuring $\pi^n 
\sqsubseteq \pi^{n+1}$ for all $n \geq 0$ uniquely defines an infinite sequence 
$\pi \in D^\omega$: for all $k \geq 0$, $\pi_{\leq k} = \pi^{n}_{\leq k}$ for 
all $n \geq 0$ such that $|\pi^n|-1 \geq k$.

\section{Complement on Section~\ref{sec:gameForm}}
\label{sec:determinedVStree}

A tree game form is like a finite turn-based game on a tree, with outcomes at the leaves, but without preferences, see Figure~\ref{fig:simple_tree}. A tree game form can be seen as (or trivially translated into) a game form: the players, strategies, and outcomes all remain the same. Figure~\ref{fig:game_from_tree} shows the translation of Figure~\ref{fig:simple_tree}.

\begin{figure*}[htb]
	\begin{minipage}[t]{0.45\linewidth}
		\centering
		\includegraphics[scale=1]{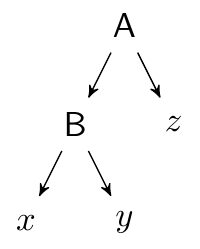}
		\caption{A turn-based tree where Player $\A$ plays at the root and 
		Player $\B$ plays at the other internal node
		.}
		\label{fig:simple_tree}
	\end{minipage} \hfill
	\begin{minipage}[t]{0.45\linewidth}
		\centering
		\includegraphics[scale=2]{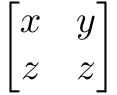}
		\caption{A representation of the game form derived from the simple tree on the left. Player $\A$ chooses the lines whereas Player $\B$ chooses the columns.}
		\label{fig:game_from_tree}
	\end{minipage}
\end{figure*}

To find out whether the determined game forms are nothing but tree game forms in disguise, we need to define ``disguise''. Below we define three equivalence relations over game forms, each inducing a notion of similarity.
\begin{itemize}
	\item Let $G \sim_d G'$ if $G$ and $G'$ are equal up to duplication of rows or columns.
	
	\item Let $G \sim G'$ if for every subset $O'$ of outcomes, there exists a row (resp. column) of $G$ involving exactly the outcomes in $O'$ iff there is one in $G'$. 
	
	\item Let $G \sim_w G'$ if for every Boolean valuation $v$ over the outcomes, for every player, she has a winning strategy in $G(v)$ iff she has one in $G'(v)$.
\end{itemize}
It is straightforward to show that $G \sim_w G'$ iff the following holds: for every subset $O'$ of outcomes, there exists a row (resp. column) of $G$ involving only outcomes in $O'$ iff there is one in $G'$. Using this, it is clear that $\sim_d \subseteq \sim \subseteq \sim_w$.

We argue below that $\sim_d$ is the natural  relation to express similarity to a tree game form, and show that some determined game forms are not even $\sim$-similar to any tree game form.

In many settings, the similarity $\sim_d$ may be considered the natural one. One reason pertains to tree games, where there are two natural notions of strategy: in \emph{complete} strategies, each player says which subgame she would choose at each node that she controls; in \emph{minimalistic} strategies, choices are not specified in subgames that the strategy discards. See Figures ~\ref{fig:tree_example}, \ref{fig:minimalist} and \ref{fig:complete}, where $G$ and $G'$ are the minimalist and complete trivial translations of the tree game form. It is straightforward to prove that $G \sim_d G'$ always holds. Also, informally, $\sim_d$ is the smallest ``simple'' equivalence relation that equates the two translated of a given tree game form.

\begin{figure*}[htb]
	\begin{minipage}[t]{0.3\linewidth}
		\centering
		\includegraphics[scale=0.75]{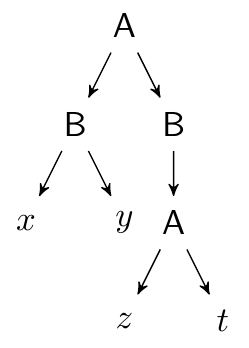}
		\caption{A turn-based tree between Player $\A$ and $\B$.}
		\label{fig:tree_example}
	\end{minipage} 
	\begin{minipage}[t]{0.3\linewidth}
		\centering
		\includegraphics[scale=1.5]{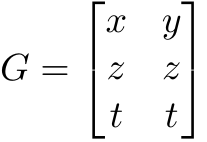}
		\caption{A representation of the game form derived from the tree o the left 
		with minimalist strategies: if Player $\A$ plays left at the root, it does not matter what is played if she had played right.}
		\label{fig:minimalist}
	\end{minipage}
	\begin{minipage}[t]{0.3\linewidth}
		\centering
		\includegraphics[scale=1.5]{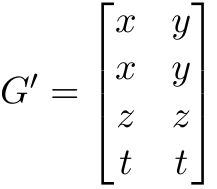}
		\caption{A representation of the game form derived from the tree on the left with complete strategies: an action is specified at each node Played $\A$ owns, regardless of what is played on the root.}
		\label{fig:complete}
	\end{minipage}
\end{figure*}

In many settings, $\sim_w$ is too generous, as suggested by the following two examples. First, consider a determined game form $G$. We can derive two two-step tree game forms from $G$: by letting Player A (resp. B)  choose her strategy first, then by letting the opponent do it. Then, $G$, $T_A$, and $T_B$ are $\sim_w$-similar, although they do not have much in common. This is especially striking if $G$ itself is a tree game form: $G$ and, say, $T_A$ differ wrt other solution concepts such as subgame perfect equilibrium. Second example, the game form $I_4$ in Figure~\ref{fig:examples_game_forms} is determined, but if we instantiate $(x,y,z)$ into $(1,0,\frac{1}{2})$ and players into maximizer/minimizer, from every strategy profile but one, the better-response dynamics (i.e. one player at a time change strategies to improve) cannot lead to a Nash equilibrium: the players will cycle between $0$ and $1$, while possibly also visiting $\frac{1}{2}$. To the contrary \cite{DBLP:journals/geb/Kukushkin02}  
have shown that better-response dynamics in tree games is weakly terminating (actually something stronger). Thus, $\sim_w$ equates game forms whose dynamics are fundamentally different. 

Therefore, in general $\sim_w$ is too generous. Of course, it is sometimes exactly what is needed: our determinacy transfer and its maximality rely on a game form being determined iff it is $\sim_w$-similar to its derived two-step tree game forms.

The similarity $\sim$ lies between the natural $\sim_d$ and the generous $\sim_w$. Proposition~\ref{prop:gf-sim} below shows that some determined game forms are not $\sim$-similar to any tree game form. (Then, of course, neither $\sim_d$-similar.)


\begin{proposition}\label{prop:gf-sim}
	The game form (I4) in Figure~\ref{fig:examples_game_forms} is not $\sim$-similar to any tree game form.
	
	\begin{proof}
		Towards a contradiction, let $T$ be a tree game form where each of the players can offer either $\{z\}$ or $\{x,y,z\}$. (We say that the row/column player can offer some subset $O'$ of outcomes if some row/column contains exactly the outcomes in $O'$.)
		
		First, let us squeeze $T$ as follows: among the subgames of $T$ that involve only one outcome, collapse each largest ones into a leaf with that unique outcome. This yields $T_0$, satisfying $T_0 \sim T$. Second, let us prune $T_0$ by removing all the leaves with outcome $z$. This yields $T_1$, which is a tree game form since pruning was performed after squeezing/collapsing. 
		
		Let us make a case disjunction. First case, Player B can offer $\{x\}$ in $T_1$. So, using the same strategy, Player B can offer $\{x,z\}$  in $T$ (not just $\{x\}$ since Player B can offer $\{z\}$), contradiction. Second case, Player B cannot offer $\{x\}$ in $T_1$. By determinacy of $T_1$ (since it is a tree game form) and since it involves outcomes $x,y$ only, Player A can offer $\{y\}$. Contradiction as in the first case.
	\end{proof}
\end{proposition}
It is also possible, but more difficult, to prove that the game form $I_3$ in Figure~\ref{fig:examples_game_forms} is not $\sim_d$-similar to any tree game form. Although this game form is $\sim$-similar to the tree game form in Figure~\ref{fig:new_tree}, it suggests that $\sim$ is too generous. Indeed, on the one hand in the game form $I_3$ in Figure~\ref{fig:examples_game_forms}, every valuation of the issues with ordered pairs of real-valued payoffs yields a potential game, i.e. better-response dynamics terminates; on the other hand in Figure~\ref{fig:new_game_from_tree}, setting $x := 1$ and $y,z := 0$ yields an antagonistic game with a better-response cycle involving the four corners. So, $\sim$ equates two game forms whose dynamics are fundamentally different.  

\begin{figure*}[htb]
	\begin{minipage}[t]{0.45\linewidth}
		\centering
		\includegraphics[scale=1]{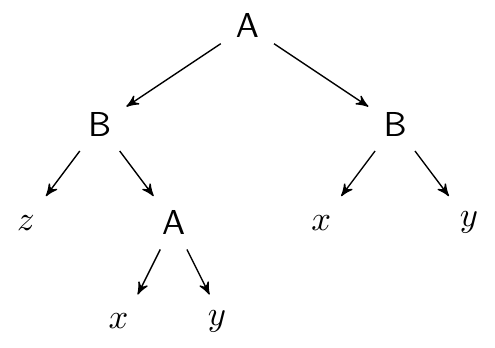}
		\caption{A turn-based tree where Player $\A$.}
		\label{fig:new_tree}
	\end{minipage} \hfill
	\begin{minipage}[t]{0.45\linewidth}
		\centering
		\includegraphics[scale=2]{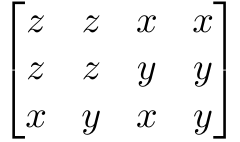}
		\caption{The game form induced by the game tree on the left.}
		\label{fig:new_game_from_tree}
	\end{minipage}
\end{figure*}

\section{Complement on Section~\ref{sec:arenas}}

\subsection{Probability distribution induced by a pair of strategies}
\label{subsec:strategies_probability_distribution}
\begin{definition}[Mentioned~\ref{ref:pair_strat_proba}]
	Let us consider an arena $\Aconc
	$ and $\s_\A,\s_\B \in \SetStrat{\Aconc}{\A} \times \SetStrat{\Aconc}{\B}$ two state strategies for Player $\A$ and $\B$. We denote by $\deltaDistrib^{\s_\A,\s_\B}: Q^+ \rightarrow \Dist(Q)$ the function giving the probabilistic distribution over the next state of the arena given the sequence of states already seen. For all finite path $\pi \in Q^+$ and $q \in Q$, we have:
	\begin{displaymath}
	\deltaDistrib^{\s_\A,\s_\B}(\pi)[q] = \sum_{a \in \Supp(\s_\A(\pi))} \sum_{b \in \Supp(\s_\B(\pi))} \s_\A(\pi)[a] \cdot \s_\B(\pi)[b] \cdot \deltaDistrib(\head(\pi),a,b)[q]
	\end{displaymath}
	
	Then, the probability of occurrence of a finite path $\pi = \pi_0 \cdots \pi_n \in Q^+$ with the pair of strategies $(\s_\A,\s_\B)$ is equal to $\prob{\Aconc}{\s_\A}{\s_\B}(\pi) = \Pi_{i = 0}^{n-1} \deltaDistrib^{\s_\A,\s_\B}(\pi_{\leq i})[\pi_{i+1}]$ if $\pi_0 = q_0$ and $0$ otherwise. The set of finite or infinite paths that can occur on a game $\Aconc$ is denoted $\OutStrat{\Aconc}^\vartriangle$ for $\vartriangle \in \{ +,\omega \}$ with $\OutStrat{\Aconc}^\vartriangle = \{ \pi \in Q^\vartriangle \mid \exists \s_\A,\s_\B \in \SetStrat{\Aconc}{\A} \times \SetStrat{\Aconc}{\B},\; \prob{\Aconc}{\s_\A}{\s_\B}(\pi_{\leq |\pi|-1}) > 0 \}$. 
	The probability of a cylinder set $\cyl(\pi)$ is $\prob{\Aconc}{\s_\A}{\s_\B}[\cyl(\pi)] = \prob{\Aconc}{\s_\A}{\s_\B}(\pi)$ for any finite path $\pi \in Q^*$. This induces the probability of any Borel set in the usual way,
	we denote by $\prob{\Aconc}{\s_\A}{\s_\B}: \Borel(Q) \rightarrow [0,1]$ the corresponding probability measure.
\end{definition}

\subsection{Definition of turn-based game}
\label{subsec:turn_based}
\begin{definition}[Mentioned~\ref{ref:turn_based}]
	Consider a stochastic concurrent arena $\Aconc
	$. It is said to be \emph{turn-based} if, for all $q \in Q$, we have that either for all $a \in \setA$, the partial function $\delta(q,a,\cdot): B \rightarrow \distribSet$ is constant -- in which case we say that the state $q$ belongs to Player $\A$ -- or for all $b \in \setB$, the partial function $\delta(q,\cdot,B): A \rightarrow \distribSet$ is constant --  in that case, the state $q$ is said to belong to Player $\B$. In the special case where a state belongs to both players -- which can only happen in states with trivial interaction -- who the state formally belongs to is a matter of convention.
\end{definition}

\section{Complements on Section~\ref{sec:SeqPar}}
\subsection{Projection of paths in a concurrent arena and its sequential version}
\begin{proposition}[Mentioned~\ref{ref:prop_projection_colors_colors_projection}]
	Consider a concurrent arena $\Aconc = \AcoloredConc$ and its sequential 
	version $\ATurnConc = \AcoloredSeq$. Then, for $\vartriangle \in \{+,\omega 
	\}$ and for all paths $\pi \in \OutStrat{\ATurnConc}^\vartriangle$, we have:
	\begin{itemize}
		\item Every path that may occur in that sequential version alternate between states $V_\A$ and $V_\B$: $\OutStrat{\ATurnConc}^\vartriangle \subseteq (V_\A \cdot V_\B)^\vartriangle \cdot (\epsilon + V_\A)$
		\item the projection of the path $\pi$ in $\Seq(\Aconc)$ on $Q$ is a path in $\Aconc$: 
		$\projec{V}{Q}(\pi) \in \OutStrat{\Aconc}^\vartriangle$;
		\item the projection of the colors of $\pi$ on $\colSet$ exactly 
		corresponds to the colors of the projection of $\pi$ on $Q$: 
		$\projec{\colSetSeq}{\colSet} \circ \extendFunc{\colFunc_\Aconc}(\pi) = 
		\extendFunc{\colFunc} \circ \projec{V}{Q}(\pi)$. 
	\end{itemize}
	\label{prop:projection_colors_colors_projection}
\end{proposition}
\begin{proof}
	For all $a \in A$, $b \in B$, $v \in V_\A$, and $(q,\Sq{q}{a'}) \in V_\B$, we have $\Supp(\deltaDistrib_\Aconc(v,a,b)) = \{ (v,\Sq{v}{a}) \} \subseteq V_\B$ and $\Supp(\deltaDistrib_\Aconc((q,\Sq{q}{a'}),a,b)) = \Supp(\deltaDistrib(q,a',b)) \subseteq Q = V_\A$. Therefore, a path $\pi \in \OutStrat{\ATurnConc}^\vartriangle$ is in $(V_\A \cdot V_\B)^\vartriangle \cdot (\epsilon + V_\A)$ and $\projec{V}{Q}(\pi) \in \OutStrat{\Aconc}^\vartriangle$. 
	
	
	Let us denote that path $\pi$ by $q_0 \cdot (q_0,\Sq{q_0}{a_0}) \cdot q_1 \cdot 
	(q_1,\Sq{q_1}{a_1}) \cdot q_2 \cdots$ for some $q_i \in Q$ and $a_i \in A$ 
	for all $i 
	\geq 0$. We have $\extendFunc{\colFunc_\Aconc}(\pi) = k_\Aconc \cdot 
	\colFunc(q_0,q_1) \cdot k_\Aconc \cdot \colFunc(q_1,q_2) \cdots$. 
	Therefore: 
	$$\projec{\colSetSeq}{\colSet} \circ 
	\extendFunc{\colFunc_\Aconc}(\pi) = \colFunc(q_0,q_1) \cdot 
	\colFunc(q_1,q_2) \cdots = \extendFunc{\colFunc}(q_0 \cdot q_1 \cdot q_2 
	\cdots) = \extendFunc{\colFunc} \circ \projec{V}{Q}(\pi)$$
\end{proof}

\subsection{Preimage of the projection of a finite path}
\label{subsec:path_extension}
\begin{proposition}[Mentioned~\ref{ref:path_extension}]
	Consider a concurrent arena $\Aconc = \AcoloredConc$ and a pair of state strategies $(\s_\A,\s_\B) \in \SetStrat{\Aconc}{\A} \times \SetStrat{\Aconc}{\B}$ in the concurrent game $\Aconc$ and a pair of state strategies in its sequential version $(\sigma_\A,\sigma_\B) \in \SetStrat{\Seq(\Aconc)}{\A} \times \SetStrat{\Seq(\Aconc)}{\B}$. Assume that, for all finite sequence of states $\pi \in Q^+$, we have $$\prob{\Aconc}{\s_\A}{\s_\B}(\cyl(\pi)) = \prob{\Seq(\Aconc)}{\sigma_\A}{\sigma_\B}(\inv{\projec{V}{Q}}[\cyl(\pi)])$$
	
	Then, for all Borel set $B \in \Borel(Q)$, we have:
	$$\prob{\Aconc}{\s_\A}{\s_\B}[B] = \prob{\Seq(\Aconc)}{\sigma_\A}{\sigma_\B}(\inv{\projec{V}{Q}}[B])$$
	\label{lem:sufficient_cond_prob_borel}
\end{proposition}
\vspace*{-7mm}
This proposition is particularly useful as it relates the probability of a winning set $W \in \Borel(Q)$ in the concurrent game $\Aconc$ with the probability of its sequential version $\Seq(W) \in \Borel(V)$ in the sequential version $\Seq(\Aconc)$ of the concurrent arena $\Aconc$. Let us proceed to the proof of this proposition.
\begin{proof}
	In the following we denote $\prob{\Aconc}{\s_\A}{\s_\B}$ by $p_\s$ and $\prob{\Seq(\Aconc)}{\sigma_\A}{\sigma_\B}$ by $p_\sigma$. Now, consider the function $p: \Borel(Q) \rightarrow [0,1]$ such that, for all Borel set $B \subseteq Q^\omega$, we have $p(B) = p_\sigma(\inv{\projec{V}{Q}}[B])$. We want to show that $p = p_\s$.
	
	Let us first show that $p$ is a probability measure on $\Borel(Q)$. Straightforwardly, $p(\emptyset) = p_\sigma(\emptyset) = 0$. Furthermore, $(V_\A \cdot V_\B)^\omega \subseteq \inv{\projec{V}{Q}}(Q^\omega)$ and $p_\sigma(V^\omega \setminus (V_\A \cdot V_\B)^\omega) = 0$ as by Proposition~\ref{prop:projection_colors_colors_projection}, we have $\OutStrat{\ATurnConc}^\omega \subseteq (V_\A \cdot V_\B)^\omega$ and $p_\sigma(V^\omega \setminus \OutStrat{\ATurnConc}^\omega) = 0$. Thus, $1 = p_\sigma((V_\A \cdot V_\B)^\omega) \leq p_\sigma(\inv{\projec{V}{Q}}[Q^\omega]) = p(Q^\omega)$ since $p_\sigma(V^\omega \setminus (V_\A \cdot V_\B)^\omega) = 0$
	. 
	Furthermore, for $(A_n)_{n \in \mathbb{N}} \in (\Borel(Q))^\mathbb{N}$ pairwise disjoint, we have $\inv{\projec{V}{Q}}[\cup_{n \in \mathbb{N}} A_n] = \cup_{n \in \mathbb{N}} \inv{\projec{V}{Q}}[A_n]$ and $(\inv{\projec{V}{Q}}[A_n])_{n \in \mathbb{N}}$ are also pairwise disjoint. Henceforth, we have $$p(\bigcup_{n \in \mathbb{N}} A_n) = p_\sigma(\bigcup_{n \in \mathbb{N}} \inv{\projec{V}{Q}}[A_n]) = \sum_{n \in \mathbb{N}} p_\sigma (\inv{\projec{V}{Q}}[A_n]) = \sum_{n \in \mathbb{N}} p(A_n)$$
	
	We conclude that $p$ is a probability measure on $\Borel(Q)$.
	
	Let us now show that $p$ and $p_\s$ coincide on cylinder sets. Let $c = \cyl(\pi) \in \cyl_Q$ for some $\pi \in Q^*$. We have:
	$$p(c) = p_\sigma(\inv{\projec{V}{Q}}[c]) = p_\s(c)$$
	Therefore, $p$ and $p_\s$ coincide on cylinder sets and it follows that $p$ and $p_\s$ are equal by Theorem~\ref{thm:probability_unique_Borel}. That is, for all Borel set $B \in \Borel(Q)$, we have $p_\s(B) = p(B) = p_\sigma(\inv{\projec{V}{Q}}(B))$.
\end{proof}
\begin{remark}
	\label{rmq:reverse_projec_cylinder}
	For all finite path $\pi \in Q^+$, we have:
	$$\inv{\projec{V}{Q}}[\cyl(\pi)] = \bigcup_{\substack{\rho \in \inv{\projec{V}{Q}}[\pi] \\ \head(\rho) = \head(\pi)}} \cyl(\rho)$$
\end{remark}

We define how to extend a sequence of colors that could occur in $\Aconc$ into the corresponding sequence of colors that could occur in $\Seq(\Aconc)$.
\begin{definition}[Sequence of colors extension]
	\label{def:rho_kc}
	In a concurrent arena $\Aconc = \AcoloredConc$, for a finite sequence of colors $\rho \in \colSet^*$, we denote by $\rho_{k_\Aconc} \in \colSetSeq^+$ the sequence of colors where the colors $k_\Aconc$ is added at every other index: if $\rho = k_1 \cdot k_2 \cdots k_n$, we have $\rho_{k_\Aconc} = k_{\Aconc} \cdot k_1 \cdot k_\Aconc \cdot k_2 \cdots k_\Aconc \cdot k_n$.
\end{definition} 
\begin{remark}
	\label{rmq:sequence_col_path_ext}
	For all finite path 
	$\rho \in (V_\A \cdot V_\B)^* \cdot V_\A$, 
	we have $\extendFunc{\colFunc}(\projec{V}{Q}(\rho))_{k_\Aconc} = \extendFunc{\colFunc_\Aconc}(\rho)$.
\end{remark}

\subsection{Definition of sequentialization}
\label{subsec:def_sequen}
\begin{definition}[Sequentialization of memory skeletons][Mentioned~\ref{ref:sequentialization_strat}]
	\label{def:sms}
	Consider a memory skeleton $\mathcal{M} = \langle M,\minit,\mu \rangle$ on 
	the set of colors $\colSet$. Then, the \emph{sequential version} of that 
	memory skeleton is the memory skeleton $\Seq(\mathcal{M}) = \langle 
	M,\minit,\Seq(\mu) \rangle$ on $\colSetSeq = \colSet \cup \{ k_\Aconc \}$ where $\Seq(\mu): M \times 
	\colSetSeq \rightarrow M$ is such that, for all $m \in M$ and $k \in 
	\colSetSeq$, we have $\Seq(\mu)(m,k) = \mu(m,k)$ if $k 
	\neq k_\Aconc$ and $\Seq(\mu)(m,k_\Aconc) = m$ otherwise.
	\label{def:sequentialization_memory_skeleton}
\end{definition}

\begin{definition}[Sequentialization of action maps][Mentioned~\ref{ref:sequentialization_strat}]
	Consider an action function for Player $\A$ (resp. $\B$) $\lambda: M \times Q \rightarrow A$ (resp. $M \times Q \rightarrow B$). Then, its 
	\emph{sequential version} is the action map $\Seq(\lambda): M \times 
	V \rightarrow A$ (resp. $M \times V \rightarrow B$) where, for all $m \in M$ and $v \in V_\A$, we have $\Seq(\lambda)(m,v) = \lambda(m,v)$ (resp. for $m \in M$ and $(v,a) \in V_\B$, we have $\Seq(\lambda)(m,(v,a)) = \lambda(m,v)$) and for $v' \in V_\B$, we have $\Seq(\lambda)(m,v') = \min_{<_\A} A$ (resp. for $v' \in V_\A$, we have $\Seq(\lambda)(m,v') = \min_{<_\B} B$).
	\label{def:sequen_action_function}
\end{definition}

{\label{ref:sequen_strat}}
\begin{definition}[Sequentialization of strategies][Justification~\ref{proof:remark5}][Mentioned~\ref{ref:sequentialization_strat}]
	\label{def:seq-strat}
	Consider a color strategy $\s$ in $\Aconc$ implemented by a memory skeleton $\mathcal{M}$ and an action map $\lambda$. The sequential version of $\s$ is the strategy $\Seq(s)$ implemented by the sequential versions of the memory skeleton $\Seq(\mathcal{M})$ and of the action map $\Seq(\lambda)$.
	\label{def:sequentialization_strat}
\end{definition}

\subsection{Proposition regarding the sequentialization of the update function}
\label{proof:Proposition3}
The definition of the update function $\Seq(\mu)$ from the update function $\mu$ ensures the following:
\begin{proposition}[Mentioned~\ref{ref:projection_memory}]
	For all non-empty set $M$, $m \in M$, $\mu: M \times \colSet \rightarrow M$, and $\rho \in \colSetSeq^*$, the state of memory w.r.t. the 
	update function $\Seq(\mu)$ after the colors in $\rho$ have been seen is the same  as after the colors of the projection of $\rho$ on $\colSet$  are 
	seen w.r.t. the update function $\mu$: $\extendFunc{\Seq(\mu)}(m,\rho) = 
	\extendFunc{\mu}(m,\projec{\colSetSeq}{\colSet}(\rho))$.
	\label{prop:projection_memory}
\end{proposition}
\begin{proof}
	Consider a non-empty set $M$, $\mu: M \times \colSet \rightarrow M$, and $m 
	\in M$. We proceed by induction on $\rho \in \colSetSeq^*$. We have 
	$\extendFunc{\Seq(\mu)}(m,\epsilon) = m =
	\extendFunc{\mu}(m,\projec{\colSetSeq}{\colSet}(\epsilon))$. Now, let $\rho 
	= \rho' \cdot k \in \colSetSeq^+$ and assume that 
	$\extendFunc{\Seq(\mu)}(m,\rho') =	\extendFunc{\mu}(m,\projec{\colSetSeq}{\colSet}(\rho'))$. 
	
	If $k = k_\Aconc$, we have: 
	\begin{align*}
		\extendFunc{\Seq(\mu)}(m,\rho) & = 	\Seq(\mu)(\extendFunc{\Seq(\mu)}(m,\rho'),k_\Aconc) & \text{ by definition of } \extendFunc{\Seq(\mu)}\\
		&= 	\extendFunc{\Seq(\mu)}(m,\rho') & \text{ by definition of } \Seq(\mu)\\
		&= \extendFunc{\mu}(m,\projec{\colSetSeq}{\colSet}(\rho')) & \text{ by assumption }\\
		&= 	\extendFunc{\mu}(m,\projec{\colSetSeq}{\colSet}(\rho)) & \text{ since } \rho' = \rho \cdot k_\Aconc \text{ and } k_\Aconc \notin \colSet
	\end{align*}
	
	Otherwise:
	\begin{align*}
		\extendFunc{\Seq(\mu)}(m,\rho) & = 
		\Seq(\mu)(\extendFunc{\Seq(\mu)}(m,\rho'),k) & \text{ by definition of } \extendFunc{\Seq(\mu)} \\ 
		& = \mu(\extendFunc{\Seq(\mu)}(m,\rho'),k) & \text{ by definition of } \Seq(\mu) \\ 
		& =	\mu(\extendFunc{\mu}(m,\projec{\colSetSeq}{\colSet}(\rho')),k) & \text{ by assumption } \\ 
		& =	\extendFunc{\mu}(m,\projec{\colSetSeq}{\colSet}(\rho') \cdot k) & \text{ by definition of } \extendFunc{\mu}\\
		& =	\extendFunc{\mu}(m,\projec{\colSetSeq}{\colSet}(\rho)) & \text{ since } \rho' = \rho \cdot k \text{ and } k_\Aconc \in \colSet
	\end{align*}
	In any case, we have $\extendFunc{\Seq(\mu)}(m,\rho) = \extendFunc{\mu}(m,\projec{\colSetSeq}{\colSet}(\rho))$.
\end{proof}

\subsection{Proof that the sequentialization of strategies is well defined}
\label{proof:remark5}

As it is proved below, the strategy $\Seq(\s)$ is well defined in the sense that if two different pairs of memory skeleton and action map implement the strategy $\s$, then the strategies implemented by the sequential versions of both of these pairs are also the same.
\begin{proof}[Mentioned~\ref{ref:sequen_strat}]
	We consider the case of Player $\A$. The case of Player $\B$ is analogous. 
	
	Consider a strategy $\s: \colSet^* \times Q \rightarrow A$ in the 
	concurrent game $\Aconc$ implemented by the memory skeleton 
	$\mathcal{M} = \langle M,\minit,\mu \rangle$ and the action map
	$\lambda: M \times Q \rightarrow A$. Let $\sigma: \colSetSeq^* \times V 
	\rightarrow A$ be the strategy implemented the memory 
	skeleton $\Seq(\mathcal{M})$ and the action map $\Seq(\lambda)$. 
	
	Let $\rho \in \colSetSeq^*$ and $v \in V$. We have $\sigma(\rho,v) = 
	\Seq(\lambda)(\extendFunc{\Seq(\mu)}(\minit,\rho),v)$. If $v \in V_\B$, then we have $\sigma(\rho,v) = \min_{<_\A} A$. Otherwise, if $v \in V_\A$, we have: $\sigma(\rho,v) = \lambda(\extendFunc{\Seq(\mu)}(\minit,\rho),v) = 
\lambda(\extendFunc{\mu}(\minit,\projec{\colSetSeq}{\colSet}(\rho)),v)$
	by Proposition~\ref{prop:projection_memory}. Furthermore, we have 
	$\lambda(\extendFunc{\mu}(\minit,\projec{\colSetSeq}{\colSet}(\rho)),v) = 
	\s(\projec{\colSetSeq}{\colSet}(\rho),v)$. That is, $\sigma(\rho,v) = 	
	\s(\projec{\colSetSeq}{\colSet}(\rho),v)$. Hence, in both cases, the value of $\sigma$ is independent of $\mathcal{M}$ and $\lambda$. Therefore, the value of the sequential strategy $\sigma = \Seq(\s)$ does not depend on the pair of memory skeleton and action map considered that implements the strategy $\s$.
\end{proof}

\subsection{Proof of Lemma~\ref{lem:sequential_win}}
\label{proof:lemma1}
	Before proceeding to the proof of this lemma, we introduce a definition and state and prove a useful lemma that will be used several times later on. 
\begin{definition}
	Consider a concurrent arena $\Aconc = \AcoloredConc$ and a pair of strategies $(\s_\A,\s_\B) \in \SetStrat{\Aconc}{\A} \times  \SetSt{\Aconc}{\B}$. We denote by $\nabla^{\s_\A,\s_\B}_\Aconc: Q^+ \rightarrow \Dist(\distribSet)$ the function ensuring, for all $\pi \in Q^+$ and $d \in \distribSet$:
	$$\nabla^{\s_\A,\s_\B}_\Aconc(\pi)[d] = \sum_{\substack{(a,b) \in A \times B \\ \delta(\head(\pi),a,b) = d}} \s_\A(\pi)[a] \cdot \s_\B(\pi)[b]$$
	
	If we now consider its sequential version $\ATurnConc = \AcoloredSeq$ and a pair of state strategies $(\sigma_\A,\sigma_\B) \in \SetStrat{\ATurnConc}{\A} \times \SetStrat{\ATurnConc}{\B}$ in $\ATurnConc$, we denote by $\nabla^{\sigma_\A,\sigma_\B}_{\ATurnConc}: (V_\A \cdot V_\B)^* \cdot V_\A \rightarrow \Dist(\distribSet_\B)$ the function ensuring, for all $\rho \in (V_\A \cdot V_\B)^* \cdot V_\A$ with $\head(\rho) = t \in V_\A = Q$ and $d \in \distribSet_\B$:
	$$\nabla^{\sigma_\A,\sigma_\B}_{\ATurnConc}(\rho)[d] = \sum_{\substack{(a,b) \in A \times B \\ \delta(t,a,b) = d}} \sigma_\A(\rho)[a] \cdot \sigma_\B(\rho \cdot (t,a))[b]$$
\end{definition}
That is, given two state strategies and a finite path $\pi \in Q^+$, $\nabla^{\s_\A,\s_\B}(\pi)$ gives the probability distribution over the set of Nature states. Note that the function $\deltaDistrib^{\s_\A,\s_\B}$ can be expressed with $\nabla^{\s_\A,\s_\B}$.
\begin{remark}
	\label{rmq:nabla_delta}
	For all $\pi \in Q^+$ and $q \in Q$, we have $$\deltaDistrib^{\s_\A,\s_\B}(\pi)[q] = \sum_{d \in \distribSet} \nabla^{\s_\A,\s_\B}(\pi)[d] \cdot \distribFunc(d)[q]$$
\end{remark}

Now, we have the following proposition:
\begin{proposition}
	Let $\Aconc = \AcoloredConc$ be a concurrent arena and $\ATurnConc = \AcoloredSeq$ be its sequential version. Consider a pair of state strategies $(\s_\A,\s_\B) \in \SetStrat{\Aconc}{\A} \times  \SetStrat{\Aconc}{\B}$ in $\Aconc$ and a pair of state strategies $(\sigma_\A,\sigma_\B) \in \SetStrat{\ATurnConc}{\A} \times  \SetStrat{\ATurnConc}{\B}$ in $\ATurnConc$. Assume that for all $\pi \in Q^+$
	, we have: 
	\begin{displaymath}
		\sum_{\substack{\rho \in \inv{\projec{V}{Q}}[\pi] \\ \rho \in (V_\A \cdot V_\B)^* \cdot V_\A}} \prob{\Seq(\Aconc)}{\sigma_\A}{\sigma_\B}(\rho) \cdot \nabla^{\sigma_\A,\sigma_\B}_{\ATurnConc}(\rho) = 	\prob{\Seq(\Aconc)}{\sigma_\A}{\sigma_\B}(\inv{\projec{V}{Q}}[\cyl(\pi)]) \cdot  \nabla^{\s_\A,\s_\B}_\Aconc(\pi)
	\end{displaymath}
	
	Then, for all Borel winning set $W$, we have that $\prob{\Aconc}{\s_\A}{\s_\B}[U_W] = \prob{\Seq(\Aconc)}{\sigma_\A}{\sigma_\B}[U_{\Seq(W)}]$. 
	\label{prop:eq_nabla_eq_prob}
\end{proposition}

\begin{proof}
	
	In the following, $\prob{\Aconc}{\s_\A}{\s_\B}$ and 	$\prob{\Seq(\Aconc)}{\sigma_\A}{\sigma_\B}$ will be denoted $p_\s$ and $p_\sigma$ respectively. Let us show by induction that, for all $\pi \in Q^+$, we have
	$p_\s(\cyl(\pi)) = p_\sigma(\inv{\projec{V}{Q}}[\cyl(\pi)])$. This holds trivially when $\pi = q_0$. Assume now that the property holds for some $\pi' \in Q^+$, that is: $$p_\s(\cyl(\pi')) = p_\sigma(\inv{\projec{V}{Q}}[\cyl(\pi')])$$
	Let $\pi = \pi' \cdot q \in Q^+$ and $\rho \in \inv{\projec{V}{Q}}[\pi]$ such that $\head(\rho) = \head(\pi) = q \in Q = V_\A$. Assume that $p_\sigma(\rho) > 0$. By Proposition~\ref{prop:projection_colors_colors_projection}, it follows that $\rho = \rho' \cdot r \cdot q$ with 
	$r = (\head(\rho'),a_r) \in V_\B$ for some $a_r \in A$ and $\rho' \in (V_\A \cdot V_\B)^* \cdot V_\A$. We have that:
	\begin{align*}
	p_\sigma(\inv{\projec{V}{Q}}[\cyl(\pi)]) & = \sum_{\substack{\rho \in \inv{\projec{V}{Q}}[\pi] \\ \head(\rho) = \head(\pi)}} p_\sigma(\cyl(\rho)) 
	\\
	& = \sum_{\substack{\rho = \rho' \cdot r \cdot q \in \inv{\projec{V}{Q}}[\pi] \\ \rho \in (V_\A \cdot V_\B)^+ \cdot V_\A}} p_\sigma(\rho') \cdot \deltaDistrib^{\sigma_\A,\sigma_\B}(\rho')[r] \cdot \deltaDistrib^{\sigma_\A,\sigma_\B}(\rho' \cdot r)[q] 
	\\
	& = \sum_{\substack{\rho' \in \inv{\projec{V}{Q}}[\pi'] \\ \rho' \in (V_\A \cdot V_\B)^* \cdot V_\A}} p_\sigma(\rho') \cdot \left(\sum_{r \in V_\B} \deltaDistrib^{\sigma_\A,\sigma_\B}(\rho')[r] \cdot \deltaDistrib^{\sigma_\A,\sigma_\B}(\rho' \cdot r)[q] \right) \\
	& = \sum_{\substack{\rho' \in \inv{\projec{V}{Q}}[\pi'] \\ \rho' \in (V_\A \cdot V_\B)^* \cdot V_\A}} p_\sigma(\rho') \cdot \left(\sum_{a_r \in A} \deltaDistrib^{\sigma_\A,\sigma_\B}(\rho')[(\head(\rho'),a_r)] \cdot \deltaDistrib^{\sigma_\A,\sigma_\B}(\rho' \cdot (\head(\rho'),a_r))[q] \right)
	\end{align*}
	
	In addition, for $\rho' \in (V_\A \cdot V_\B)^* \cdot V_\A$ and $a_r \in A$, by definition of $\deltaDistrib_\Aconc$ and since $\head(\rho') \in V_\A$, we have:
	\begin{align*}
	\deltaDistrib^{\sigma_\A,\sigma_\B}(\rho')[(\head(\rho'),a_r)] & = \sum_{a \in \Supp(\sigma_\A(\rho'))} \sum_{b \in \Supp(\sigma_\B(\rho'))} \sigma_\A(\rho')[a] \cdot \sigma_\B(\rho')[b] \cdot \deltaDistrib_\Aconc(\head(\rho'),a,b)[(\head(\rho'),a_r)]\\
	& = \sigma_\A(\rho')[a_r]
	\end{align*}
	
	Furthermore, for $r = (\head(\rho),a_r)$ we also have:
	\begin{align*}
	\deltaDistrib^{\sigma_\A,\sigma_\B}(\rho' \cdot r)[q] & = \sum_{a \in \Supp(\sigma_\A(\rho' \cdot r))} \sum_{b \in \Supp(\sigma_\B(\rho' \cdot r))} \sigma_\A(\rho' \cdot r)[a] \cdot \sigma_\B(\rho' \cdot r)[b] \cdot \deltaDistrib_\Aconc(r,a,b)[q]\\
	& = \sum_{a \in \Supp(\sigma_\A(\rho' \cdot r))} \sum_{b \in \Supp(\sigma_\B(\rho' \cdot r))} \sigma_\A(\rho' \cdot r)[a] \cdot \sigma_\B(\rho' \cdot r)[b] \cdot \deltaDistrib(\head(\rho'),a_r,b)[q]\\
	& = \sum_{b \in \Supp(\sigma_\B(\rho' \cdot r))}\sigma_\B(\rho' \cdot r)[b] \cdot \deltaDistrib(\head(\rho'),a_r,b)[q]\\
	& = \sum_{d \in \distribSet} \sum_{\substack{b \in \Supp(\sigma_\B(\rho' \cdot r)) \\ \delta(\head(\rho',a_r,b)) = d}} \sigma_\B(\rho' \cdot r)[b] \cdot \distribFunc(d)[q]\\
	& = \sum_{d \in \distribSet} \distribFunc(d)[q] \cdot \left( \sum_{\substack{b \in \Supp(\sigma_\B(\rho' \cdot r)) \\ \delta(\head(\rho',a_r,b)) = d}} \sigma_\B(\rho' \cdot r)[b]\right)
	\end{align*}
	
	By combining these three results, it follows that:
	\begin{align*}
	p_\sigma(\inv{\projec{V}{Q}}[\cyl(\pi)]) & = \sum_{\substack{\rho' \in \inv{\projec{V}{Q}}[\pi'] \\ \rho' \in (V_\A \cdot V_\B)^* \cdot V_\A}} p_\sigma(\rho') \cdot \left(\sum_{a_r \in A} \deltaDistrib^{\sigma_\A,\sigma_\B}(\rho')[(\head(\rho'),a_r)] \cdot \deltaDistrib^{\sigma_\A,\sigma_\B}(\rho' \cdot (\head(\rho'),a_r))[q] \right) \\
	& = \sum_{\substack{\rho' \in \inv{\projec{V}{Q}}[\pi'] \\ \rho' \in (V_\A \cdot V_\B)^* \cdot V_\A}} p_\sigma(\rho') \cdot \left(\sum_{a_r \in A} \sigma_\A(\rho')[a_r] \cdot \left( \sum_{d \in \distribSet} \distribFunc(d)[q] \cdot \left( \sum_{\substack{b \in \Supp(\sigma_\B(\rho' \cdot r)) \\ \delta(\head(\rho'),a_r,b) = d}} \sigma_\B(\rho' \cdot r)[b]\right) \right) \right) \\
	& = \sum_{d \in \distribSet} \distribFunc(d)[q] \cdot \left(\sum_{\substack{\rho' \in \inv{\projec{V}{Q}}[\pi'] \\ \rho' \in (V_\A \cdot V_\B)^* \cdot V_\A}} p_\sigma(\rho') \cdot \left(\sum_{a_r \in A} \sigma_\A(\rho')[a_r] \cdot \sum_{\substack{b \in \Supp(\sigma_\B(\rho' \cdot r)) \\ \delta(\head(\rho'),a_r,b) = d}} \sigma_\B(\rho' \cdot r)[b]\right) \right) \\
	& = \sum_{d \in \distribSet} \distribFunc(d)[q] \cdot \left(\sum_{\substack{\rho' \in \inv{\projec{V}{Q}}[\pi'] \\ \rho' \in (V_\A \cdot V_\B)^* \cdot V_\A}} p_\sigma(\rho') \cdot \left(\sum_{\substack{a_r \in \\ \Supp(\sigma_\A(\rho'))}} \sum_{\substack{b \in \Supp(\sigma_\B(\rho' \cdot r)) \\ \delta(\head(\rho'),a_r,b) = d}} \sigma_\A(\rho')[a_r] \cdot \sigma_\B(\rho' \cdot r)[b]\right) \right) \\
	& = \sum_{d \in \distribSet} \distribFunc(d)[q] \cdot \left(\sum_{\substack{\rho' \in \inv{\projec{V}{Q}}[\pi'] \\ \rho' \in (V_\A \cdot V_\B)^* \cdot V_\A}} p_\sigma(\rho') \cdot \nabla^{\sigma_\A,\sigma_\B}_{\ATurnConc}(\rho')[d] \right) 
	\; \; \; \; \; \text{ by definition of } \nabla^{\sigma_\A,\sigma_\B}_{\ATurnConc}\\
	& = \sum_{d \in \distribSet} \distribFunc(d)[q] \cdot \left( p_\sigma(\inv{\projec{V}{Q}}[\cyl(\pi')]) \cdot  \nabla^{\s_\A,\s_\B}_\Aconc(\pi')[d] \right)  
	\; \; \; \; \; \; \; \; \text{ by hypothesis of the proposition }\\
	& = p_\s(\cyl(\pi')) \cdot \left( \sum_{d \in \distribSet} \distribFunc(d)[q] \cdot \nabla^{\s_\A,\s_\B}_\Aconc(\pi')[d] \right)  
	\; \; \; \; \; \; \; \; \; \; \; \; \; \; \; \; \; \; \; \text{ by the induction hypothesis } \\
	& = p_\s(\cyl(\pi')) \cdot 	\deltaDistrib^{\s_\A,\s_\B}(\pi')[q] = p_\s(\pi') \cdot \deltaDistrib^{\s_\A,\s_\B}(\pi')[q]  
	\; \; \; \; \; \; \; \; \; \; \; \; \; \text{ by Remark~\ref{rmq:nabla_delta} } \\
	& = p_\s(\pi) = p_\s(\cyl(\pi))
	\end{align*}
	
	Overall, we have shown that, for all $\pi \in Q^+$, we have $p_\s(\cyl(\pi)) = p_\sigma(\inv{\projec{V}{Q}}[\cyl(\pi)])$. Then, for all Borel winning set $W$, by Proposition~\ref{lem:sufficient_cond_prob_borel} (since $U_W$ is Borel as the preimage of a Borel set by a continuous function), we have: $$p_\s[U_W] = p_\sigma[\inv{\projec{V}{Q}}[U_W]]$$
	
	We want to show that $p_\s[U_W] = p_\sigma[U_{\Seq(W)}]$. For all Borel set $B \subseteq V^\omega$, we have $p_\sigma[B] = p_\sigma[B \cap \OutStrat{\ATurnConc}^\omega]$ (since $p_\sigma[V^\omega \setminus  \OutStrat{\ATurnConc}^\omega] = p_\sigma[\cup_{\pi \in V^+ \setminus \OutStrat{\ATurnConc}^+} \cyl(\pi)] = 0$).
	In addition:
	\begin{align*}
	\inv{\projec{V}{Q}}[U_W] \cap \OutStrat{\ATurnConc}^\omega &= \inv{\projec{V}{Q}}[\inv{\extendFunc{\colFunc}}[W]] \cap \OutStrat{\ATurnConc}^\omega & \text{ by definition of }U_{W}\\
	&= \inv{(\extendFunc{\colFunc} \circ \projec{V}{Q})}[W] \cap \OutStrat{\ATurnConc}^\omega & \\
	& = \inv{(\projec{\colSet_\Aconc}{\colSet} \circ \extendFunc{\colFunc_\Aconc})}[W] \cap \OutStrat{\ATurnConc}^\omega & \text{ by Proposition~\ref{prop:projection_colors_colors_projection}}\\ 
	& = \inv{\extendFunc{\colFunc_\Aconc}}[\inv{\projec{\colSet_\Aconc}{\colSet}}[W]] \cap \OutStrat{\ATurnConc}^\omega & \\ & = \inv{\extendFunc{\colFunc_\Aconc}}[\Seq(W)] \cap \OutStrat{\ATurnConc}^\omega & \text{ by definition of }\Seq(W) \\ & = U_{\Seq(W)} \cap \OutStrat{\ATurnConc}^\omega & \text{ by definition of }U_{\Seq(W)}
	\end{align*}
	
	We can deduce that 
	$p_\s[U_W] = p_\sigma[\inv{\projec{V}{Q}}[U_W]] = p_\sigma[U_{\Seq(W)}]$.
\end{proof}

We can proceed to the proof of Lemma~\ref{lem:sequential_win}.
\begin{proof}[Of Lemma~\ref{lem:sequential_win} for Player $\A$]
	Consider the color strategy $\s_\A: \colSet^* \times Q \rightarrow A \in \SetColStrat{\Aconc}{\A}$ for Player $\A$ in the concurrent arena $\Aconc$. In the rest of the proof, we denote the color strategy $\Seq(\s_\A) \in \SetColStrat{\ATurnConc}{\A}$ by $\sigma_\A$. For a state strategy $\sigma_\B: V^+ \rightarrow \Dist(B) \in \SetStrat{\ATurnConc}{\B}$ for Player $\B$ in the turn-based game $\Seq(\Aconc)$, we exhibit a state strategy $\s_\B: Q^+ \rightarrow \Dist(B) \in \SetStrat{\Aconc}{\B}$ in the concurrent game $\Aconc$ such that $\prob{\Aconc}{\stratMod{\s_\A}}{\s_\B}[U_W] = \prob{\Seq(\Aconc)}{\stratMod{\sigma_\A}}{\sigma_\B}[U_{\Seq(W)}]$. This will show that: 
	$$\forall \sigma_\B \in \SetStrat{\Seq(\Aconc)}{\B},\; \exists \s_\A \in \SetStrat{\Aconc}{\A},\; \prob{\Aconc}{\stratMod{\s_\A}}{\s_\B}[U_W] = \prob{\Seq(\Aconc)}{\stratMod{\sigma_\A}}{\sigma_\B}[U_{\Seq(W)}]$$
	That is, 
	$$\val{\Seq(\Aconc)}{\sigma_\A}[\Seq(W)] = \inf_{\sigma_\B \in \SetStrat{\Seq(\Aconc)}{\B}}  \prob{\Seq(\Aconc)}{\stratMod{\sigma_\A}}{\sigma_\B}[U_{\Seq(W)}] \geq \inf_{\s_\B \in \SetStrat{\Aconc}{\B}} \prob{\Aconc}{\stratMod{\s_\A}}{\s_\B}[U_W] = \val{\Aconc}{\s_\A}[W]$$
	
	
	First, consider some $\rho \in (V_\A \cdot V_\B)^* \cdot V_\A$, then we have:
	\begin{align*}
	\stratMod{\sigma_\A}(\rho) & = \sigma_\A(\extendFunc{\colFunc_\Aconc}(\rho),\head(\rho)) & \text{ by definition of } \stratMod{\sigma_\B}\\	
	&= \Seq(\lambda)(\extendFunc{\Seq(\mu)}(\minit,\extendFunc{\colFunc_\Aconc}(\rho),\head(\rho)) & \text{ by definition of } \sigma_\B\\
	&=  \Seq(\lambda)(\extendFunc{\Seq(\mu)}(\minit,\extendFunc{\colFunc}(\projec{V}{Q}(\rho))_{k_\Aconc}),\head(\rho)) & \text{ by Remark~\ref{rmq:sequence_col_path_ext}}\\
	&= \Seq(\lambda)(\extendFunc{\mu}(\minit,\extendFunc{\colFunc}(\projec{V}{Q}(\rho))),\head(\rho)) & \text{ by Proposition~\ref{prop:projection_memory}}\\
	&= \lambda(\extendFunc{\mu}(\minit,\extendFunc{\colFunc}(\projec{V}{Q}(\rho))),\head(\rho)) & \text{ by definition of }\Seq(\lambda)\\
	&= \s_\A(\extendFunc{\colFunc}(\projec{V}{Q}(\rho)),\head(\rho)) & \text{ by definition of } \s_\B\\
	&= \stratMod{\s_\A}(\projec{V}{Q}(\rho)) & \text{ since } \head(\projec{V}{Q}(\rho)) = \head(\rho)
	\end{align*}
	
	Consider now a state strategy $\sigma_\B: V^+ \rightarrow \Dist(B) \in \SetStrat{\ATurnConc}{\B}$ for Player $\B$. Let us denote by $p_\sigma$ the probability distribution $\prob{\Seq(\Aconc)}{\stratMod{\sigma_\A}}{\sigma_\B}$.	For all $\pi = \pi_0 \cdot \pi_n \in Q^+$, we define the path $\rho^\pi \in (V_\A \cdot V_\B)^* \cdot V_\A$ in the following way:
	\begin{displaymath}
		\rho^\pi = \pi_0 \cdot (\pi_0,\stratMod{\sigma_\A}(\rho^\pi_{\leq 1})) \cdots \pi_{n-1} \cdot (\pi_{n-1},\stratMod{\sigma_\A}(\rho^\pi_{\leq 2n-1})) \cdot \pi_n
	\end{displaymath}
	Since the strategy $\sigma_\A$ is deterministic, for all $\pi \in Q^+$, the path $\rho^\pi \in (V_\A \cdot V_\B)^* \cdot V_\A$ is the only path that can ensure $\projec{V}{Q}(\rho^\pi) = \pi$ and $p_\sigma(\rho^\pi) > 0$. In particular, we have $p_\sigma(\rho^\pi) = p_\sigma(\inv{\projec{V}{Q}}[\cyl(\pi)])$. Now, we define a strategy $\s_\B: Q^+ \rightarrow \Dist(B)$ by, for all $\pi \in Q^+$, we have:
	\begin{displaymath}
		\s_\B(\pi) = \sigma_\B(\rho^\pi \cdot (\head(\rho^\pi),\stratMod{\sigma_\A}(\rho^\pi)))
	\end{displaymath}	
	Now, let $\pi \in Q^+$ and $d \in \distribSet_\B = \distribSet$, we have: 
	\begin{align*}
	\sum_{\substack{\rho \in \inv{\projec{V}{Q}}[\pi] \\ \rho \in (V_\A \cdot V_\B)^* \cdot V_\A}} p_\sigma(\rho) \cdot \nabla^{\sigma_\A,\stratMod{\sigma_\B}}_{\ATurnConc}(\rho)[d] & =
	p_\sigma(\rho^\pi) \cdot \left( \sum_{\substack{(a,b) \in A \times B \\ \delta(\head(\rho),a,b) = d}} \sigma_\A(\rho)[a] \cdot \stratMod{\sigma_\B}(\rho^\pi \cdot (\head(\rho),a))[b] \right) \\ 
	& = p_\sigma(\rho^\pi) \cdot \left( \sum_{\substack{b \in B \\ \delta(\head(\rho),\stratMod{\sigma_\A}(\rho^\pi),b) = d}} \sigma_\B(\rho^\pi \cdot (\head(\rho),\stratMod{\sigma_\A}(\rho^\pi)))[b] \right) \\ 
	& = p_\sigma(\rho^\pi) \cdot \sum_{\substack{b \in B \\ \delta(\head(\pi),\stratMod{\s_\A}(\pi),b) = d}} \s_\B(\pi)[b]  \\ 
	& = p_\sigma(\inv{\projec{V}{Q}}[\cyl(\pi)]) \cdot  \nabla^{\s_\A,\s_\B}_\Aconc(\pi)[d]
	\end{align*}
	As this holds for all $\pi \in Q^+$, we conclude by applying Proposition~\ref{prop:eq_nabla_eq_prob}. That is, we obtain $	\val{\Seq(\Aconc)}{\sigma_\A}[\Seq(W)] \geq \val{\Aconc}{\s_\A}[W]$, from which we get $\val{\Seq(\Aconc)}{\A}[\Seq(W)] \geq \val{\Aconc}{\A}[W]$.
\end{proof}

\begin{proof}[Of Lemma~\ref{lem:sequential_win} for Player $\B$]
	Consider the color strategy $\s_\B: \colSet^* \times Q \rightarrow B \in \SetColStrat{\Aconc}{\B}$ for Player $\B$ in the concurrent arena $\Aconc$. In the rest of the proof, we denote the color strategy $\Seq(\s_\B) \in \SetColStrat{\ATurnConc}{\B}$ by $\sigma_\B$. For a state strategy $\sigma_\A: V^+ \rightarrow \Dist(A) \in \SetStrat{\ATurnConc}{\A}$ for Player $\A$ in the turn-based game $\Seq(\Aconc)$, we exhibit a state strategy $\s_\A: Q^+ \rightarrow \Dist(A) \in \SetStrat{\Aconc}{\A}$ in the concurrent game $\Aconc$ such that $\prob{\Aconc}{\s_\A}{\stratMod{\s_\B}}[U_W] = \prob{\Seq(\Aconc)}{\sigma_\A}{\stratMod{\sigma_\B}}[U_{\Seq(W)}]$. This will show that: 
	$$\forall \sigma_\A \in \SetStrat{\Seq(\Aconc)}{\A},\; \exists \s_\A \in \SetStrat{\Aconc}{\A},\; \prob{\Aconc}{\s_\A}{\stratMod{\s_\B}}[U_W] = \prob{\Seq(\Aconc)}{\sigma_\A}{\stratMod{\sigma_\B}}[U_{\Seq(W)}]$$
	That is, 
	$$\val{\Seq(\Aconc)}{\sigma_\B}[\Seq(W)] = \sup_{\sigma_\A \in \SetStrat{\Seq(\Aconc)}{\A}}  \prob{\Seq(\Aconc)}{\sigma_\A}{\stratMod{\sigma_\B}}[U_{\Seq(W)}] \leq \sup_{\s_\A \in \SetStrat{\Aconc}{\A}} \prob{\Aconc}{\s_\A}{\stratMod{\s_\B}}[U_W] = \val{\Aconc}{\s_\B}[W]$$
	
	
	First, consider some $\rho \in (V_\A \cdot V_\B)^* \cdot V_\A$ and $r = (\head(\rho),a) \in V_\B$, we have:
	\begin{align*}
	\stratMod{\sigma_\B}(\rho \cdot r) & = \sigma_\B(\extendFunc{\colFunc_\Aconc}(\rho \cdot r),r) & \text{ by definition of } \stratMod{\sigma_\B}\\	
	&= \Seq(\lambda)(\extendFunc{\Seq(\mu)}(\minit,\extendFunc{\colFunc_\Aconc}(\rho \cdot r),r) & \text{ by definition of } \sigma_\B\\
	&=  \Seq(\lambda)(\extendFunc{\Seq(\mu)}(\minit,\extendFunc{\colFunc}(\projec{V}{Q}(\rho))_{k_\Aconc} \cdot k_\Aconc),r) & \text{ by Remark~\ref{rmq:sequence_col_path_ext}}\\
	&= \Seq(\lambda)(\extendFunc{\mu}(\minit,\extendFunc{\colFunc}(\projec{V}{Q}(\rho))),r) & \text{ by Proposition~\ref{prop:projection_memory}}\\
	&= \lambda(\extendFunc{\mu}(\minit,\extendFunc{\colFunc}(\projec{V}{Q}(\rho))),\head(\rho)) & \text{ by definition of }\Seq(\lambda)\\
	&= \s_\B(\extendFunc{\colFunc}(\projec{V}{Q}(\rho)),\head(\rho)) & \text{ by definition of } \s_\B\\
	&= \stratMod{\s_\B}(\projec{V}{Q}(\rho)) & \text{ since } \head(\projec{V}{Q}(\rho)) = \head(\rho)
	\end{align*}
	
	Consider now a state strategy $\sigma_\A: V^+ \rightarrow \Dist(A) \in \SetStrat{\ATurnConc}{\A}$ for Player $\A$. Let us denote by $p_\sigma$ the probability distribution $\prob{\Seq(\Aconc)}{\sigma_\A}{\stratMod{\sigma_\B}}$.	We define a strategy $\s_\A: Q^+ \rightarrow \Dist(A)$ by, for all $\pi \in Q^+$, and $a \in A$, we have:
	\begin{displaymath}
		\s_\A(\pi)[a] = \frac{1}{p_\sigma(\inv{\projec{V}{Q}}[\cyl(\pi)])} \cdot \sum_{\substack{\rho \in \inv{\projec{V}{Q}}[\pi] \\ \rho \in (V_\A \cdot V_\B)^* \cdot V_\A}} p_\sigma(\rho) \cdot \sigma_a(\rho)[a]
	\end{displaymath}	
	Note that the support of the strategy $\s_\A$ is countable since the set of paths $\rho \in \inv{\projec{V}{Q}}[\pi]$ such that $p_\sigma(\rho) > 0$ is countable. We also denote by $p_\s$ the probability distribution $\prob{\Aconc}{\s_\A}{\stratMod{\s_\B}}$. Now, let $\pi \in Q^+$ and $d \in \distribSet_\B = \distribSet$, we have: 
	\begin{align*}
		\sum_{\substack{\rho \in \inv{\projec{V}{Q}}[\pi] \\ \rho \in (V_\A \cdot V_\B)^* \cdot V_\A}} p_\sigma(\rho) \cdot \nabla^{\sigma_\A,\stratMod{\sigma_\B}}_{\ATurnConc}(\rho)[d] & =
		\sum_{\substack{\rho \in \inv{\projec{V}{Q}}[\pi] \\ \rho \in (V_\A \cdot V_\B)^* \cdot V_\A}} p_\sigma(\rho) \cdot \left( \sum_{\substack{(a,b) \in A \times B \\ \delta(\head(\rho),a,b) = d}} \sigma_\A(\rho)[a] \cdot \stratMod{\sigma_\B}(\rho \cdot (\head(\rho),a))[b] \right) \\ 
		& =
		\sum_{\substack{\rho \in \inv{\projec{V}{Q}}[\pi] \\ \rho \in (V_\A \cdot V_\B)^* \cdot V_\A}} p_\sigma(\rho) \cdot \left( \sum_{\substack{(a,b) \in A \times B \\ \delta(\head(\rho),a,b) = d}} \sigma_\A(\rho)[a] \cdot \stratMod{\s_\B}(\pi)[b] \right) \\ 
		& = \sum_{\substack{(a,b) \in A \times B \\ \delta(\head(\rho),a,b) = d}} \stratMod{\s_\B}(\pi)[b] \cdot \left( \sum_{\substack{\rho \in \inv{\projec{V}{Q}}[\pi] \\ \rho \in (V_\A \cdot V_\B)^* \cdot V_\A}} p_\sigma(\rho) \cdot \sigma_\A(\rho)[a] \right) \\ 
		& = \sum_{\substack{(a,b) \in A \times B \\ \delta(\head(\rho),a,b) = d}} \stratMod{\s_\B}(\pi)[b] \cdot  p_\sigma(\inv{\projec{V}{Q}}[\cyl(\pi)]) \cdot \s_\A(\pi)[a] \\ 
		& = p_\sigma(\inv{\projec{V}{Q}}[\cyl(\pi)]) \cdot  \nabla^{\s_\A,\s_\B}_\Aconc(\pi)[d]
	\end{align*}
	As this holds for all $\pi \in Q^+$, we conclude by applying Proposition~\ref{prop:eq_nabla_eq_prob}. That is, we obtain $	\val{\Seq(\Aconc)}{\sigma_\B}[\Seq(W)] \leq \val{\Aconc}{\s_\B}[W]$, from which we get $\val{\Seq(\Aconc)}{\B}[\Seq(W)] \leq \val{\Aconc}{\B}[W]$.
\end{proof}

\subsection{Proposition regarding the sequentialization of the update function}
\label{proof:Proposition5}
The definition of the parallelization of the update function ensures the following proposition.
\begin{proposition}[Mentioned~\ref{ref:parallelization_update}]
	For all non-empty set $M$, $\mu: M \times \colSetSeq \rightarrow M$, $m \in 
	M$, and $\rho \in \colSet^*$, the state of memory w.r.t. the update function $\Par(\mu)$ after the colors in $\rho$ have been seen is the 
	same as after the colors $\rho_{k_\Aconc}$ have been seen w.r.t. the update function $\mu$: $\extendFunc{\Par(\mu)}(m,\rho) = 
	\extendFunc{\mu}(m,\rho_{k_\Aconc})$. (Recall that 
	$\rho_{k_\Aconc}$ is defined in Definition~\ref{def:rho_kc}).
	\label{prop:projection_memory_parallel}
\end{proposition}
\begin{proof}
	Consider a non-empty set $M$, $\mu: M \times \colSetSeq \rightarrow M$, and $m \in M$. We proceed by induction on $\rho \in \colSet^*$. We have 
	$\extendFunc{\Par(\mu)}(m,\epsilon) = m =
	\extendFunc{\mu}(m,\epsilon) = \extendFunc{\mu}(m,\epsilon_{k_\Aconc})$. Now, let $\rho = \rho' \cdot k \in \colSet^+$ and assume that 
	$\extendFunc{\Par(\mu)}(m,\rho') = m' =
	\extendFunc{\mu}(m,\rho'_{k_\Aconc})$ for some $m' \in M$. Then, we have:
	\begin{align*}
		\extendFunc{\mu}(m,\rho_{k_\Aconc}) & = \extendFunc{\mu}(m,\rho'_{k_\Aconc} \cdot k_{\Aconc} \cdot k) & \text{ by definition of } \rho_{k_\Aconc}\\
		& = \extendFunc{\mu}(m',k_{\Aconc} \cdot k) & \text{ since } \extendFunc{\mu}(m,\rho'_{k_\Aconc}) = m'\\
		& = \mu(\mu(m',k_\Aconc),k) & \text{ by definition of } \extendFunc{\mu}\\
		& = \Par(\mu)(m',k) & \text{ by definition of } \Par(\mu) \\
		& = \Par(\mu)(\extendFunc{\Par(\mu)}(m,\rho'),k) & \text{ since } \extendFunc{\Par(\mu)}(m,\rho') = m' \\
		& = \extendFunc{\Par(\mu)}(m,\rho' \cdot k)&  \text{ by definition of }\extendFunc{\Par(\mu)}\\
		& = \extendFunc{\Par(\mu)}(m,\rho) & \text{ since } \rho = \rho' \cdot k
	\end{align*}
\end{proof}

\subsection{Proof that the parallelization of strategies is well defined}
\label{proof:remark6}
As for the case of sequentialization, let us show that the parallelization of strategies is well defined.
\begin{proof}[In the case of Player $\A$][Mentioned~\ref{ref:parallelization}]
	Consider a strategy $\sigma_\A: \colSet^* \times V \rightarrow A$ in the 
	turn-based game $\ATurnConc$ implemented by the memory skeleton 
	$\mathcal{M} = \langle M,\minit,\mu \rangle$ and the action function 
	$\lambda: M \times V \rightarrow A$. Let $\s: \colSet^* \times Q 
	\rightarrow A$ be the strategy implemented the memory 
	skeleton $\Par(\mathcal{M})$ and the action map $\Par(\lambda)$. 
	
	Let $\rho \in \colSet^*$ and $q \in Q = V_\A \in V$. We have:
	\begin{align*}
		\s(\rho,q) & = \Par(\lambda)(\extendFunc{\Par(\mu)}(\minit,\rho),q) & \text{ by definition of } \s\\
		&= \lambda(\extendFunc{\Par(\mu)}(\minit,\rho),q) & \text{ by definition of } \Par(\lambda) \\
		&= \lambda(\extendFunc{\mu}(\minit,\rho_{k_\Aconc}),q) & \text{ by Proposition~\ref{prop:projection_memory_parallel} }\\
		&= \sigma_\A(\rho_{k_\Aconc},q) & \text{ by definition of } \sigma_\A
	\end{align*}
	Henceforth, the value of $\s$ does not depend on the pair of memory skeleton and action function considered that implements the strategy $\sigma_\A$.
\end{proof}
\begin{proof}[In the case of Player $\B$][Mentioned~\ref{ref:parallelization}]
	Consider a strategy $\sigma: \colSet^* \times V \rightarrow B$ in the 
	turn-based game $\ATurnConc$ implemented by the memory skeleton 
	$\mathcal{M} = \langle M,\minit,\mu \rangle$ and the action function 
	$\lambda: M \times  V \rightarrow B$. Let $\s: \colSet^* \times Q 
	\rightarrow B$ be the strategy implemented the memory 
	skeleton $\Par(\mathcal{M})$ and the action map $\Par(\lambda)$. 
	
	Let $\rho \in \colSet^*$ and $q \in Q$. We have: $$\s(\rho,q) = 
	\Par(\lambda)(\extendFunc{\Par(\mu)}(\minit,\rho),q) = \min_{<_\B} 	\reachStrat{\B}(\formNF_q,\Rech{\mu}{\lambda}{m}{q})$$ for $m = \extendFunc{\Par(\mu)}(\minit,\rho)$.
	Furthermore, for an action $a \in A$, we have:
	$$\rech{\mu}{\lambda}{m}{q}(a) = \delta(q,a,\lambda(\mu(m,k_\Aconc),(q,\Sq{q}{a})))$$
	
	In addition, we have:
	\begin{align*}
		\lambda(\mu(m,k_\Aconc),(q,\Sq{q}{a})) &= \lambda(\mu(\extendFunc{\Par(\mu)}(\minit,\rho),k_\Aconc),(q,\Sq{q}{a}))& \text{ by definition of }m\\
		&= \lambda(\mu(\extendFunc{\mu}(\minit,\rho_{k_\Aconc}),k_\Aconc),(q,\Sq{q}{a})) & \text{ by Proposition~\ref{prop:projection_memory_parallel} }\\ & = \lambda(\extendFunc{\mu}(\minit,\rho_{k_\Aconc} \cdot k_\Aconc),(q,\Sq{q}{a})) & \text{ by definition of }\extendFunc{\mu}\\
		&= \sigma(\rho_{k_\Aconc} \cdot k_\Aconc,(q,\Sq{q}{a})) & \text{ by definition of } \sigma
	\end{align*}
	It follows that: $$\rech{\mu}{\lambda}{m}{q}(a) = \delta(q,a,\sigma(\rho_{k_\Aconc} \cdot k_\Aconc,(q,\Sq{q}{a})))$$
	
	Note that, this does not depend on $\lambda$ nor $\mu$. Hence, this is also the case for $\Rech{\mu}{\lambda}{m}{q}$ and $\reachStrat{\B}(\formNF_q,\Rech{\mu}{\lambda}{m}{q})$. Therefore $\s$ does not depend on the pair of memory skeleton and action function implementing $\sigma$.
\end{proof}

\subsection{Proof of Lemma~\ref{lem:parallel_win_A}}
First, let us state a corollary to Proposition~\ref{prop:eq_nabla_eq_prob}.
\begin{corollary}
	\label{coro:equal_nabla}
	Let $\Aconc = \AcoloredConc$ be a concurrent arena and $\ATurnConc = \AcoloredSeq$ be its sequential version. Consider a pair of state strategies $(\s_\A,\s_\B) \in \SetStrat{\Aconc}{\A} \times  \SetStrat{\Aconc}{\B}$ in $\Aconc$ and a pair of state strategies $(\sigma_\A,\sigma_\B) \in \SetStrat{\ATurnConc}{\A} \times  \SetStrat{\ATurnConc}{\B}$ in $\ATurnConc$. Assume that for all $\rho \in (V_\A \cdot V_\B)^* \cdot V_\A$
	, we have: 
	\begin{displaymath}
	\nabla^{\sigma_\A,\sigma_\B}_{\ATurnConc}(\rho) = \nabla^{\s_\A,\s_\B}_\Aconc(\projec{V}{Q}(\rho))
	\end{displaymath}
	
	Then, for all Borel winning set $W$, we have that $\prob{\Aconc}{\s_\A}{\s_\B}[U_W] = \prob{\Seq(\Aconc)}{\sigma_\A}{\sigma_\B}[U_{\Seq(W)}]$. 
\end{corollary}
\begin{proof}
	Let $\pi \in Q^+$, we have: 
	\begin{align*}
		\sum_{\substack{\rho \in \inv{\projec{V}{Q}}[\pi] \\ \rho \in (V_\A \cdot V_\B)^* \cdot V_\A}} \prob{\Seq(\Aconc)}{\sigma_\A}{\sigma_\B}(\rho) \cdot \nabla^{\sigma_\A,\sigma_\B}_{\ATurnConc}(\rho) & = \sum_{\substack{\rho \in \inv{\projec{V}{Q}}[\pi] \\ \rho \in (V_\A \cdot V_\B)^* \cdot V_\A}} \prob{\Seq(\Aconc)}{\sigma_\A}{\sigma_\B}(\rho) \cdot \nabla^{\s_\A,\s_\B}_\Aconc(\pi) & \text{ by hypothesis }\\
		& = 	\prob{\Seq(\Aconc)}{\sigma_\A}{\sigma_\B}(\inv{\projec{V}{Q}}[\cyl(\pi)]) \cdot  \nabla^{\s_\A,\s_\B}_\Aconc(\pi) & \text{ by Remark~\ref{rmq:reverse_projec_cylinder} and Prposition~\ref{prop:projection_colors_colors_projection}}\\
	\end{align*}
	We conclude by applying Proposition~\ref{prop:eq_nabla_eq_prob}.
\end{proof}

We can now proceed to the proof of Lemma~\ref{lem:parallel_win_A}.
\label{proof:lemma2}
\begin{proof}[Of Lemma~\ref{lem:parallel_win_A}]
	Consider the color strategy $\sigma_\A: \colSet^* \times V \rightarrow A$ for Player $\A$ in the turn-based arena $\ATurnConc$. In the following of the proof, we denote $\Par(\sigma_\A)$ by $\s_\A$. 
	For a state strategy $\s_\B \in \SetStrat{\Aconc}{B}$ for Player $\B$ in the concurrent game $\Aconc$, we exhibit a strategy $\sigma_\B \in \SetStrat{\Seq(\Aconc)}{B}$ in its sequential version $\Seq(\Aconc)$ such that $\prob{\Aconc}{\s_\A}{\s_\B}[U_W] = \prob{\Seq(\Aconc)}{\sigma_\A}{\sigma_\B}[U_{\Seq(W)}]$. This will show that: 
	$$\forall \s_\B \in \SetStrat{\Aconc}{\B},\; \exists \sigma_\B \in \SetStrat{\Seq(\Aconc)}{\B},\; \prob{\Aconc}{\s_\A}{\s_\B}[U_W] = \prob{\Seq(\Aconc)}{\sigma_\A}{\sigma_\B}[U_{\Seq(W)}]$$
	From which we can conclude that:
	$$\val{\Aconc}{\s_\A}[W] = \inf_{\s_\B \in \SetStrat{\Aconc}{\B}} \prob{\Aconc}{\s_\A}{\s_\B}[U_W] \geq \inf_{\sigma_\B \in \SetStrat{\Seq(\Aconc)}{\B}}  \prob{\Seq(\Aconc)}{\sigma_\A}{\sigma_\B}[U_{\Seq(W)}] = \val{\Seq(\Aconc)}{\sigma_\A}[\Seq(W)]$$
	
	We want to apply Corollary~\ref{coro:equal_nabla}. 
	Let us assume that the strategy $\sigma_\A$ is implemented by a memory skeleton $\mathcal{M} = \langle M,\minit,\mu \rangle$ and an action map $\lambda$. The strategy $\s_\A = \Par(\sigma_\A)$ is implemented by the memory skeleton $\Par(\mathcal{M})$ and the action map $\Par(\lambda)$. For $\rho \in (V_\A \cdot V_\B)^* \cdot V_\A$, we have the following $(1)$:
	\begin{align*}
		\stratMod{\sigma_\A}(\rho) & = \sigma_\A(\extendFunc{\colFunc_\Aconc}(\rho),q) & \text{ by definition of } \stratMod{\sigma_\A} \\
		& = \lambda(\extendFunc{\mu}(\minit,\extendFunc{\colFunc_\Aconc}(\rho)),q) & \text{ by definition of } \sigma_\A \\
		& = \lambda(\extendFunc{\mu}(\minit,\extendFunc{\colFunc}(\projec{V}{Q}(\rho)_{k_\Aconc})),q) & \text{ by Remark~\ref{rmq:sequence_col_path_ext}} \\
		& = \Par(\lambda)(\extendFunc{\Par(\mu)}(\minit,\extendFunc{\colFunc}(\projec{V}{Q}(\rho))),q) &  \text{ by Proposition~\ref{prop:projection_memory_parallel}} \\
		& = \stratMod{\s_\A}(\projec{V}{Q}(\rho)) & \text{ by definition of } \stratMod{\s_\A}
	\end{align*}
	
	That is, for all $\rho \in (V_\A \cdot V_\B)^* \cdot V_\A$, we have $\stratMod{\s_\A}(\projec{V}{Q}(\rho)) = \stratMod{\sigma_\A}(\rho)$. Consider now a state strategy $\s_\B: Q^+ \rightarrow \Dist(B)$ for Player $\B$ in the concurrent game $\Aconc$. We consider a strategy $\sigma_\B: V^+ \rightarrow \Dist(B)$ for Player $\B$ in the turn-based game $\ATurnConc$ such that: for all $\rho \in (V_\A \cdot V_\B)^* \cdot V_\A$ and $r \in V_\B$, we have $\sigma_\B(\rho \cdot r) = \s_\B(\projec{V}{Q}(\rho))$. Now, consider $\rho \in (V_\A \cdot V_\B)^* \cdot V_\A$ and $d \in \distribSet = \distribSet_\B$, we have:
	\begin{align*}
		\nabla^{\stratMod{\sigma_\A},\sigma_\B}_{\ATurnConc}(\rho)[d] & = \sum_{\substack{(a,b) \in A \times B \\ \delta(t,a,b) = d}} \stratMod{\sigma_\A}(\rho)[a] \cdot \sigma_\B(\rho \cdot (t,\Sq{t}{a}))[b] & \text{ by definition of }\nabla^{\stratMod{\sigma_\A},\sigma_\B}_{\ATurnConc}\\
		& = \sum_{\substack{(a,b) \in A \times B \\ \delta(t,a,b) = d}} \stratMod{\s_\A}(\projec{V}{Q}(\rho))[a] \cdot \sigma_\B(\rho \cdot (t,\Sq{t}{a}))[b] & \text{ by }(1)\\
		& = \sum_{\substack{(a,b) \in A \times B \\ \delta(t,a,b) = d}} \stratMod{\s_\A}(\projec{V}{Q}(\rho))[a] \cdot \s_\B(\projec{V}{Q}(\rho))[b] & \text{ by definition of }\sigma_\B\\
		& = \nabla^{\stratMod{\s_\A},\s_\B}_{\ATurnConc}(\projec{V}{Q}(\rho))[d] &\text{ by definition of }\nabla^{\stratMod{\s_\A},\s_\B}_{\ATurnConc}
	\end{align*}

	Overall, we get $\nabla^{\stratMod{\sigma_\A},\sigma_\B}_{\ATurnConc}(\rho) = \nabla^{\stratMod{\s_\A},\s_\B}_{\ATurnConc}(\projec{V}{Q}(\rho))$ for all $\rho \in (V_\A \cdot V_\B)^* \cdot V_\A$. We conclude with Corollary~\ref{coro:equal_nabla}. As for the proof of Lemma~\ref{lem:sequential_win}, we then obtain that $\val{\Aconc}{\A}[W] \geq \val{\Seq(\Aconc)}{\A}[\Seq(W)]$.
\end{proof}

\subsection{Proof of Lemma~\ref{lem:parallel_win_B}}
\label{proof:lemma3}
\begin{proof}
	Consider the color strategy $\sigma_\B: \colSet^* \times V \rightarrow B$ for Player $\B$ in the turn-based game $\ATurnConc$. In the following of the proof, we denote $\Par(\sigma_\B)$ by $\s_\B$. We proceed similarly to the proof of Lemma~\ref{lem:parallel_win_A}: for a state strategy $\s_\A \in \SetStrat{\Aconc}{\A}$ for Player $\A$ in the concurrent game $\Aconc$, we exhibit a strategy $\sigma_\A \in \SetStrat{\Seq(\Aconc)}{\A}$ in the turn-based game $\Seq(\Aconc)$ such that $\prob{\Aconc}{\s_\A}{\s_\B}[U_W] = \prob{\Seq(\Aconc)}{\sigma_\A}{\sigma_\B}[U_{\Seq(W)}]$. This will show that: 
	$$\forall \s_\A \in \SetStrat{\Aconc}{\A},\; \exists \sigma_\A \in \SetStrat{\Seq(\Aconc)}{\A},\; \prob{\Aconc}{\s_\A}{\s_\B}[U_W] = \prob{\Seq(\Aconc)}{\sigma_\A}{\sigma_\B}[U_{\Seq(W)}]$$
	Henceforth:
	$$\val{\Aconc}{\s_\B}[W] = \sup_{\s_\A \in \SetStrat{\Aconc}{\A}} \prob{\Aconc}{\s_\A}{\s_\B}[U_W] \leq \sup_{\sigma_\A \in \SetStrat{\Seq(\Aconc)}{\A}}  \prob{\Seq(\Aconc)}{\sigma_\A}{\sigma_\B}[U_{\Seq(W)}] = \val{\Seq(\Aconc)}{\sigma_\B}[\Seq(W)]$$
	
	Let us assume that the strategy $\sigma_\B$ is implemented by a memory skeleton $\mathcal{M} = \langle M,\minit,\mu \rangle$ and an action map $\lambda$. The strategy $\s_\B = \Par(\sigma_\B)$ is implemented by the memory skeleton $\Par(\mathcal{M})$ and the action map $\Par(\lambda)$. Consider now a state strategy $\s_\A: Q^+ \rightarrow \Dist(A) \in \SetStrat{\Aconc}{\A}$ for Player $\A$ in the concurrent game $\Aconc$. We want to define a strategy $\sigma_\A: V^+ \rightarrow \Dist(A)$ for Player $\A$ in the turn-based game $\ATurnConc$. Consider some $\rho \in (V_\A \cdot V_\B)^*\cdot V_\A$ with $\head(\rho) = t \in Q = V_\A$.
	By definition of $\s_\B$ and $\Par(\lambda)$, we have:
	$$\stratMod{\s_\B}(\projec{V}{Q}(\rho)) = \s_\B(\extendFunc{\colFunc_\Aconc}(\projec{V}{Q}(\rho)),t) = \Par(\lambda)(m,t) \in \reachStrat{\B}(\formNF_t,\Rech{\mu}{\lambda}{m}{t})$$ for $m = \extendFunc{\Par(\mu)}(\minit,\extendFunc{\colFunc_\Aconc}(\projec{V}{Q}(\rho))) = \extendFunc{\mu}(\minit,\extendFunc{\colFunc_\Aconc}(\rho))$ 
	by Remark~\ref{rmq:sequence_col_path_ext} and Proposition~\ref{prop:projection_memory_parallel}.
	By definition of $\reachStrat{\B}$
	, we have: $$\delta(t,a,\stratMod{\s_\B}(\projec{V}{Q})) \in \Rech{\mu}{\lambda}{m}{t} = \rech{\mu}{\lambda}{m}{t}[A]$$
	This allows us to define a function $f_{\rho}: A \rightarrow A$ as follows, for all $a \in A$:
	\begin{displaymath}
		f_\rho(a) = \min_{<_A} \inv{\rech{\mu}{\lambda}{m}{t}}[\delta(t,a,\stratMod{\s_\B}(\rho))]
	\end{displaymath}
	Note that, for all $a \in A$, we have:
	\begin{equation}
		\label{eqn:local_det}
		\delta(t,a,\stratMod{\s_\B}(\rho)) = \rech{\mu}{\lambda}{m}{t}(f_\rho(a)) = \delta(t,f_\rho(a),\lambda(\mu(m,k_\Aconc),(t,f_\rho(a)))) = \delta(t,f_\rho(a),\stratMod{\sigma_\B}(\rho \cdot (t,f_\rho(a))))
	\end{equation}
	since $m = \extendFunc{\mu}(\minit,\extendFunc{\colFunc_\Aconc}(\rho))$ and $\colFunc_\Aconc(t,f_\rho(a)) = k_\Aconc$. Then, for all $a \in A$, we set:
	$$\sigma_\A(\rho)[a] = 
	\sum_{\alpha \in \inv{f_\rho}[a]} \s_\A(\projec{V}{Q}(\rho))[\alpha]$$
	The strategy $\sigma_\A$ is defined arbitrarily on other sequence of states $\rho$. Now, consider some $\rho \in (V_\A \cdot V_\B)^* \cdot V_\A$ with $t = \head(\rho)$ and $d \in \distribSet = \distribSet_\B$, we have:
	\begin{displaymath}
		\nabla^{\sigma_\A,\stratMod{\sigma_\B}}_{\ATurnConc}(\rho)[d] = \sum_{\substack{(a,b) \in A \times B \\ \delta(t,a,b) = d}} \sigma_\A(\rho)[a] \cdot \stratMod{\sigma_\B}(\rho \cdot (t,a))[b] = \sum_{\substack{a \in A \\ \delta(t,a,\stratMod{\sigma_\B}(\rho \cdot (t,a))) = d}} \sigma_\A(\rho)[a]
	\end{displaymath}
	since the color strategy $\sigma_\B: \colSet^* \times V \rightarrow B$ is deterministic. Note that we have the following disjoint union $A = \uplus_{a \in A} \inv{f_\rho}[a]$ (denoted $(u)$). It follows that:	
	\begin{align*}
		\nabla^{\s_\A,\stratMod{\s_\B}}_{\Aconc}(\projec{V}{Q}(\rho))[d] & = \sum_{\substack{\alpha \in A \\ \delta(t,\alpha,\stratMod{\s_\B}(\projec{V}{Q}(\rho))) = d}} \s_\A(\projec{V}{Q}(\rho))[\alpha] & \text{ as previously } \\
		& = \sum_{a \in A} \sum_{\substack{\alpha \in \inv{f_\rho}[a]\\ \delta(t,\alpha,\stratMod{\s_\B}(\projec{V}{Q}(\rho))) = d}} \s_\A(\projec{V}{Q}(\rho))[\alpha] & \text{ by (u) }\\
		& = \sum_{a \in A} \sum_{\substack{\alpha \in \inv{f_\rho}[a]\\ \delta(t,f_\rho(\alpha),\stratMod{\sigma_\B}(\rho \cdot (t,f_\rho(\alpha)))) = d}} \s_\A(\projec{V}{Q}(\rho))[\alpha] & \text{ by (\ref{eqn:local_det})} \\
		& = \sum_{\substack{a \in A \\ \delta(t,a,\stratMod{\sigma_\B}(\rho \cdot (t,a))) = d}} \left( \sum_{\alpha \in \inv{f_\rho}[a]} \s_\A(\projec{V}{Q}(\rho))[\alpha] \right) & \text{ since }f_{\rho}(\alpha) = a\\
		& = \sum_{\substack{a \in A \\ \delta(t,a,\stratMod{\sigma_\B}(\rho \cdot (t,a))) = d}} \sigma_\A(\rho)[a] & \text{ by definition of  } \sigma_\A(\rho)\\
		& = \nabla^{\sigma_\A,\stratMod{\sigma_\B}}_{\ATurnConc}(\rho)[d] & \text{ by definition of }\nabla^{\sigma_\A,\stratMod{\sigma_\B}}_{\ATurnConc}
	\end{align*}	
	
	Overall, for all $\rho \in (V_\A \cdot V_\B)^* \cdot V_\A$, we have $\nabla^{\s_\A,\stratMod{\s_\B}}_{\Aconc}(\projec{V}{Q}(\rho))[d] = \nabla^{\sigma_\A,\stratMod{\sigma_\B}}_{\ATurnConc}(\rho)[d]$. We conclude by applying Corollary~\ref{coro:equal_nabla}. Finally, since $\val{\Aconc}{\B}[W] = \inf_{\s_\B \in \SetSt{\Aconc}{\B}} \val{\Aconc}{\s_\B}[W]$ and $\val{\Seq(\Aconc)}{\B}[\Seq(W)] = \inf_{\sigma_\B \in \SetSt{\Seq(\Aconc)}{\B}} \val{\Seq(\Aconc)}{\sigma_\B}[\Seq(W)]$, it follows that $\val{\Aconc}{\B}[W] \leq \val{\Seq(\Aconc)}{\B}[\Seq(W)]$.
\end{proof}

\section{Complement on Section~\ref{sec:applications}}
\subsection{Definition of a deterministic concurrent arena}
\begin{definition}[Mentioned~\ref{ref:deterministic}]
	A concurrent arena $\Aconc = \AcoloredConc$ is \emph{deterministic} if, for all $d \in \distribSet$, there exists $q \in Q$ such that $\distribFunc(d)[q] = 1$.	
	the probability distribution $\distribFunc(d)$ is Dirac.
	\label{def:deterministic_game}
\end{definition}

In the deterministic case of Definition~\ref{def:deterministic_game}, we define the function $\deltaDirac: Q \times A \times B \rightarrow Q$ defined by $\deltaDirac(q,a,b) = q'$ where $q' \in Q$ is the state ensuring $\deltaDistrib(q,a,b)(q') = 1$. Furthermore, in the specification of $\Aconc$, the components $\distribSet$ and $\distribFunc$ are omitted. 

\subsection{Definition of compatible paths and winning strategies}
\begin{definition}[Mentioned~\ref{ref:compatible}]
	\label{def:compatible_paths}
	Consider a deterministic concurrent arena $\Aconc = \AdetConc$. Let $\s_\A \in \SetColStrat{\Aconc}{\A}$ be a color strategy for Player $\A$. 
	An infinite path $\rho \in Q^\omega$ is \emph{compatible} with the strategy $\s_\A$ if there exists a state strategy $\s_\B \in \SetStrat{\Aconc}{\B}$ such that, for all $i \geq 0$, we have $\prob{\Aconc}{\s_\A}{\s_\B}(\pi_{\leq i}) > 0$. We denote by $\OutStrat{\Aconc}^{\s_\A} \subseteq Q^\omega$ the set of infinite paths compatible with the strategy $\s_\A$. The definition is analogous for a strategy for Player $\B$.
	
	Then, a strategy $\s_\A \in \SetSt{\Aconc}{\A}$ (resp. $\s_\B \in \SetSt{\Aconc}{\B}$) for Player $\A$ (resp. Player $\B$) is \emph{winning} for Player $\A$ (resp. $\B$) if $\OutStrat{\Aconc}^{\s_\A} \subseteq U_W$ (resp. $\OutStrat{\Aconc}^{\s_\A} \subseteq Q^\omega \setminus U_W$).
\end{definition}

\subsection{Proposition regarding winning strategies in a deterministic setting}
\label{proof:lemma_winning_proba_one}
\begin{proposition}[Mentioned~\ref{ref:prop_det_value}]
	\label{lem:winning_proba_one}
	Consider a deterministic concurrent game $\langle \Aconc,W \rangle$, a strategy $\s$ for Player $\A$ (resp. $\B$). Then, the following assertions are equivalent:
	\begin{itemize}
		\item the strategy $\s$ is winning for Player $\A$ (resp. $\B$);
		\item $\val{\Aconc}{\s}[W] = 1$ (resp. $\val{\Aconc}{\s}[W] = 0$).
		\item $\val{\Aconc}{\s}[W] > 0$ (resp. $\val{\Aconc}{\s}[W] < 1$). 
	\end{itemize}
\end{proposition}
\begin{proof}
	We prove the lemma for Player $\A$, the case of Player $\B$ is analogous. Consider a deterministic concurrent game $\langle \Aconc,W \rangle$ with $\Aconc = \AdetConc$ and consider a color strategy $\s_\A \in \SetColStrat{\Aconc}{\A}$ for Player $\A$. Consider an infinite path $\rho \in \OutStrat{\Aconc}^{\s_\A}$ that is compatible with strategy $\s_\A$. Let us show that there exists a strategy $\s_\B: Q^+ \rightarrow \Dist(B) \in \SetStrat{\Aconc}{\B}$ for Player $\B$ such that $\prob{\Aconc}{\stratMod{\s_\A}}{\s_\B}[{ \rho }] = 1$. Consider the state strategy $\s_\B' \in \SetStrat{\Aconc}{\B}$ such that, for all $i \geq 0$, we have $\prob{\Aconc}{\stratMod{\s_\A}}{\s_\B}(\rho_{\leq i}) > 0$, which exists by definition of $\OutStrat{\Aconc}^{\s_\A}$. Let $i \geq 0$. Since the strategy $\s_\A$ is deterministic, we have:
	\begin{displaymath}
		\deltaDistrib^{\stratMod{\s_\A},\s'_\B}(\rho_{\leq i})[\rho_{i+1}] = \sum_{b \in \Supp(\s'_\B(\pi_{\leq i}))} \s'_\B(\pi)[b] \cdot \deltaDistrib(\head(\pi),\stratMod{\s_\A}(\rho_{\leq i}),b)[\rho_{i+1}] > 0
	\end{displaymath}
	Since the game $\Aconc$ is deterministic, for all $b \in \Supp(\s'_\B(\pi_{\leq i}))$, we have $\deltaDistrib(\head(\pi),\stratMod{\s_\A}(\rho_{\leq i}),b)[\rho_{i+1}] \in \{ 0,1 \}$. It follows that there exists $b_i \in \Supp(\s'_\B(\pi_{\leq i}))$ such that $\deltaDistrib(\head(\pi),\stratMod{\s_\A}(\rho_{\leq i}),b_i)[\rho_{i+1}] = 1$. 
	
	Then, for all $i \geq 0$, we set $\s_\B(\rho_i) = b_i$ with $\s_\B(\rho_i)$ Dirac. The strategy $\s_\B$ is defined arbitrarily on other sequence of states. With this definition, for all $i \geq 0$, we have $\deltaDistrib^{\stratMod{\s_\A},\s_\B}(\rho_{\leq i})[\rho_{i+1}] = 1$. It follows that, for all $i \geq 0$, we have $\prob{\Aconc}{\stratMod{\s_\A}}{\s_\B}(\rho_{\leq i}) = 1$. 
	Furthermore, we have $\cap_{n \in \mathbb{N}} \cyl(\rho_{\leq n}) = \{ \rho \}$. By the continuity of probability measure, we have: $$\prob{\Aconc}{\stratMod{\s_\A}}{\s_\B}[\{ \rho \}] = \prob{\Aconc}{\stratMod{\s_\A}}{\s_\B}[\cap_{n \in \mathbb{N}} \cyl(\rho_{\leq n})] = \lim_{n \rightarrow \infty} \prob{\Aconc}{\stratMod{\s_\A}}{\s_\B}[\cyl(\rho_{\leq n})] = \lim_{n \rightarrow \infty} \prob{\Aconc}{\stratMod{\s_\A}}{\s_\B}(\rho_{\leq n}) = 1$$
	
	For all paths $\rho \in \OutStrat{\Aconc}^{\s_\A}$ compatible with the strategy $\s_\A$, we denote by $\s_\B^\rho \in \SetStrat{\Aconc}{\B}$ a state strategy for Player $\B$ ensuring $\prob{\Aconc}{\stratMod{\s_\A}}{\s_\B}[\{ \rho \}] = 1$.
	
	Now, assume that the strategy $\s_\A$ is winning for Player $\A$. Let $\s_\B \in \SetStrat{\Aconc}{\B}$ be a strategy for Player $\B$. For all $\rho \notin \OutStrat{\Aconc}^{\s_\A} \subseteq U_W$, let us denote by $\rho_{\leq i_\rho} \in Q^+$ a prefix  of $\rho$ such that $\prob{\Aconc}{\stratMod{\s_\A}}{\s_\B}[\cyl(\rho_{\leq i_\rho})] = \prob{\Aconc}{\stratMod{\s_\A}}{\s_\B}[\rho_{\leq i_\rho}] = 0$. In that case, we have: 
	\begin{displaymath}
		\prob{\Aconc}{\stratMod{\s_\A}}{\s_\B}[Q^\omega \setminus U_W] \leq \prob{\Aconc}{\stratMod{\s_\A}}{\s_\B}[\bigcup_{\rho \notin \OutStrat{\Aconc}^{\s_\A}} \cyl(\rho_{\leq i_\rho})] \leq \sum_{\rho \notin \OutStrat{\Aconc}^{\s_\A}} \prob{\Aconc}{\stratMod{\s_\A}}{\s_\B}[\cyl(\rho_{\leq i_\rho})] = 0
	\end{displaymath}
	That is, $\prob{\Aconc}{\stratMod{\s_\A}}{\s_\B}[U_W] = 1$. As this holds for all strategy $\s_\B \in \SetStrat{\Aconc}{\B}$, it follows that $\val{\Aconc}{\s_\A}[W] = 1$ and $\val{\Aconc}{\s_\A}[W] > 0$.
	
	Conversely, let us denote by $\mathbf{1}_{U_W}: Q^\omega \rightarrow \{ 0,1 \}$ the function ensuring, for all $\rho \in Q^\omega$, $\mathbf{1}_{U_W}(\rho) = 1$ if and only if $\rho \in U_W$. Then, we have:
	\begin{align*}
	\val{\Aconc}{\s_\A}[W] > 0 & \Rightarrow \forall \s_\B \in \SetStrat{\Aconc}{\B},\; \prob{\Aconc}{\stratMod{\s_\A}}{\s_\B}[U_W] > 0 \\
	& \Rightarrow \forall \rho \in \OutStrat{\Aconc}^{\s_\A},\; \prob{\Aconc}{\stratMod{\s_\A}}{\s_\B^\rho}[U_W] > 0 \\
	& \Leftrightarrow \forall \rho \in \OutStrat{\Aconc}^{\s_\A},\; \prob{\Aconc}{\stratMod{\s_\A}}{\s_\B^\rho}[{\rho} \cap U_W] > 0 \\
	& \Leftrightarrow \forall \rho \in \OutStrat{\Aconc}^{\s_\A},\; \mathbf{1}_{U_W}(\rho) > 0\\
	& \Leftrightarrow \forall \rho \in \OutStrat{\Aconc}^{\s_\A},\; \rho \in U_W \\
	& \Leftrightarrow \text{ the strategy }\s_\A\text{ is winning for player }\A
	\end{align*}
\end{proof}

\subsection{Proof of Theorem~\ref{thm:turn_based_borel_determined}}
\label{proof:theorem3}
\begin{theorem}[Martin~\cite{martin1975borel,martin1985purely,martin1998determinacy}, Mentioned~\ref{ref:martin_borel}]
	Consider a deterministic turn-based graph arena $\Aconc
	$. For all Borel winning set $W \subseteq \Borel(\colSet)$, 
	the game $\Games{\Aconc}{W}$ is determined.
	\label{thm:turn_based_borel_determined}
\end{theorem}
\begin{proof}
	Consider a Borel winning set $W \subseteq \colSet^\omega$ in a deterministic turn-based arena $\Aconc = \AdetTurn$. The corresponding subset of infinite paths that is winning for Player $\A$ is the set $U_W = \inv{\extendFunc{\colFunc}}[W]$ is Borel as the preimage of a Borel set by the continuous function $\extendFunc{\colFunc}: Q \times  \starom{(Q \times Q)} \rightarrow \starom{\colSet}$. Then, the result by Martin~\cite{martin1975borel,martin1985purely,martin1998determinacy} gives that the game $\langle \Aconc,W \rangle$ is determined for strategies of the type $\sigma_\A: Q^+ \rightarrow A$ for Player $\A$ and $\sigma_\B: Q^+ \rightarrow B$ for Player $\B$. In our case, the strategies can be seen as such functions $\s$ that, in addition, ensures that if $\tr{\colFunc}(\pi) = \tr{\colFunc}(\pi')$ for two paths $\pi,\pi' \in Q \times (Q \times Q)^*$, then $\stratMod{\s}(\pi) = \stratMod{\s}(\pi')$. Therefore, we want to apply Corollary 4 from~\cite{DBLP:conf/cie/Roux20} which is possible since the winner of the game only depends on the sequence of colors seen, regardless of the states visited. 
\end{proof}

\subsection{Proof of Theorem~\ref{coro:conc_borel_determined}}
\label{proof:coro62}
\begin{proof}
	Consider a deterministic concurrent graph game $\Games{\Aconc}{W}$ (recall that $W \in \Borel(\colSet)$). Note that its sequential version, the turn-based game $\Games{\ATurnConc}{\Seq(W)}$ is deterministic. Hence, by Theorem~\ref{thm:turn_based_borel_determined}, it is determined, since
	, the winning set $\Seq(W)$ is Borel. Hence, by Theorem~\ref{thm:main}, the concurrent game $\Games{\Aconc}{W}$ is determined. 
	Conversely, consider a Borel set
	$\emptyset \subsetneq W \subsetneq \colSet^\omega$ and a non-determined game form $\formNF$. There exists a subset of outcomes $\subVal$
	such that the win/lose game $\gameNF = \langle \formNF,\subVal \rangle$ is not	determined. Furthermore, there exists some $\rho^1 \in W$ and some $\rho^2 \in	\colSet^\omega \setminus W$. We build the concurrent game where the local interaction in $q_0$ corresponds to $\formNF$ and two parts of the game are reachable from there: one is reached if the outcome of the game form would be
	in $\subVal$, and the other otherwise. Then, in the first part, we have a 
	unique infinite sequences of edges whose sequence of colors corresponds to  
	$\rho_1$ whereas the second part is similar, but it induces the sequence 
	$\rho_2$
	. The game constructed in this manner is not determined, since the local
	interaction in $q_0$ is not.
\end{proof}

\subsection{Complements on the determinacy of deterministic locally determined concurrent game with Borel objective}
\label{details:rmq64}

As soon as a concurrent game uses a local interaction that is not determined, it may also not enjoy determinacy. In turn, we obtain the following equivalence:
\begin{theorem}[Mentioned~\ref{ref:equivalence_borel}]
	Consider a set of game forms $\mathcal{I}$, and a set of colors
	$\colSet$ with $|\colSet| > 1$. Then the following are equivalent:
	\begin{itemize}
		\item[1.] $\mathcal{I} \subseteq \allDet$; 
		\item[2.] For all Borel set $W \in \Borel(\colSet)$, all concurrent 
		games $\langle \Aconc,W \rangle$ built on $\mathcal{I}$ are determined;
		\item[2'.] There exists a Borel set $\emptyset \subsetneq W \subsetneq 
		\colSet^\omega$ such that all concurrent games $\langle \Aconc,W \rangle$  built on $\mathcal{I}$ are determined. 
	\end{itemize}
\end{theorem}
This dis very similar to the proof of Theorem~\ref{coro:conc_borel_determined}. 
\begin{proof}
	Implication $(1) \Rightarrow (2)$ comes from
	Corollary~\ref{coro:conc_borel_determined}, implication $(2)
	\Rightarrow (2')$ is straightforward since $|\colSet| > 1$. Consider now implication $(2') \Rightarrow (1)$. Assume that $\mathcal{I} \not \subseteq
	\allDet$. Let $\formNF$ be a game form in $\mathcal{I}$ that is not
	determined. That is, there exists a subset of outcomes $\subVal$
	such that the win/lose game $\gameNF = \langle \formNF,\subVal \rangle$ is not
	determined. 
	Assume towards a contradiction that there is a Borel set
	$\emptyset \subsetneq W \subsetneq \colSet^\omega$ such that all
	concurrent games built on $\mathcal{I}$ are determined for
	$W$. There exists some $\rho^1 \in W$ and some $\rho^2 \in
	\colSet^\omega \setminus W$
	. We build the concurrent game where the local interaction in $q_0$
	corresponds to $\formNF$ and two parts of the game are reachable
	from there: one is reached if the outcome of the game form would be
	in $\subVal$, and the other otherwise. Then, in the first part, we have a 
	unique infinite sequences of edges whose sequence of colors corresponds to  
	$\rho_1$ whereas the second part is similar, but it induces the sequence 
	$\rho_2$
	. The game constructed in this manner is not determined, since the local
	interaction in $q_0$ is not. Hence the contradiction.
\end{proof}

In the proof of the above theorem, we established that as soon as the local interaction in the initial state is not determined, we can exhibit a winning condition for which the whole game is not determined. We can strengthen that result with the following: as soon as a state whose local interaction is not determined is reachable (in a specific sense) from the initial state, then we can exhibit a winning condition for which the whole game is not determined. This highlights how unsafe (i.e. breaking global determinacy) non-determined local interactions are.

We consider the notion of strong reachability. Informally, we say that a state $q$ is strongly reachable from a state $q'$ if, from $q'$ until $q$ is reached, either of the player has an action to ensure getting closer to $q$. Formally:
\begin{definition}[Strongly reachable]
	Let $\Aconc = \AdetConc$ be a deterministic concurrent arena. Consider a state $q \in Q$. We define $R_0^q = \{ q \}$ and, for all $i \geq 0$, $R_{i+1}^q = R_i^q \cup (P_{i+1}^q(\A) \cup P_{i+1}^q(\B))$ with $P_{i+1}^q(p) = \{q' \in Q \mid \reachStrat{p}(\formNF_{q'},R_{i}^q) \neq \emptyset \}$ for both players $p \in \{ \A,\B \}$. We define $R^q = \cup_{i \in \mathbb{N}} R_i^q$. Then, we say that the state $q \in Q$ is \emph{strongly reachable} from a state $q' \in Q$ if $q' \in R^q$.
\end{definition}
A state $q' \in P_{i+1}^q(p)$ is said to be at \emph{level} $i+1$ for Player $p$. Note that we could have $P_{i+1}^q(\A) \cap P_{i+1}^q(\B) \neq \emptyset$. In addition, the definition above could be extended to ordinals, but we consider this one to facilitate the explanations.

For a state $q \in Q$, we define the function $p^q$ that associates to a state $q' \neq q \in Q$ its level and player if $q$ is strongly reachable from $q'$ as follows:
\[   
p^q(q')  = 
\begin{cases}
c^k_\A &\quad\text{if } k = \min_{i \geq 1} \{ q' \in P_i^q(\A) \} \text{ if it exists; }\\
c^k_\B &\quad\text{otherwise, if } k = \min_{i \geq 1} \{ q' \in P_i^q(\B) \} \text{ if it exists; }\\
c_n&\quad\text{otherwise, if q is not strongly reachable from q'.}\\ 
\end{cases}
\]

Now, we consider an arena with a state strongly reachable from the initial state whose local interaction is not determined
. We want to exhibit a winning condition such that the corresponding game is not determined. However, that winning condition is expressed on colors, and the coloring function of the arena may not be as accurate as it may need to be (for instance, if the coloring function is constant) for us to properly express that winning condition. Hence, we consider the notion of \emph{colorless arena}, that is a concurrent arena $\Aconc = \Acolorless$ without coloring set and coloring function. Then, we have the following proposition. 
\begin{proposition}
	Consider a colorless deterministic concurrent arena $\Acolorless$ and assume that a state $q \in Q$ is strongly reachable from $q_0$ and that the local interaction $\formNF_q$ is not determined. Then, there exists a set of colors $\colSet$, a coloring function $\colFunc: Q \times Q \rightarrow \colSet$, and a winning condition $W \subseteq (\colSet)^\omega$ that is an open set such that the game $\Games{\Aconc}{W}$ is not determined, with $\Aconc = \AdetConc$.
\end{proposition}
\begin{proof}
	Consider the colorless deterministic concurrent arena $\Acolorless$ and a state $q \in Q$ such that the local interaction $\formNF_q$ is not determined. We define the set of colors $\colSet = \{ c_{\wa},c_{\wb},c_n \} \cup \{ c_\A^i,c_\B^i \mid i \geq 1 \}$. Since the local interaction $\formNF_q$ is not determined, we can consider the winning set  $\subVal \subseteq Q$ for which there is no winning strategy (for either of the players) in $\formNF_q$. We define the coloring function $\colFunc: Q \times Q \rightarrow \colSet$ by: for all $q' \in Q$, $\colFunc(q,q') = c_\wa$ if $q' \in \subVal$, and $\colFunc(q,q') = c_\wb$ otherwise. Furthermore, for $q'' \neq q \in Q$, $\colFunc(q'',q') = p^q(q'')$.
	
	Let us now define the set of prefixes that are bad for either of the player. Let $p \in \{ \A,\B \}$ denote a player and $\bar{p}$ the other. Then, the set of prefixes of colors that are bad for Player $p$ are the ones where she could have ensured getting closer to $q$ but it did not happen: $\mathsf{Bad}_p = \{ \rho \cdot c^{i+1}_p \cdot c \mid \rho \in (\colSet)^*,\; i \geq 0,\; c \notin \{ c_\wa,c_\wb,c^{k}_p,c^{k}_{\bar{p}} \mid 1 \leq k \leq i \} \} \subseteq (\colSet)^{\omega}$. We can now define the winning set $W \subseteq (\colSet)^\omega$ for Player $\A$:
	$$W = \bigcup_{\rho \in \mathsf{Bad}_\B \setminus \mathsf{Bad}_\A} \cyl(\rho) \cup \bigcup_{\rho \in (\colSet)^* \setminus \mathsf{Bad}_\A \wedge c_\wb \notin \rho} \cyl(\rho \cdot c_\wa)$$
	
	We have that $W$ is open. Furthermore, for $\Aconc = \AdetConc$, the game $\Games{\Aconc}{W}$ is not determined. Indeed, consider a strategy $\s_\B: (\colSet)^* \times Q \rightarrow B$ for Player $\B$ and let us show that it is not winning for Player $\B$ (it is simpler to show for Player $\A$, as we do not need the hypothesis that $q$ is strongly reachable from $q_0$). We define the strategy $\s_\A: (\colSet)^* \times Q \rightarrow B$ in the following way, for all $\rho,q' \in (\colSet)^* \times Q$:
	\[   
	\s_\A(\rho,q')  = 
	\begin{cases}
	\min_{<_\A} \reachStrat{\A}(\formNF_{q'},R_{i}^q) & 
\text{if } p^q(q') = c^{i+1}_\A\\
	\min_{<_\A} \inv{\deltaDirac}(q',\s_\B(\rho,q'))[Q \setminus R_{i}^q] &
\text{if } p^q(q') = c^{i+1}_\B \\
	 & \quad\wedge \s_\B(\rho,q') \notin \reachStrat{\B}(\formNF_{q'},R_{i}^q) \\
	 \min_{<_\A} \inv{\deltaDirac}(q,\s_\B(\rho,q))[\subVal] &
\text{if } q' = q \\
	 \text{arbitrary}& \text{otherwise}\\
	\end{cases}
	\]
	
	First, note that, in the third case, the strategy $\s_\A$ is well defined since there is no winning strategy for Player $\B$ in the local interaction $\formNF_q$ for the winning set $\subVal$, and therefore $\s_\B(\rho,q') \notin \reachStrat{\B}(\formNF_{q},Q \setminus \subVal)$. Now, consider the state sequence $\gamma = \outCome{\Aconc}{\s_\A,\s_\B} \in Q^\omega$ induced by the pair of strategies $(\s_\A,\s_\B)$. The strategy $\s_\A$ ensures that if a state $q' \in Q$ is such that $p^q(q') = c^{i+1}_\A$ 
	for some $i \geq 0$, then the next state seen is in $R_{i}^q$ (in particular, it may be $q$, in which case $c_\wa$ or $c_\wb$ will be seen afterwards). Therefore, it implies that no prefix of $\extendFunc{\colFunc}(\gamma) \in (\colSet)^\omega$ is in $\mathsf{Bad}_\A$. 
	
	First, assume that the state $q$ is not reached. Then, since $q$ is strongly reachable from $q_0$, it implies that $q_0 \in R_i^q$ for some $i \geq 0$. Hence, it follows that the level of states has not strictly decreased. Thus, a prefix of $\extendFunc{\colFunc}(\gamma)$ is in $\mathsf{Bad}_\B$. Therefore, $\extendFunc{\colFunc}(\gamma) \in W$. 
	
	Second, assume that the state $q$ is seen. Consider the first time it appears in $\gamma \in Q^\omega$. The next state seen in $\gamma$ is in $\subVal$, ensuring that the color $c_\wa$ is seen, while the color $c_\wb$ has not been seen yet. Therefore, we have $\extendFunc{\colFunc}(\gamma) \in W$. In any case, we have $\outCome{\Aconc}{\s_\A,\s_\B} \in U_W = \inv{\extendFunc{\colFunc}}[W]$, i.e. the strategy $\s_\B$ is not winning for Player $\B$.
\end{proof}

\subsection{Definitions from \cite{DBLP:conf/concur/Bouyer0ORV20}}
\label{definition:monotony_selectivity}
We recall here the condition for the existence of finite memory strategy in turn-based games established in \cite{DBLP:conf/concur/Bouyer0ORV20}. In the following, we will be considering a deterministic concurrent arena $\Aconc = \AdetConc$.

First, in \cite{DBLP:conf/concur/Bouyer0ORV20}, the authors do not consider a winning set $W \subseteq \colSet^\omega$ but rather a preference relation $\preceq: \colSet^\omega \times \colSet^\omega$ for Player $\A$ and the antagonistic preference $\preceq^{-1}$ for Player $\B$. Hence, we need to translate a winning set into a 
preference relation. For $W \subseteq \colSet^\omega$, we consider the preference relation $\preceq_W: \colSet^\omega \times \colSet^\omega$ defined by $\rho \prec_W \rho'$ for all $\rho \not \in W$ and $\rho' \in W$.

Furthermore, the authors do not consider winning strategies, but optimal ones. However, it has to be noted that a winning strategy with a winning set $W \subseteq \colSet^\omega$ is an optimal strategy with preference $\preceq_W$ whereas the converse does not hold: if a player has no winning strategy, every one of his strategy is optimal, whereas if a player has a winning strategy, optimal and winning strategies coincide.

Let us now focus more closely on the condition stated in 
\cite{DBLP:conf/concur/Bouyer0ORV20} for the existence of optimal finite-memory 
strategies. Let us first recall a few definitions. Consider a language $L \subseteq \colSet^*$. The language $[L] := 
\{ \rho \in \colSet^\omega \mid \forall n \in \Nat,\; \exists \pi \in L,\; 
\rho_{\leq n} \sqsubset \pi \}$ refers to the set of infinite words whose 
prefixes are also the prefixes of a word in $L$. Furthermore, the notation 
$\rec(\colSet)$ refers to the regular languages on a finite subset of 
$\colSet$. In addition, for a preference $\preceq \subseteq \colSet^\omega 
\times \colSet^\omega$, and two languages $L,L' \subseteq \colSet^\omega$, $L 
\preceq L'$ refers to $\forall \rho \in L,\; \exists \rho' \in L',\; \rho 
\preceq \rho'$ and $L \prec L'$ refers to $\exists \rho' \in L',\; \forall \rho 
\in L,\; \rho \prec \rho'$. Note that $L \prec L' \Leftrightarrow \lnot (L' 
\preceq L)$. Finally, for a memory skeleton $\langle M,m_{init},\mu \rangle$ on 
$\colSet$, and two states of the memory $m,m' \in M$, we denote 
$L^\mathcal{M}_{m,m'} := \{ \rho \in \colSet^* \mid 
\extendFunc{\mu}(m,\rho) = m' \}$. Let us consider the 
definitions of $\mathcal{M}$-monotony and $\mathcal{M}$-selectivity from 
\cite{DBLP:conf/concur/Bouyer0ORV20}:

\begin{definition}
	Let $\mathcal{M} = \langle M,m_{init},\mu \rangle$ be a memory skeleton. A preferences $\preceq \subseteq \colSet^\omega \times \colSet^\omega$ is $\mathcal{M}$-monotone if, for all $m \in M$ and $L_1,L_2 \in \rec(\colSet)$: $(\exists \rho \in L^\mathcal{M}_{m_{init},m},\; [\rho \cdot L_1] \prec [\rho \cdot L_2]) \Rightarrow (\forall \rho' \in L^\mathcal{M}_{m_{init},m}, [\rho' \cdot L_1] \preceq [\rho' \cdot L_2])$.
\end{definition}

\begin{definition}
	Let $\mathcal{M} = \langle M,m_{init},\mu \rangle$ be a memory skeleton. A preference $\preceq \subseteq \colSet^\omega \times \colSet^\omega$ is $\mathcal{M}$-selective if, for all $\rho \in \colSet^*,\; m = \extendFunc{\mu}(m_{init},\rho) \in M$, for all $L_1,L_2 \in \rec(\colSet)$ such that $L_1,L_2 \subseteq L^\mathcal{M}_{m,m}$, for all $L_3 \in \rec(\colSet)$, $[\rho \cdot (L_1 \cup L_2)^* \cdot L_3] \preceq [\rho \cdot L_1^*] \cup [\rho \cdot L_2^*] \cup [\rho \cdot L_3]$.
\end{definition}

Then the authors proved the following result (Theorem 9 in \cite{DBLP:conf/concur/Bouyer0ORV20}, and Theorem~\ref{thm:turn_based_memory_winning} in the main part of this paper):
\begin{theorem}
	Let $\preceq \subseteq \colSet^\omega \times \colSet^\omega$ be a preference relation and $\mathcal{M}$ be a memory skeleton. Then, both players have  finite-memory strategies implemented with memory skeleton $\mathcal{M}$ in all finite deterministic turn-based games if and only if $\preceq$ and $\preceq^{-1}$ are $\mathcal{M}$-monotone and $\mathcal{M}$-selective.
	\label{thm:turn_equiv_memory_winning}
\end{theorem}

Note that when $M$ is a singleton, this characterization coincide with the earlier one proved in \cite{DBLP:conf/concur/GimbertZ05}.

\subsection{Theorem of finite-memory determinacy}
\begin{theorem}[Theorem 9 in \cite{DBLP:conf/concur/Bouyer0ORV20}][Mentioned~\ref{ref:finite_mem}]
	Let $\mathcal{M}$ be a memory skeleton and $W \subseteq \colSet^\omega$
	. 
	The two following assertions are equivalent:
	\begin{enumerate}
		\item every finite deterministic turn-based game with $W$ as winning set is
		determined with winning strategies for both players that can be
		found among strategies implemented with memory skeleton
		$\mathcal{M}$;
		\item $W$ and $\colSet^\omega \setminus W$ are
		$\mathcal{M}$-monotone and $\mathcal{M}$-selective. 
	\end{enumerate}
	\label{thm:turn_based_memory_winning}
\end{theorem}

\subsection{Theorem stating that monotony and selectivty are preserved by sequentialization}
\label{proof:theorem5}
\begin{theorem}[Mentioned~\ref{ref:thm_mono_sequen}]
	Let $\langle \Aconc,W \rangle$ be a deterministic concurrent game on the concurrent arena $\Aconc = \AdetConc$, let $\langle \ATurnConc,\Seq(W) \rangle$ be its sequential version, and let $\mathcal{M}$ be a memory skeleton on $\colSet$. Then, $W$ is $\mathcal{M}$-monotone and $\mathcal{M}$-selective if and only if $\Seq(W)$ is $\Seq(\mathcal{M})$-monotone and $\Seq(\mathcal{M})$-selective.
	\label{thm:winning_condition_monotony_selectivity}
\end{theorem}

Before proving this theorem, we state and prove a few simpler lemmas. Recall 
that $\Seq(W) = \inv{\projec{\colSetSeq}{\colSet}}[W]$ where $\projec{\colSetSeq}{\colSet}: \starom{\colSetSeq} \rightarrow \starom{\colSet}$ is the projection function from $\colSetSeq$ to $\colSet$.
\begin{lemma}
	For all $L \subseteq \starom{\colSet}$, we have $L = \projec{\colSet_\Aconc}{\colSet}[\inv{\projec{\colSet_\Aconc}{\colSet}}[L]]$.
	\label{lem:projection_equality_surjectivity}
\end{lemma}
\begin{proof}
	This is straightforward, by surjectivity of the projection 
	$\projec{\colSet_\Aconc}{\colSet}: \starom{\colSetSeq} \rightarrow 
	\starom{\colSet}$.
\end{proof}

\begin{lemma}
	For all $L_1,L_2 \subseteq \starom{\colSetSeq}$, we have $\projec{\colSetSeq}{\colSet}[L_1 \cup L_2] = \projec{\colSetSeq}{\colSet}[L_1] \; \cup \; \projec{\colSetSeq}{\colSet}[L_2]$. Furthermore, if $L_1 \subseteq \colSetSeq^*$, $\projec{\colSetSeq}{\colSet}[L_1 \cdot L_2] = \projec{\colSetSeq}{\colSet}[L_1] \; \cdot \; \projec{\colSetSeq}{\colSet}[L_2]$, and $\projec{\colSetSeq}{\colSet}[L_1^*] = (\projec{\colSetSeq}{\colSet}[L_1])^*$.
	\label{lem:algebraic_properties_projection}
\end{lemma}
\begin{proof}
	This is straightforward, by definition of the function $\projec{\colSet_\Aconc}{\colSet}$.
\end{proof}

\begin{lemma}
	For $L \subseteq \colSetSeq^*$, we have $L \in \rec(\colSetSeq) \Rightarrow \projec{\colSet_\Aconc}{\colSet}[L] \in \rec(\colSet)$. For $L \subseteq \colSet^*$, we have $L \in \rec(\colSet) \Leftrightarrow \projec{\colSet_\Aconc}{\colSet}^{-1}[L] \in \rec(\colSetSeq)$ 
	\label{lem:rationnal_stability_projection}
\end{lemma}
\begin{proof}
	Assume that $L \in \rec(\colSetSeq)$. Then, we can modify a finite automaton recognizing $L$ with a finite alphabet of $\colSetSeq$ by replacing every transition where $k_\Aconc$ appears by an $\epsilon$-transition. Then, the obtained automaton is finite and recognizes the language $\projec{\colSet_\Aconc}{\colSet}[L]$.
	
	Furthermore, a finite automaton recognizing the language $L$ with a finite alphabet on $\colSet$ can be modified by adding self loops labeled by $k_\Aconc$ on every state. The obtained automaton is finite and recognizes the language $\projec{\colSet_\Aconc}{\colSet}^{-1}[L]$. Finally, if $\projec{\colSet_\Aconc}{\colSet}^{-1}[L] \in \rec(\colSetSeq)$, then, with Lemma~\ref{lem:projection_equality_surjectivity},  $\projec{\colSet_\Aconc}{\colSet}[\projec{\colSet_\Aconc}{\colSet}^{-1}[L]] = L \in \rec(\colSet)$.
\end{proof}

\begin{lemma}
	For all $m,m' \in M$, we have: $L^\mathcal{M}_{m,m'} = \projec{\colSet_\Aconc}{\colSet}[L^{\Seq(\mathcal{M})}_{m,m'}]$.
	\label{lem:language_equality_projection}
\end{lemma}
\begin{proof}
	Let $\rho \in \colSetSeq^*$. By Proposition~\ref{prop:projection_memory}, we have
	$\extendFunc{\Seq(\mu)}(m,\rho) = 
	\extendFunc{\mu}(m,\projec{\colSetSeq}{\colSet}(\rho))$
	. Hence, $\rho \in L^{\Seq(\mathcal{M})}_{m,m'} \Leftrightarrow \projec{\colSet_\Aconc}{\colSet}(\rho) \in L^\mathcal{M}_{m,m'}$.
\end{proof}

\begin{lemma}
	For all $N_1,N_2 \subseteq \colSet^\omega_\Aconc$, we have $N_1 \prec_{\Seq(W)} N_2 \Leftrightarrow \colSet^\omega \cap \projec{\colSet_\Aconc}{\colSet}[N_1] \prec_{W} \colSet^\omega \cap \projec{\colSet_\Aconc}{\colSet}[N_2]$.
	\label{lem:inequality_stability_projection}
\end{lemma}
\begin{proof}
	By definition of the preference relation $\prec_{\Seq(W)}$, we have:
	\begin{align*}
		N_1 \prec_{\Seq(W)} N_2 & \Leftrightarrow N_1 \cap \Seq(W) = \emptyset \wedge N_2 \cap \Seq(W) \neq \emptyset & \text{ by definition of }\prec_{\Seq(W)}\\ & \Leftrightarrow N_1 \cap \projec{\colSetSeq}{\colSet}^{-1}[W] = \emptyset \wedge N_2 \cap \projec{\colSetSeq}{\colSet}^{-1}[W] \neq \emptyset & \text{ by definition of }\Seq(W)\\ & \Leftrightarrow (\colSet^\omega \cap \projec{\colSetSeq}{\colSet}[N_1]) \cap W = \emptyset \wedge (\colSet^\omega \cap \projec{\colSetSeq}{\colSet}[N_2]) \cap W \neq \emptyset & \text{ by definition of }\projec{\colSetSeq}{\colSet}\\ & \Leftrightarrow (\colSet^\omega \cap \projec{\colSetSeq}{\colSet}[N_1]) \prec_{W} (\colSet^\omega \cap \projec{\colSetSeq}{\colSet}[N_2]) & \text{ by definition of }\prec_{W}
	\end{align*}
\end{proof}

\begin{lemma}
	For all $L \in \rec(\colSetSeq)$, we have $\colSet^\omega \cap \projec{\colSet_\Aconc}{\colSet}[\; [L] \;] = [\projec{\colSet_\Aconc}{\colSet}[L]]$.
	\label{lem:projection_infinite_language}
\end{lemma}
\begin{proof}
	Let $L \subseteq \colSet^{*}_\Aconc$. Let us prove $\colSet^\omega \cap \projec{\colSet_\Aconc}{\colSet}[\; [L] \;] \subseteq [\projec{\colSet_\Aconc}{\colSet}[L]]$. Let $\rho = \projec{\colSet_\Aconc}{\colSet}(\rho') \in \colSet^\omega \cap \projec{\colSet_\Aconc}{\colSet}[\; [L] \;]$ with $\rho' \in [L]$. Let $\varphi: \Nat \rightarrow \Nat$ such that $\rho_{\leq n} = \projec{\colSet_\Aconc}{\colSet}(\rho'_{\leq \varphi(n)})$ for all $n \geq 0$. Then, for all $n \geq 0$, there exists $l^n \in L$ such that $\rho'_{\leq n} \sqsubseteq l^n$. It follows that, for all $n \geq 0$, we have $\rho_{\leq n} = \projec{\colSet_\Aconc}{\colSet}(\rho'_{\varphi(n)}) \sqsubseteq \projec{\colSet_\Aconc}{\colSet}(l^{\varphi(n)}) \in \projec{\colSet_\Aconc}{\colSet}[L]$. Hence, $\rho \in [\; \projec{\colSet_\Aconc}{\colSet}[L] \; ]$.
	
	Let us now prove that $\projec{\colSet_\Aconc}{\colSet}[\; [L] \;] 
	\supseteq [\projec{\colSet_\Aconc}{\colSet}[L]]$. Since $L \in 
	\rec(\colSetSeq)$, there exists a finite automaton $\mathcal{A}$ 
	recognizing $L$. Let us denote by $i$ the number of states of 
	$\mathcal{A}$. Let $\rho \in [\projec{\colSet_\Aconc}{\colSet}[L]]$. Then, 
	for all $n \in \Nat$, there exists $u^n \in 
	\projec{\colSet_\Aconc}{\colSet}[L]$ such that $\rho_{\leq n} \sqsubseteq 
	u^n = \projec{\colSet_\Aconc}{\colSet}(v^n)$ for some $v^n \in L$ that we 
	choose so that there is no sequence in $(k_\Aconc)^*$ longer than $i$ 
	appearing in $v^n$. For all $n \geq 0$, let $l^n \sqsubseteq v^n$ be a 
	prefix of $v_n$ such that $\rho_{\leq n} = 
	\projec{\colSet_\Aconc}{\colSet}(l^n)$. In turn, with our choice for $v^n$, 
	we have that for all $k \leq |l^n|-1$, $|l^n_{\leq k}| \leq i \times 
	(|\projec{\colSetSeq}{\colSet}(l^n_{\leq k})| + 1)$. Then, the set of 
	prefixes $\{ l^n \mid n \geq 0 \}$ is infinite, with finitely many colors, 
	since $L \in \rec(\colSetSeq)$. Hence, by Koenig's lemma, we have $[ \{ l^n 
	\mid n \geq 0 \} ] \neq \emptyset$. Let $l \in [ \{ l^n \mid n \geq 0 \} 
	]$. Then, for all $n \geq 0$, we have $l_{\leq n} \sqsubseteq l^k 
	\sqsubseteq v^k \in L$ for some $k \in \Nat$. Hence, $l \in [L]$. Let us 
	now show that $\projec{\colSetSeq}{\colSet}(l) \in \colSet^\omega$. To do 
	so, we prove that $|\projec{\colSetSeq}{\colSet}(l_{\leq n})| 
	\xrightarrow[n \to \infty]{} \infty$. Let $N \geq 0$. We set $n_0 = (N + 1) 
	\cdot i \in \Nat$. Let $n \geq n_0$. There exists $n' \in \Nat$ such that 
	$l_{\leq n} \sqsubseteq l^{n'}$. Hence, we have 
	$|\projec{\colSetSeq}{\colSet}(l_{\leq n})| = 
	|\projec{\colSetSeq}{\colSet}(l^{n'}_{\leq n})| \geq |l^{n'}_{\leq n}|/i - 
	1 = (n+1)/i - 1 \geq N$ since $n \geq (N+1) \cdot i$. It follows that 
	$\projec{\colSetSeq}{\colSet}(l) \in \colSet^\omega$. Then, for all $n \in 
	\Nat$, we have $\projec{\colSet_\Aconc}{\colSet}(l)_{\leq n} \sqsubseteq 
	\projec{\colSet_\Aconc}{\colSet}(l^k_{\leq n'}) \sqsubseteq \rho_{\leq k} 
	\sqsubseteq \rho$ for some $k \in \Nat$ and $n' \leq |l^k|-1$. It follows 
	that $\rho = \projec{\colSet_\Aconc}{\colSet}(l)$.
\end{proof}

We can now proceed to the proof of Theorem~\ref{thm:winning_condition_monotony_selectivity}.
\begin{proof}
	We prove the equivalence of monotony. The case of selectivity is similar. In the following, $\projec{\colSet_\Aconc}{\colSet}$ will be abbreviated $\proj$. Let $m \in M$, $L_1,L_2 \in \rec(\colSetSeq)$, $L'_1,L'_2 \in \rec(\colSet)$, and $\rho \in L^{\Seq(\mathcal{M})}_{m_{init},m}$. Assume that $L'_1 = \proj(L_1)$ and $L'_2 = \proj(L_2)$. In that case:
	\begin{align*}
	[\rho \cdot L_1] \prec_{\Seq(W)} [\rho \cdot L_2] & \Leftrightarrow \colSet^\omega \cap \proj[ \; [\rho \cdot L_1] \; ] \prec_{W} \colSet^\omega \cap \proj[ \; [\rho \cdot L_2] \; ] &\text{ by Lemma~\ref{lem:inequality_stability_projection}} \\
	& \Leftrightarrow [ \; \proj[\rho \cdot L_1] \; ] \prec_{W} [ \; \proj[\rho \cdot L_2] \; ] & \text{by Lemma~\ref{lem:projection_infinite_language}}\\
	& \Leftrightarrow [ \; \proj(\rho) \cdot L'_1 \; ] \prec_{W} [ \; \proj(\rho) \cdot L'_2 \; ] & \text{by Lemma~\ref{lem:algebraic_properties_projection}}
	\end{align*}
	The two following equivalences are then given by Lemma~\ref{lem:language_equality_projection}:
	\begin{align*}
	(\exists \rho \in L^{\Seq(\mathcal{M})}_{m_{init},m},\; & [\rho \cdot L_1]  \prec_{\Seq(W)} [\rho \cdot L_2]) \\
	& \Leftrightarrow (\exists \rho' \in L^\mathcal{M}_{m_{init},m},\; [\rho' \cdot L'_1] \prec_W [\rho' \cdot L'_2])
	\end{align*}
	
	Similarly, with a universal quantifier:
	\begin{align*}
	(\forall \rho \in L^{\Seq(\mathcal{M})}_{m_{init},m},\; & [\rho \cdot L_1] \preceq_{\Seq(W)} [\rho \cdot L_2]) \\
	&\Leftrightarrow (\forall \rho' \in L^\mathcal{M}_{m_{init},m},\; [\rho' \cdot L'_1] \preceq_W [\rho' \cdot L'_2])
	\end{align*}
	By Lemma~\ref{lem:rationnal_stability_projection}, this holds for two arbitrary sets $L_1,L_2 \in \rec(\colSetSeq)$ and $\proj(L_1),\proj(L_2) \in \rec(\colSet)$. It also does with two arbitrary sets $L'_1,L'_2 \in \rec(\colSet)$ and $\proj^{-1}(L'_1),\proj^{-1}(L'_2) \in \rec(\colSetSeq)$, by Lemma~\ref{lem:projection_equality_surjectivity}. It follows that $\preceq_W$ is $\mathcal{M}$-monotone if and only if $\preceq_{\Seq(W)}$ is $\Seq(\mathcal{M})$-monotone.	
\end{proof}

\subsection{Proof of Corollary~\ref{thm:conc_equiv_memory_winning}}
\label{proof:thm66}
\begin{proof}
	If either $W$ or $\colSet^\omega \setminus W$ is not $\mathcal{M}$-monotone 
	and $\mathcal{M}$-selective, then 
	Theorem~\ref{thm:turn_based_memory_winning} 
	gives us an example of a finite turn-based game, which is a special case of 
	finite locally determined concurrent games, where the players do not have 
	winning strategies implemented with memory skeleton $\mathcal{M}$. Now, 
	assume that both $W$ and $\colSet^\omega \setminus W$ are 
	$\mathcal{M}$-monotone and $\mathcal{M}$-selective for some $\mathcal{M} = 
	\langle M,\minit,\mu \rangle$. It 
	follows, by Theorem~\ref{thm:winning_condition_monotony_selectivity}, that 
	$\Seq(W)$ and $(\colSetSeq)^\omega \setminus \Seq(W)$ are 
	$\Seq(\mathcal{M})$-monotone and $\Seq(\mathcal{M})$-selective. Furthermore, since the concurrent game $\Aconc$ is finite, its sequential version $\ATurnConc$ also is, since $A$ is finite. Thus, by Theorem~\ref{thm:turn_based_memory_winning}, 
	the determinisitc turn-based game $\Games{\ATurnConc}{\Seq(W)}$ is determined and there exists a winning strategy implemented with the memory skeleton $\Seq(\mathcal{M})$. Then, with Theorem~\ref{thm:all_prop} 
	, we have that there exists a winning strategy $\s$ in $\Games{\Aconc}{W}$ implemented with the memory skeleton $\Par(\Seq(\mathcal{M}))$. Finally, note that for all $m \in M$ and $k \in \colSet$, we have $\Par(\Seq(\mu))(m,k) =	\Seq(\mu)(\Seq(\mu)(m,k_\Aconc),k) = \Seq(\mu)(m,k) = \mu(m,k)$. That is, $\Par(\Seq(\mu)) = \mu$. In fact, the strategy $\s$ can be implemented with the memory skeleton $\mathcal{M}$.
\end{proof}

\subsection{Proof of Theorem~\ref{thm:memory_conc}}
\label{proof:thm_memory_conc}
Overall, we obtain in locally determined concurrent games the same
equivalence as the authors of \cite{DBLP:conf/concur/Bouyer0ORV20} proved in turn-based games: every finite game is determined for $W \subseteq \colSet^\omega$ with memory skeleton $\mathcal{M}$ if and only if $W$ and $\colSet^\omega \setminus W$ are $\mathcal{M}$-monotone and $\mathcal{M}$-selective. In addition, we have that the local determinacy assumption is somehow a necessary condition in the sense that as soon as concurrent games are built on at least one local interactions that is not determined, the aforementioned equivalence does not hold anymore. More specifically, we have the following theorem with nested equivalence:

\begin{theorem}[Mentioned~\ref{ref:equiv_memory}]
	Consider a set of game forms $\mathcal{I}$, and a set of colors $\colSet$ with $|\colSet| > 1$ and, for $W \subseteq \colSet^\omega$, set $E(W)$ as the equivalence:
	\begin{itemize}
		\item $W$ and $\colSet^\omega \setminus W$ are 
		$\mathcal{M}$-monotone and $\mathcal{M}$-selective;
		\item every finite concurrent game built on $\mathcal{I}$ is 
		determined for $W$ and there exists a winning strategy implemented 
		with memory skeleton $\mathcal{M}$;
	\end{itemize}
	Then the following assertions are equivalent:
	\begin{itemize}
		\item[1.] $\mathcal{I} \subseteq \allDet$; 
		\item[2.] For all memory skeleton $\mathcal{M}$ and $W 
		\subseteq \colSet^\omega$, $E(W)$ holds;
		\item[2'.] There exists a memory skeleton $\mathcal{M}$ such that for 
		all $W \subseteq \colSet^\omega$, $E(W)$ holds. 
	\end{itemize}
	\label{thm:memory_conc}
\end{theorem}
\begin{proof}
	Implication $(1) \Rightarrow (2)$ comes from
	Theorem~\ref{thm:conc_equiv_memory_winning}, implication $(2)
	\Rightarrow (2')$ is straightforward. As for implication $(2')
	\Rightarrow (1)$, consider a memory skeleton $\mathcal{M}$ and a
	reachability objective $W$, with $\colSet^\omega \setminus W$ a
	safety objective. Then, both of these kinds of objectives can be
	solved with positional strategies. Since positional strategies can be 
	implemented by any memory skeleton, in particular it can be implemented with 
	$\mathcal{M}$.
	Then, by Theorem~\ref{thm:turn_based_memory_winning}, it follows
	that $W$ and $\colSet^\omega \setminus W$ are $\mathcal{M}$-monotone
	and $\mathcal{M}$-selective. Now, if $\mathcal{I} \subsetneq
	\allDet$, a finite concurrent game with an initial interaction that
	is not determined that can either reach the desired color or not
	(which is possible since $|\colSet| > 1$) is built on $\mathcal{I}$
	and is not determined for $W$ which contradicts $(2')$. It follows
	that $(2') \Rightarrow (1)$.
\end{proof}

\subsection{Theorem of positional determinacy of parity stochastic games}
\begin{theorem}[\cite{DBLP:conf/soda/ChatterjeeJH04,DBLP:conf/fossacs/Zielonka04}][Mentioned~\ref{ref:thm_parity}]
	Consider a (stochastic) finite turn-based graph arena $\Aconc = \AdetTurn$ with $\colSet = \llbracket m,n \rrbracket$ for some $m,n \in \mathbb{N}$. For all	
	parity objective $W$ on $\colSet$, 
	the game $\Games{\Aconc}{W}$ is positionnaly determined.
	\label{thm:positional_determinacy_parity_turn}
\end{theorem}

\subsection{Equivalence with parity objective}
Like in the deterministic case, we obtain an equivalence:
\begin{theorem}[Mentioned~\ref{ref:parity_equiv}]
	Consider a set of game forms $\mathcal{I}$, and a set of colors
	$\colSet$ with $|\colSet| > 1$. The following are equivalent:
	\begin{itemize}
		\item[1.] $\mathcal{I} \subseteq \allDet$; 
		\item[2.] For all parity objective $W$ on $\colSet$, all finite concurrent games $\langle \Aconc,W \rangle$ with the set of actions $A$ finite built on $\mathcal{I}$ are positionnaly determined;
		\item[2'.] There exists a parity objective $W$ on $\colSet$ such that all finite concurrent games $\langle \Aconc,W \rangle$ with the set of actions $A$ finite  built on $\mathcal{I}$ are positionnaly determined. 
	\end{itemize}
	\label{thm:parity_positionaly_determined_conc}
\end{theorem}
The proof is analogous to the deterministic case.

\subsection{Proof of Theorem~\ref{coro:positional_determinacy_parity_conc}}
\label{proof:coro71}
\begin{proof}
	Consider such a finite arena $\Aconc = \AcoloredConc$. First, its sequential version $\Seq(\Aconc)$ is also finite since $A$ is finite. Furthermore, if we set $k_\Aconc = m-1$, the winning objective $\Seq(W)$ is a parity objective on $\colSet_\Aconc = \llbracket m-1,n \rrbracket$, therefore the turn-based game $\Games{\Seq(\Aconc)}{\Seq(W)}$ is positionnaly determined by Theorem~\ref{thm:positional_determinacy_parity_turn}. We can then conclude with Theorem~\ref{thm:all_prop}.
\end{proof}

\subsection{Theorem regarding tail objective}
\begin{theorem}[\cite{DBLP:conf/soda/GimbertH10}][Mentioned~\ref{ref:thm_tail}]
	Consider a (stochastic) finite turn-based arena $\Aconc$. For all winning set $W \subseteq \colSet^\omega$ that is a tail objective, 
	if the value of the game is extreme, that is if $\val{\Aconc}{\A}[W] = 1$ or $\val{\Aconc}{\B}[W] = 0$, then the game is determined.
	\label{thm:prefix_independent_determinacy_parity_turn}
\end{theorem}

\subsection{Proof of Corollary~\ref{coro:prefix_independent_determinacy_parity_conc}}
\label{proof:coro77}
\begin{proof}
	Consider such a finite arena $\Aconc = \AcoloredConc$. Its sequential version $\Seq(\Aconc)$ is also finite since $A$ is finite, straightforwardly, the sequential version $\Seq(W)$ of a winning condition $W \in \Borel(\colSet)$ that is a tail objective is also a tail objective. We conclude with Theorem~\ref{thm:all_prop}.
\end{proof}

\subsection{Nash equilibrium}
\label{appendix:nash}
Another application of the results we have established in this paper
lies in the existence of Nash equilibria. In
\cite{DBLP:journals/mlq/Roux14}, the author proved that the existence
of winning strategies could be transferred to the existence of Nash
equilibria. That is, consider a deterministic concurrent graph arena $\Aconc
= \AdetConc$ and a set $\mathcal{D} \subseteq 2^{\colSet^\omega}$
of subsets of outcomes that forms a partition of the set
$\colSet^\omega$: $\colSet^\omega = \uplus_{D \in \mathcal{D}} D$
along with two preference relations ${\prec_\A} \subseteq \mathcal{D}
\times \mathcal{D}$ and ${\prec_\B} \subseteq \mathcal{D} \times
\mathcal{D}$ respectively for Player $\A$ and Player
$\B$\footnote{Note that this a generalization of win/lose game where
	we consider the special case where $\mathcal{D} = \{
	W,\colSet^\omega \setminus W \}$ for some $W \subseteq
	\colSet^\omega$ and $\prec_\A^{-1} = \prec_\B$.}. We have a concurrent 
	graph game $\langle \Aconc,\mathcal{D},\prec_A,\prec_B \rangle$. 
	Informally, a Nash equilibrium roughly consists of a pair of strategies 
	where no player has an interest (with regard to his preference relation) in 
	unilaterally
changing his strategy. That is, if for a pair of strategies
$(\s_\A,\s_\B) \in \SetSt{\Aconc}{\A} \times
\SetSt{\Aconc}{\B}$, we denote by $\delta(\s_\A,\s_\B) \in \mathcal{D}$ the 
subset of $\colSet^\omega$ in which lies the outcome $\outCome{\Aconc}{\s_\A,\s_\B}$ induced by the strategies $\s_\A$ and $\s_\B$, we have the following definition:
\begin{definition}
	A pair of strategies $(\s_\A,\s_\B) \in \SetSt{\Aconc}{\A} \times 
	\SetSt{\Aconc}{\B}$ is a \emph{Nash equilibrium} in $\langle 
	\Aconc,\mathcal{D},\prec_A,\prec_B \rangle$ if, for all strategy 
	$\s \in \SetSt{\Aconc}{\A}$ (resp. $\SetSt{\Aconc}{\B}$), we have 
	that $\delta(\s_\A,\s_\B) \prec_\A \delta(\s,\s_\B)$ (resp. 
	$\delta(\s_\A,\s_\B) \prec_\B \delta(\s_\A,\s)$) does not hold.
\end{definition}

Then, the following theorem is directly derived from Theorem 1.7 in 
\cite{DBLP:journals/mlq/Roux14}.
\begin{theorem}
	Consider a deterministic concurrent arena $\Aconc = \AdetConc$, a partition 
	$\mathcal{D} \subseteq 2^{\colSet^\omega}$ of $\colSet^\omega$, two 
	preference relations $\prec_\A,\prec_\B \subseteq \mathcal{D} \times 
	\mathcal{D}$, and two subsets of strategies $R_\A \subseteq 
	\SetSt{\Aconc}{\A}$ and $R_\B \subseteq \SetSt{\Aconc}{\B}$
	and assume the following:
	\begin{itemize}
		\item for all $\mathcal{H} \subseteq \mathcal{D}$, the win/lose game 
		$\langle \mathcal{C},W \rangle$ is determined for $W = \uplus_{H \in 
			\mathcal{H}} H$ with winning strategies in $R_\A$ or $R_\B$;
		\item the preference $\prec_\A$ (resp. $\prec_\B$) has finite height, 
		that is there exists $n \in \mathbb{N}$ such that there is no $D_1,D_2, 
		\ldots, D_n \in \mathcal{D}$ such that $D_1 \prec_\A D_2 \prec_\A 
		\ldots \prec_\A D_n$.
	\end{itemize}
	Then, there exists a Nash equilibrium in $\mathcal{C}$ with the preference 
	relations $\prec_\A$ and $\prec_\B$ that can be found in $R_\A \times R_\B$.
	\label{thm:nash_existence}
\end{theorem}

In \cite{DBLP:journals/mlq/Roux14} three applications of the main 
result were provided: the existence of Nash equilibrium in Borel games, of 
positional Nash equilibrium in priority games (multi-outcome games on parity 
arenas), and of finite-memory Nash equilibrium in generalized Muller games. 
Here we could transfer the three results to concurrent games with determined 
local interactions, but we only give an idea of what theorem regarding the
results for Borel games and Nash equilibrium in priority games could look like 
in the form of remarks as the proofs are not written yet.

\paragraph*{Application: Nash equilibrium with Borel outcomes} This result 
could be obtained by combining Theorem~\ref{thm:nash_existence} and 
Corollary~\ref{coro:conc_borel_determined}:
\begin{remark}
	Consider a deterministic concurrent arena $\Aconc = \AdetConc$, a partition 
	$\mathcal{D} \subseteq 2^{\colSet^\omega}$ of $\colSet^\omega$, two 
	preference relations $\prec_\A,\prec_\B \subseteq \mathcal{D} \times 
	\mathcal{D}$ and assume the following:
	\begin{itemize}
		\item $\mathcal{D}$ is countable and all sets of $\mathcal{D}$ are 
		Borel;
		\item the preference $\prec_\A$ (resp. $\prec_\B$) has finite height.
	\end{itemize}
	Then, there exists a Nash equilibrium in $\mathcal{C}$ with the preference 
	relations $\prec_\A$ and $\prec_\B$.
	\label{thm:nash_existence_conc_borel}
\end{remark}

\paragraph*{Application: Nash equilibrium in a priority game}
This result may be derived from Theorem~\ref{thm:nash_existence}, 
Theorem~\ref{thm:conc_equiv_memory_winning}, and the fact that turn-based 
priority games are positionally determined, we obtain the following theorem:
\begin{remark}
	Consider a finite deterministic concurrent arena $\Aconc = \AdetConc$ where $\colSet 
	= \llbracket 0,n \rrbracket$ for some $n \in \mathbb{N}$. For all $k \in 
	\colSet$, we define $D_k = \{ \rho \in \colSet^\omega \mid \max (\inf \rho) 
	= k \}$
	. Consider $\mathcal{D} = \{ D_k \mid k \in \colSet \}$ and two 
	preference relations $\prec_\A,\prec_\B \subseteq \mathcal{D} \times 
	\mathcal{D}$ and assume that  the preferences $\prec_\A$ and $\prec_\B$ are 
	acyclic.
	
	Then, there exists a Nash equilibrium in $\mathcal{C}$ with the preference 
	relations $\prec_\A$ and $\prec_\B$ among positional strategies.
	\label{thm:nash_existence_conc_priority}
\end{remark}

